\documentclass[12pt, reqno, twoside, letterpaper]{amsart}
\usepackage[fixmeHIDE, 
            nixHIDE,
            nixnixHIDE,
            nixaltHIDE,
            commentsHIDE,
            jetsamHIDE,
            bibnixHIDE,
            arXivON]{sums-of-two-squares-in-short-intervals}

\title[Poisson spacings between sums of two squares]
      {Poisson distribution for gaps between sums of two
       squares and level spacings for toral point scatterers}

\author[T.\ Freiberg]{Tristan Freiberg}

\address{Department of Pure Mathematics, 
         University of Waterloo, 
         Waterloo ON, CANADA.}

\email{tfreiberg@uwaterloo.ca}

\author[P.\ Kurlberg]{P\"ar Kurlberg}

\address{Department of Mathematics, 
         KTH Royal Institute of Technology, 
         Stockholm, SWEDEN.}

\email{kurlberg@kth.se}

\author[L.\ Rosenzweig]{Lior Rosenzweig}

\address{Department of Mathematics, 
         ORT Braude College, 
         Karmiel, ISRAEL.}

\email{liorr@braude.ac.il}

\date{\today}

\begin{document}

%%%%%%%%%%%%%%%%%%%%%%%%%%%%%%%%%%%%%%%%%%%%%%%%%%%%%%%%%%%%%%%%%%
%%%%%%%%%%%%%%%%%%%%%%%%%%%% ABSTRACT %%%%%%%%%%%%%%%%%%%%%%%%%%%%
%%%%%%%%%%%%%%%%%%%%%%%%%%%%%%%%%%%%%%%%%%%%%%%%%%%%%%%%%%%%%%%%%%

\begin{abstract} 
We investigate the level spacing distribution for the quantum
spectrum of the square billiard.  
Extending work of Connors--Keating, and Smilansky, we formulate 
an analog of the Hardy--Littlewood prime $k$-tuple conjecture 
for sums of two squares, and show that it implies that the 
spectral gaps, after removing degeneracies and rescaling, are 
Poisson distributed.  
Consequently, by work of Rudnick and Uebersch\"ar, the level 
spacings of arithmetic toral point scatterers, in the weak 
coupling limit, are also Poisson distributed.
We also give numerical evidence for the conjecture and its 
implications.
\end{abstract}

\maketitle

%%%%%%%%%%%%%%%%%%%%%%%%%%%%%%%%%%%%%%%%%%%%%%%%%%%%%%%%%%%%%%%%%%
%%%%%%%%%%%%%%%%%%%%%%%%%%%% SECTION 01 %%%%%%%%%%%%%%%%%%%%%%%%%%
%%%%%%%%%%%%%%%%%%%%%%%%%%%%%%%%%%%%%%%%%%%%%%%%%%%%%%%%%%%%%%%%%%

\section{Introduction}
 \label{sec:intro}

According to the Berry--Tabor conjecture \cite{BT:77}, the energy 
levels for generic integrable systems should be Poisson 
distributed in the semiclassical limit.  
As noted by Connors and Keating \cite{CK:97}, the square billiard, 
though integrable, is not generic: due to spectral degeneracies, 
the level spacing distribution tends to a $\delta$-function at 
zero.
However, if we remove the degeneracies and rescale so that the 
mean spacing is unity, numerics indicate Poisson spacings.

\begin{figure}[ht]
  \centering
\includegraphics[width=7.5cm]{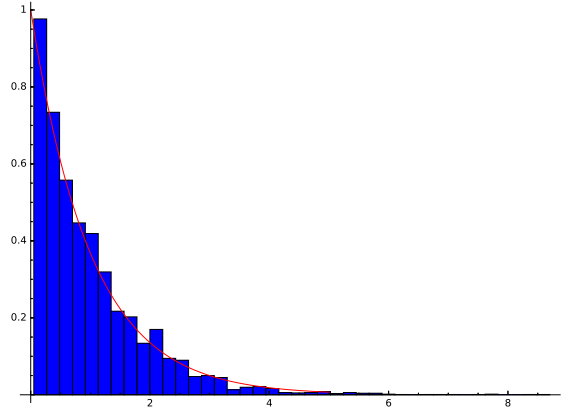}

\caption{%
  Rescaled gaps between consecutive energy levels in
  $[10^{99}, 10^{99} + 110000]$, after removing degeneracies. 
%Bins: $\protect\sage{bins}$.
  %
  The rescaled gaps have mean one; without rescaling
  the mean gap is $19.42\cdots$.
  Number of gaps: $5663$.  
  We also plot the density function (red in color
  printout) $P(x) = \e^{-x}$, consistent with Poisson spacings.
  }
\end{figure}

The energy levels of the square billiard, say with side length 
$2\pi$, are number theoretical in nature, and given by 
$a^2 + b^2$ for $a,b \in \ZZ$.
After removing degeneracies and rescaling, we are led to study 
the nearest neighbor spacing distribution 

\begin{equation}
 \label{eq:nnsd}
 \frac{1}{\speccount(x)}
  \#
   \bigg\{
    \sts_{n} \le x : \frac{\sts_{n + 1} - \sts_{n}}{x/\speccount(x)} < \lambda 
   \bigg\} 
\end{equation}
(as $x \to \infty$), where $\sts_n$ denotes the $n$th smallest 
element of the set  
\begin{equation}
 \label{eq:defSSN}
 \SS \defeq \{a^2 + b^2 : a,b \in \ZZ\}, 
  \quad 
   \text{and}
    \quad 
     \speccount(x)
      \defeq 
       \#\{\sts_{n} \le x : \sts_n \in \SS\}.
\end{equation}
(In our setting the leading order of the density of states is
asymptotically equal to $C/\sqrt{\log x}$ as $x \to \infty$ [cf.\   
\eqref{eq:sotsnt}], and hence the spacing distribution of the 
unfolded levels $\big(C\sts_n/\sqrt{\log \sts_n}\big)_{n \ge 1}$ 
has the same asymptotic distribution as the gaps in 
\eqref{eq:nnsd}.)

\begin{nixnix}
This is the probability that, given $\sts_n \le x$, there is some 
integer $b > \lambda x/\speccount(x)$ such that 
\[
\ind{\SS}(\sts_n + b)
 \prod_{j = 1}^{b - 1} 
  \big(1 - \ind{\SS}(\sts_n + j)\big)
   =
    1,
\]
where $\ind{\SS}$ is the indicator function of $\SS$.

If the positions of the levels $\sts_n$, $n \le x$, were 
uncorrelated, this probability would be 
\[
\sum_{b > \lambda x/\speccount(x)}
 \frac{\speccount(x)}{x}
  \bigg(1 - \frac{\speccount(x)}{x}\bigg)^{b - 1},
\]
which, viewed as a Riemann sum approximation to the integral 
$\int_{\lambda}^{\infty} \e^{-t} \dd{t}$, would suggest that 
the nearest neighbor spacing distribution does indeed follow 
a Poisson law, i.e.\ that 
\[
\frac{1}{\speccount(x)}
 \#\bigg\{\sts_{n} \le x : \frac{\sts_{n + 1} - \sts_{n}}{x/\speccount(x)} > \lambda \bigg\}
  \sim 
   \e^{-\lambda}
    \quad 
     (x \to \infty).
\]

If the positions of the levels $\sts_n$, $\sts_n \le x$, were 
uncorrelated, this would indeed follow a Poisson law.

However, being sums of two squares, the levels are not 
uncorrelated.

For instance, four consecutive integers cannot all be in $\SS$, 
since one of them is congruent to $3$ modulo $4$.

Nevertheless, taking into account correlations modulo all prime 
powers, we still (conjecturally) expect \eqref{eq:nnsd} to follow 
a Poisson law, consistent with numerics.
\end{nixnix}

Rather than studying the spacing distribution directly, we shall
proceed by investigating {\em unordered} $k$-tuples of elements in 
$\SS$.
Thus, given $k \ge 1$ and $\bh = \{h_1,\ldots,h_k\} \subseteq \ZZ$ 
with $\card \bh = k$, consider the correlation function
\begin{equation}
 \label{eq:defklevcor}
 R_k(\bh;x)
  \defeq 
   \frac{1}{x}
    \sum_{n \le x}
     \ind{\SS}(n + h_1)
      \cdots 
       \ind{\SS}(n + h_k),
\end{equation}
where $\ind{\SS}$ denotes the indicator function of $\SS$.
If $\bh = \{0\}$, this is the level density 
\begin{equation}
 \label{eq:deflevden}
 R_1(x) 
  \defeq 
   \frac{\speccount(x)}{x}.
\end{equation}
By a classical result of Landau \cite{LAN:08},
\begin{equation}
 \label{eq:sotsnt}
   R_1(x) 
    \sim 
     \frac{C}{\sqrt{\log x}}
      \quad 
    (x \to \infty),
\end{equation}
where $C > 0$ is an explicitly given constant (see 
\eqref{eq:defLanRamconst}).
To formulate an analog of (\ref{eq:sotsnt}) for $k > 1$ we need 
some further notation.  
Given a prime $p \not \equiv 1 \bmod 4$, define%
\begin{equation}
 \label{eq:delthp}
  \delta_{\bh}(p)
   \defeq 
    \lim_{\alpha \to \infty}
     \frac{
      \#\{0 \le a < p^{\alpha} : 
             \forall h \in \bh, 
              a + h \equiv \sots \bmod p^{\alpha}\}
          }
          {p^{\alpha}}.
\end{equation}
(That the limit exists is shown in Section \ref{sec:prelims}, cf.\  
Propositions \ref{prop:Sp3h} and \ref{prop:S2h}.)
Further, for $k \ge 1$ and a set 
$\bh = \{h_1,\ldots,h_k\} \subseteq \ZZ$ with $\card \bh = k$, we 
define the {\em singular series} for $\bh$ by  
\begin{equation}
 \label{eq:defsssP}
  \mathfrak{S}_{\bh}
   \defeq 
    \prod_{p \not\equiv 1 \bmod 4}
     \frac{\delta_{\bh}(p)}{\big(\delta_{\bz}(p)\big)^k},
\end{equation}
with $\delta_{\{0\}}(p)$ and $\delta_{\bh}(p)$ as in
\eqref{eq:delthp}.
We note that $\delta_{\{0\}}(p) > 0$ for all 
$p \not\equiv 1 \bmod 4$, and that the  product converges to a 
nonzero limit if $\delta_{\bh}(p) > 0$ for all 
$p \not \equiv 1 \bmod 4$ (cf.\ Proposition \ref{prop:sssc}).
If $\delta_{\bh}(p) = 0$ for some $p \not\equiv 1 \bmod 4$, we 
define $\mathfrak{S}_{\bh}$ to be zero;
%we  remark that
it is easy to see that
$R_k(\bh;x) = 0$ for all $x$ if $\mathfrak{S}_{\bh} = 0$.

We can now formulate an analog of the Hardy--Littlewood prime 
$k$-tuple conjecture.
\begin{conjecture}
 \label{con:sotsktups}
Fix $k \ge 1$, and a set $\bh = \{h_1,\ldots,h_k\} \subseteq \ZZ$ 
with $\card \bh = k$. 
If $\mathfrak{S}_{\bh} > 0$, then
\begin{equation}
 \label{eq:sotsktups}
  R_k(\bh;x)
    \sim 
     \mathfrak{S}_{\bh}
      \big(R_1(x)\big)^k 
       \quad      
     (x \to \infty).
\end{equation}
\end{conjecture}
\noindent 
Our main result, Theorem \ref{thm:main} below, is conditional on 
the hypothesis that \eqref{eq:sotsktups} holds on average.
To be precise, let $\cE_{\bh}(x)$ be defined by the relation 
\begin{equation}
 \label{eq:defEterm}
   R_k(\bh;x)
    \eqdef
     \big(\mathfrak{S}_{\bh} + \cE_{\bh}(x)\big)
      \big(R_1(x)\big)^k.
\end{equation}
Further, let $\Delta^k$ be the region in $\RR^k$ defined by  
\begin{equation}
 \label{eq:defsimplex}
  \Delta^k
   \defeq 
    \{(x_1,\ldots,x_k) \in \RR^k : 0 < x_1 < \cdots < x_k\},
\end{equation}
and, given $\sC \subseteq \Delta^k$ and $y \in \RR$, let $y\sC$ 
be the dilation of $\sC$ defined by  
\[
 y\sC \defeq \{(yx_1,\ldots,yx_k) : (x_1,\ldots,x_k) \in \sC\}.
\]
%
% Conjecture \ref{con:sotsktups} asserts that, for a {\em fixed} 
% set $\bh$ such that $\mathfrak{S}_{\bh} > 0$, we have 
% $|\cE_{\bh}(x)| \to 0$ as $x \to \infty$.
%
% Further, if $\mathfrak{S}_{\bh} = 0$, then $\cE_{\bh}(x) = 0$ for 
% all $x \ge 1$.
%
Our hypothesis is that the error term $\cE_{\bh}(x)$ is small when 
averaged over dilates of certain bounded convex subsets.

\begin{hypothesis}[$k,\sC,\bo$]
Fix an integer $k \ge 1$ and a bounded convex set 
$\sC \subseteq \Delta^k$.
Set $\bo \defeq \emptyset$ or set $\bo \defeq \{0\}$.
Let $x$ and $y$ be real parameters tending to infinity in such a 
way that $yR_1(x) \sim 1$.
There exists a function $\varepsilon(x)$, with 
$\varepsilon(x) \to 0$ as $x \to \infty$, such that for $x$ 
sufficiently large in terms of $k$ and $\sC$,
\begin{equation}
 \label{eq:hyp}
   \bigg|
    \sum_{(h_1,\ldots,h_k) \in y\sC \cap \, \ZZ^k}
     \cE_{\bo \cup \bh}(x)
   \bigg|
 \le 
  \varepsilon(x)
   \sum_{(h_1,\ldots,h_k) \in y\sC \cap \, \ZZ^k}
    \mathfrak{S}_{\bo \cup \bh},
\end{equation}
where $\bh = \{h_1,\ldots,h_k\}$ in both summands. 
\end{hypothesis}

Under the above hypothesis we find that the spacing distribution 
\eqref{eq:nnsd} is indeed Poissonian.
Moreover, the distribution of the number of points in intervals of 
size comparable to the mean spacing is consistent with that of a 
Poisson process.
(We remark that our hypothesis can be weakened slightly --- see 
Section \ref{sec:main}.)

\begin{theorem}
 \label{thm:main}
Let $x$ and $y$ be real parameters tending to infinity in such a 
way that $yR_1(x) \sim 1$.
Fix integers $m \ge 0$ and $r \ge 1$, and fix 
$\lambda,\lambda_1,\ldots,\lambda_r \in \RR^{+}$.
Assume that Hypothesis
\textup{(}$k,\sC,\{0\}$\textup{)}\textup{)}
\textup{(}respectively, 
Hypothesis
\textup{(}$k,\sC,\emptyset$\textup{)}
holds 
for all $k \ge 1$, and all bounded, convex sets 
$\sC \subseteq \Delta^k$.
Then \textup{(}a\textup{)} \textup{(}respectively, 
\textup{(}b\textup{)}\textup{)} holds.

\textup{(}a\textup{)} We have  
\begin{equation}
 \label{eq:thm:mainc}
  \frac{1}{\speccount(x)}
   \#\{\sts_n \le x : \forall j \le r, \sts_{n + j} - \sts_{n + j - 1} \le \lambda_j y\}
    \sim 
     \prod_{j = 1}^r 
      \int_0^{\lambda_j} \e^{-t} \dd{t} %(1 - \e^{-\lambda_j})
       \quad 
     (x \to \infty).
\end{equation}

\textup{(}b\textup{)} We have 
\begin{equation}
 \label{eq:thm:main}
  \frac{1}{x}
   \#\{n \le x : \speccount(n + \lambda y) - \speccount(n) = m\}
%    \sums[n \le x][{\# \SS \cap (n,n + \lambda y] = m}] 1
    \sim 
     \e^{-\lambda}\frac{\lambda^m}{m!}
       \quad     
     (x \to \infty).
\end{equation}

\end{theorem}
In \cite{RU:14}, Rudnick and Uebersch\"ar considered the spectrum 
of ``toral point scatterers'', namely the Laplace operator, 
perturbed by a delta potential, on two dimensional tori.  
They showed that the level spacings of the perturbed eigenvalues, 
in the weak coupling limit, have the same distribution as the 
level spacings of the unperturbed eigenvalues (after removing 
multiplicities).
An interesting consequence of Conjecture \ref{con:sotsktups} (or 
\eqref{eq:hyp}) is thus that the Berry--Tabor conjecture holds for 
toral point scatterers, in the weak coupling limit, for arithmetic 
tori of the form $\RR^2/\ZZ^2$.

We remark that Gallagher \cite{GAL:76} proved the analog of 
Theorem \ref{thm:main} (b) for primes.
Just as in his proof, a key technical result is that the singular 
series is of average order one, over certain geometric regions.

\begin{proposition}
 \label{prop:sssa}
Fix an integer $k \ge 1$, and a bounded convex set 
$\sC \subseteq \Delta^k$.
Set $\bo \defeq \emptyset$ or set $\bo \defeq \{0\}$.
As $y \to \infty$, we have  
\begin{align}
 \label{eq:sssa}
  \sum_{(h_1,\ldots,h_k) \in y\sC \cap \, \ZZ^k} 
   \mathfrak{S}_{\bo \cup \bh} 
   & =
    y^k \Big( \vol(\sC) + O\big(y^{-2/3 + o(1)}\big)\Big), 
\end{align}
where $\bh = \{h_1,\ldots,h_k\}$ in the summand, and $\vol$ 
stands for volume in $\RR^k$.
\end{proposition}

{\bfseries Acknowledgements.} 
We thank Z.\ Rudnick for stimulating discussions on the subject 
matter, and D.\ Koukoulopoulos for his comments on an early 
version of the paper.
T.\ F.\ was partially supported by a grant from the G\"oran 
Gustafsson Foundation for Research in Natural Sciences and 
Medicine.
P.\ K.\ and L.\ R.\  were partially supported by grants from the 
G\"oran Gustafsson Foundation for Research in Natural Sciences and 
Medicine, and the Swedish Research Council (621-2011-5498).
L.\ R. wishes to thank and acknowledge the Mathematics department at KTH, being his home institute during the period where most of the work on this paper was done.

%%%%%%%%%%%%%%%%%%%%%%%%%%%%%%%%%%%%%%%%%%%%%%%%%%%%%%%%%%%%%%%%%%
%%%%%%%%%%%%%%%%%%%%%%%%%%%% SECTION 02 %%%%%%%%%%%%%%%%%%%%%%%%%%
%%%%%%%%%%%%%%%%%%%%%%%%%%%%%%%%%%%%%%%%%%%%%%%%%%%%%%%%%%%%%%%%%%

\section{Discussion}
\label{sec:discussion}

Connors and Keating \cite{CK:97} determined the singular series 
for shifted pairs of sums of two squares and gave a probabilistic
derivation of Conjecture \ref{con:sotsktups} for $k = 2$, and 
found that it matched numerics quite well (to within $2\%$).
Smilansky \cite{SMI:13} then expressed the singular series for 
pairs as products of $p$-adic densities, and showed that its mean 
value (over short intervals of shifts) is consistent with a 
Poisson distribution, and that the same is true for sums of two 
squares, on assuming a uniform version of 
Conjecture \ref{con:sotsktups} for $k = 2$.  
He also determined the singular series for triples corresponding 
to the shifts $\bh = \{0,1,2\}$.

As already mentioned, the analog of Theorem \ref{thm:main} (b) for 
primes is due to Gallagher; in \cite{GAL:76} he showed that an 
appropriate form of the Hardy--Littlewood prime $k$-tuples 
conjecture implies the prime analog of \eqref{eq:thm:mainc}.
(That it implies the prime analog of \eqref{eq:thm:main} is 
mentioned in Hooley's survey article \cite[p.\ 137]{HOO:72}.)
To show that the singular series is one on average (i.e., the 
prime analog of Proposition \ref{prop:sssa}), Gallagher uses 
combinatorial identities for Stirling numbers of the second kind.
In \cite{KOW:11}, Kowalski developed an elegant probabilistic
framework for evaluating averages of singular series.  
Rather than using combinatorial identities, he showed that a 
certain duality between $k$-th moments of $m$-tuples and 
$m$-th moments of $k$-tuples holds 
(cf.\ \cite[Theorem 1]{KOW:11}).  
That the $k$-th moment of $1$-tuples equals one is essentially 
trivial; by duality he obtains the non-trivial consequence that 
first moments of $k$-tuples also equals one.
(Note that \eqref{eq:sssa} can be viewed as a first moment of 
$k$-tuples when $\bo = \emptyset$.)  

Our approach originates with techniques developed in
\cite{KR:99, KUR:00}, and further refined in \cite{GK:08, KUR:09}.
Loosely speaking, the singular series $\mathfrak{S}_{\bh}$ is 
expanded into local factors of the form $1 + \epsilon_{\bh}(p)$,
and thus
\[
 \mathfrak{S}_{\bh} = \prod_{p} (1 + \epsilon_{\bh}(p))
 =
  \sums[d \ge 1][\text{squarefree}] \epsilon_{\bh}(d),
\]
where $\epsilon_{\bh}(1) = 1$ and 
$\epsilon_{\bh}(d) \defeq \prod_{p|d} \epsilon_{\bh}(p)$.  
Hence
\[
 \sum_{\bh}
  \mathfrak{S}_{\bh} 
   =
    \sums[d \ge 1][\text{squarefree}]
     \sum_{\bh}
      \epsilon_{\bh}(d),
\]
and the main term is given by $d = 1$.  
For $d$ large, $|\epsilon_{\bh}(d)|$ can be shown to be small on 
average.
For $d$ small, we use that $\epsilon_{\bh}(d)$ (approximately) 
only depends on $\bh \bmod d$, together with complete cancellation 
when summing over the {\em full} set of residues modulo $d$, i.e.,
$\sum_{\bh \bmod d} \epsilon_{\bh}(d) = 0$.  
This follows, via the Chinese remainder theorem, from local
cancellations
$ 
\sum_{\bh \bmod p} \epsilon_{\bh}(p) = 0
$, 
which in turn can be deduced from the following easily verifiable
identity: given {\em any} subset $X_p \subseteq \ZZ/p\ZZ$, we have
(cf.\ Lemma \ref{lem:cancel} (b) and its proof for more details):
\[
 \sum_{(h_1, h_2, \ldots, h_k) \in (\ZZ/p\ZZ)^k} 
  \#\{ m \in \ZZ/p\ZZ : m + h_1, m + h_2, \ldots, m + h_k \in X_p\}
   =
    \big(\#X_p\big)^{k}.
\]

However, unlike the setup in \cite{KR:99,GK:08,KUR:09}, where the
local error terms $\epsilon_{\bh}(p)$ are determined by
$\bh \bmod p$, in the current setting the image of 
$\bh \bmod p^{\alpha}$, for any fixed $\alpha$, is not sufficient 
to determine $\epsilon_{\bh}(p)$.  
On the other hand, the function $\bh \to \epsilon_{\bh}(p)$ has 
nice $p$-adic regularity properties, allowing us to approximate
$\epsilon_{\bh}(p)$ by truncations 
$\epsilon_{\bh}(p^{\alpha})$ such that 
$\epsilon_{\bh}(p^{\alpha})$ only depends on 
$\bh \bmod p^{\alpha}$, and 
$
 \epsilon_{\bh}(p) - \epsilon_{\bh}(p^{\alpha}) 
  \ll 
   1/p^{\alpha - 1}
$ 
for all $\bh$.  
Apart from making the arguments more complicated, we also get
a weaker error term: if $\epsilon_{\bh}(p)$ only depended on
$\bh \bmod p$, we would get a relative error of size 
$y^{-1 + o(1)}$, rather than $y^{-2/3 + o(1)}$.
We also note that David, Koukoulopoulos and Smith \cite{DKS:15}, 
in studying statistics of elliptic curves, have developed quite 
general methods for finding asymptotics of weighted sums
$\sum_{\bh} w_{\bh} \mathfrak{S}_{\bh}$, provided that the local
factors have $p$-adic regularity properties similar to the ones
above. 
In fact, Proposition \ref{prop:sssa}, though with a weaker error 
term, can be deduced from \cite[Theorem 4.2]{DKS:15}.

We finally remark that the corresponding question in the function 
field setting is better understood --- Bary--Soroker and Fehm 
\cite{BS-F} recently showed that the sums of squares analog of the 
$k$-tuple conjecture holds in the large $q$-limit for the 
function field setting (e.g., replacing $\mathbb{Z}$ by 
$\mathbb{F}_{q}[T]$ and $\mathbb{Z}[i]$ by 
$\mathbb{F}_q[\sqrt{-T}]$).

\subsection{Evidence towards Conjecture \ref{con:sotsktups}.} 
 \label{sec:sotsHLktups}

We begin by formulating a qualitative version of 
Conjecture \ref{con:sotsktups}.
\begin{conjecture} 
 \label{con:qualktups}
Fix $k \ge 1$, and a set 
$\bh = \{h_1,\ldots,h_k\} \subseteq \ZZ$ with $\card \bh = k$. 
If $\mathfrak{S}_{\bh} > 0$, then there exist infinitely many 
integers $n$ such that $n + \bh \subseteq \SS$.
\end{conjecture}
\noindent
We  remark that whether or not $\mathfrak{S}_{\bh} > 0$ can be 
determined by a finite computation: this follows from 
Propositions \ref{prop:S2h} and \ref{prop:Sp3h}. 
%
%****************************************************************%
%************************* START DETAIL *************************%
%****************************************************************%
%
\begin{nixnix}
For recall that $\bh$ is $\SS$-admissible if and only if 
$\delta_{\bh}(p) > 0$ for all $p$.
(See Propositions \ref{prop:S2h} and \ref{prop:Sp3h} et seq.)
For $p \nmid \det(\bh)$ we have $\delta_{\bh}(p) > 0$ by 
Propositions \ref{prop:S2h} (c) and \ref{prop:Sp3h} (c).
For the finitely many $p$ dividing $\det(\bh)$, we have 
$\delta_{\bh}(p) > 0$ if and only if either 
$\card \bh_p \ne \emptyset$ or, in the case $p \equiv 3 \bmod 4$, 
$\V_{\bh}(p^{\alpha}) \ne \emptyset$ for 
$\alpha = 1 + \max_{i \ne j} \nu_p(h_i - h_j)$, in the case 
$p = 2$, $\T_{\bh}(2^{\alpha + 1}) \ne \emptyset$ for 
$\alpha = 2 + \max_{i \ne j} \nu_2(h_i - h_j)$.
(See Propositions \ref{prop:S2h} (b) and \ref{prop:Sp3h} (b).)
\end{nixnix}
%
%****************************************************************%
%************************** END DETAIL **************************%
%****************************************************************%
%
%
Examples of sets $\bh$ for which $\mathfrak{S}_{\bh} = 0$ are
$\{0,1,2,3\}$ and $\{0,1,2,4,5,8,16,21\}$: any translate of 
$\{0,1,2,3\}$ contains an integer congruent to $3$ modulo $4$, and 
hence $\delta_{\bh}(2) = 0$; any translate of 
$\{0,1,2,4,5,8,16,21\}$ contains an integer congruent to $3$ or 
$6$ modulo $9$, and hence $\delta_{\bh}(3) = 0$.

It is possible to show that $\mathfrak{S}_{\bh} > 0$ for {\em any} 
set $\bh$ containing at most three integers.
The question of whether, for any $h_1,h_2,h_3 \in \ZZ$, we have 
$n + \{h_1,h_2,h_3\} \subseteq \SS$ for infinitely many $n$, was 
apparently raised by Littlewood: Hooley \cite{HOO:73} showed, 
using the theory of ternary quadratic forms, that 
Conjecture \ref{con:qualktups} indeed holds for $k \le 3$.
The conjecture remains open for $k \ge 4$.

For fixed $k \ge 1$ and $\bh = \{h_1,\ldots,h_k\}$ with 
$\card \bh = k$, the upper bound 
\[
 \sum_{n \le x}
  \ind{\SS}(n + h_1)\cdots \ind{\SS}(n + h_k)
   \ll_k
    \frac{x}{(\log x)^{k/2}}
     \prods[p \equiv 3 \bmod 4][p \mid h_j - h_j][\text{some $i < j$}]
      \bigg(1 + \frac{k}{p}\bigg),
\]
can be deduced from Selberg's sieve (see \cite{SEL:77}), which is 
of the correct order of magnitude, according to 
Conjecture \ref{con:sotsktups}.
The special case $\bh = \{0,1\}$ is due to Rieger \cite{RIE:65}; 
the special case $\bh = \{0,1,2\}$ is due to Cochrane and Dressler 
\cite{CD:87}; the general case is due to Nowak \cite{NOW:05}.

Lower bounds are more subtle.  
For $k = 2$, Hooley \cite{HOO:74} and Indlekofer \cite{IND:74} 
showed that, for any nonzero integer $h$, 
\[
  \sum_{n \le x} \ind{\SS}(n)\ind{\SS}(n + h)
   \gg
    \frac{x}{\log x}
     \prods[p \mid h][p \equiv 3 \bmod 4]
      \bigg(1 + \frac{1}{p}\bigg),
\]
but we are not aware of any such bounds for $k \ge 3$.

We remark that Iwaniec deduced the asymptotic
$
 \sum_{n \le x} \ind{\SS}(n)\ind{\SS}(n + 1) \sim 3x/(8\log x)
$, 
as $x \to \infty$, from an analog of the Elliott--Halberstam 
conjecture for sums of two squares 
(cf.\ \cite[Corollary 2, (2.3)]{IWA:76}).
However, note that the leading term constant $3/8$ disagrees with 
the one due to Connors and Keating \cite{CK:97}, namely $1/2$.  
(We also obtain the constant $1/2$; see Figure \ref{fig:table1} 
below for a numerical comparison.)

\subsection{Numerical evidence}
\label{sec:numerical-evidence}
Using Propositions \ref{prop:S2h} (b), (c) and 
\ref{prop:Sp3h} (b), (c), we can give $\mathfrak{S}_{\bh}$ 
explicitly, as in the following examples.
Let us first record that the constant $C$ in \eqref{eq:sotsnt} is 
the Landau--Ramanujan constant, given by 
\begin{equation}
 \label{eq:defLanRamconst}
  C
   \defeq 
    \frac{1}{\sqrt{2}}
     \prod_{p \equiv 3 \bmod 4}
      \bigg(
       1 - \frac{1}{p^2}
      \bigg)^{-1/2}
       =
        0.764223\ldots.
\end{equation}
It is straightforward to verify that   
\begin{equation}
 \label{eq:sss01}
  \mathfrak{S}_{\{0,1\}}
   =
    \frac{1}{2C^2} 
     =
      0.856108\ldots.
\end{equation}
If \eqref{eq:sotsktups} holds with $\bh = \{0,1\}$ then, by 
\eqref{eq:sotsnt} and \eqref{eq:sss01}, 
\[
 \speccount(\{0,1\}; x)
  \defeq 
   \sum_{n \le x} \ind{\SS}(n)\ind{\SS}(n + 1)
    \sim 
     \frac{x}{2C^2}
      \big(R_1(x)\big)^2
       \sim 
        \frac{x}{2\log x}
         \quad 
       (x \to \infty).
\]
The agreement with numerics is quite good (to within $1 \%$).
\begin{figure}[ht]
\begin{tabular}{|r|r|r|r|}
\hline 
$x$ & {\small $\speccount(\{0,1\}; x)$ } & {\small $x\mathfrak{S}_{\{0,1\}}(R_1(x))^2$ } & Ratio \\
\hline
1000000000  &  25927011 &   25690391.1 & 1.00921  \\
2000000000  &  50042411 &   49603435.5 & 1.00885  \\
3000000000  &  73560246 &   72930222.0 & 1.00864  \\
4000000000  &  96705170 &   95891759.7 & 1.00848  \\
5000000000  & 119584162 &  118589346.3 & 1.00839  \\
6000000000  & 142253331 &  141080935.2 & 1.00831  \\
7000000000  & 164749254 &  163403937.1 & 1.00823  \\
8000000000  & 187100631 &  185584673.5 & 1.00817  \\
9000000000  & 209327440 &  207642640.3 & 1.00811  \\
\hline
\end{tabular}   
\caption{Observed data vs prediction for $\bh = \{0,1\}$. }
\label{fig:table1}
\end{figure}

As the simplest example with $k = 3$, we verify that 
\[
  \mathfrak{S}_{\{0,1,2\}}
   =
    \frac{A}{4C^2}, 
     \quad 
      A 
      \defeq \prod_{p \equiv 3 \bmod 4}
       \bigg(1 - \frac{2}{p(p - 1)}\bigg),
\]
so Conjecture \ref{con:sotsktups} implies that 
\[
  \speccount(\{0,1,2\}; x)
   \defeq 
    \sum_{n \le x} \ind{\SS}(n)\ind{\SS}(n + 1)\ind{\SS}(n + 2)
     \sim 
      \frac{Ax}{4C^2}
       \big(R_1(x)\big)^3
        \sim 
         \frac{ACx}{4(\log x)^{3/2}}
\]
as $x \to \infty$.
Here, the agreement between numerics and model is only 
to within $10 \%$.

\begin{figure}[ht]
\begin{tabular}{|r|r|r|r|}
\hline
$x$ & {\small $\speccount(\{0,1,2\}; x)$ } & {\small $x\mathfrak{S}_{\{0,1,2\}}(R_1(x))^3$ } & Ratio \\
\hline
1000000000  &  1490691 &   1362419.3 & 1.09415  \\
2000000000  &  2818128 &   2584683.5 & 1.09032  \\
3000000000  &  4093602 &   3762317.2 & 1.08805  \\
4000000000  &  5338091 &   4912433.3 & 1.08665  \\
5000000000  &  6560430 &   6042800.3 & 1.08566  \\
6000000000  &  7764604 &   7157833.6 & 1.08477  \\
7000000000  &  8954282 &   8260369.7 & 1.08400  \\
8000000000  & 10132295 &   9352396.2 & 1.08339  \\
9000000000  & 11299877 &  10435380.5 & 1.08284  \\
\hline
\end{tabular}  
\caption{Observed data vs prediction for $\bh=\{0,1,2\}$. }
\label{fig:table2}
\end{figure}

\begin{nixnix}
\section{Motivating the probabilistic model}
It is elementary to show (see Section \ref{sec:prelims}) that 
\begin{equation}
 \label{eq:defSp}
 \SS = \bigcap_{p \not\equiv 1 \bmod 4} S_p,
  \quad 
   \text{where}
    \quad  
 S_p 
  \defeq 
   \bigcap_{\alpha = 1}^{\infty}
    \{n \in \ZZ : n \equiv \sots \bmod p^{\alpha}\},
\end{equation}
the first intersection being over all primes  
$p \not\equiv 1 \bmod 4$, and $n \equiv \sots \bmod p^{\alpha}$ 
denoting that $n \equiv \sts \bmod p^{\alpha}$ for some 
$\sts \in \SS$. 
It is also elementary to show that, for all 
$p \not\equiv 1 \bmod 4$, the asymptotic density 
\[
 \lim_{x \to \infty}
  \frac{1}{x}
   \#\{n \le x : n + \bh \subseteq S_p\}, 
    \quad 
     n + \bh \defeq \{n + h : h \in \bh\},
\] 
exists, and is equal to 
\begin{equation}
% \label{eq:delthp}
  \delta_{\bh}(p)
   \defeq 
    \lim_{\alpha \to \infty}
     \frac{
      \#\{0 \le a < p^{\alpha} : 
             \forall h \in \bh, 
              a + h \equiv \sots \bmod p^{\alpha}\}
          }
          {p^{\alpha}},
\end{equation}
which exists, and is nonzero if $\bh = \{0\}$, by 
Propositions \ref{prop:S2h} (a) and \ref{prop:Sp3h} (a).
Thus, $R_k(\bh;x)$ is the probability that 
$n + \bh \subseteq S_p$ for all $p \not\equiv 1 \bmod 4$, and, 
regarding the probabilities of $n + \bh$ being contained in $S_p$ 
and $S_q$ as more or less independent for $p \ne q$, we therefore 
expect that 
\[
 \frac{R_k(\bh;x)}{\big(R_1(x)\big)^k}
  =
  \frac{R_k(\bh;x)}{\big(R_1(\bz;x)\big)^k}
   \sim 
    \prod_{p \not\equiv 1 \bmod 4}
     \frac{\delta_{\bh}(p)}{\big(\delta_{\bz}(p)\big)^k}
      \quad 
       (x \to \infty),
\]
provided that the product converges to a nonzero number.

\begin{definition}
 \label{def:Sss}
Fix $k \ge 1$, and a set 
$\bh = \{h_1,\ldots,h_k\} \subseteq \ZZ$ with $\card \bh = k$.
The {\em singular series} for $\bh$ is given by  
\begin{equation}
 \label{eq:defsssP}
  \mathfrak{S}_{\bh}
   \defeq 
    \prod_{p \not\equiv 1 \bmod 4}
     \frac{\delta_{\bh}(p)}{\big(\delta_{\bz}(p)\big)^k},
\end{equation}
with $\delta_{\{0\}}(p)$ and $\delta_{\bh}(p)$ as in 
\eqref{eq:delthp}.
\end{definition}

Proposition \ref{prop:Sp3h} (c) shows that, for sufficiently 
large primes $p \equiv 3 \bmod 4$, 
$\delta_{\bz}(p)^{-k}\delta_{\bh}(p) = 1 + O_k(1/p^2)$, and so 
the product in \eqref{eq:defsssP} converges to a nonzero number 
unless $\delta_{\bh}(p) = 0$ for some $p \not\equiv 1 \bmod 4$.
Proposition \ref{prop:sssc} shows that $\mathfrak{S}_{\bh} \ne 0$ 
if and only if $\bh$ is $\SS$-admissible.

\begin{definition}
 \label{def:Sadm}
Fix $k \ge 1$, and a set $\bh = \{h_1,\ldots,h_k\} \subseteq \ZZ$ 
with $\card \bh = k$.
We say that $\bh$ is {\em $\SS$-admissible} if, and only if, for 
all primes $p \not\equiv 1 \bmod 4$, there exists $n \in \ZZ$ such 
that $n + \bh \subseteq S_p$, where $S_p$ is as in 
\eqref{eq:defSp}.
\end{definition}

\noindent
In view of \eqref{eq:sotsnt}, for $\SS$-admissible $\bh$, 
\eqref{eq:sotsktups} is equivalent to 
\begin{equation}
 \label{eq:equivsotsktups}
 \sum_{n \le x}
  \ind{\SS}(n + h_1)\cdots \ind{\SS}(n + h_k)
%   \sim 
%    \mathfrak{S}_{\bh}
%     x
%      \big(R_1(x)\big)^k
       \sim 
        C^k\mathfrak{S}_{\bh}
        x(\log x)^{-k/2}
         \quad 
       (x \to \infty),
\end{equation}
and so we have the following qualitative version of 
Conjecture \ref{con:sotsktups}.

\begin{conjecture} 
 \label{con:qualktups-hide}
Fix $k \ge 1$, and a set 
$\bh = \{h_1,\ldots,h_k\} \subseteq \ZZ$ with $\card \bh = k$. 
If $\bh$ is $\SS$-admissible, then there exist infinitely many 
integers $n$ for which $n + \bh \subseteq \SS$.
\end{conjecture}

\noindent
It is plain that, if $\bh$ is not $\SS$-admissible, then there 
are no integers $n$ for which $n + \bh \subseteq \SS$.
\end{nixnix}

%%%%%%%%%%%%%%%%%%%%%%%%%%%%%%%%%%%%%%%%%%%%%%%%%%%%%%%%%%%%%%%%%%
%%%%%%%%%%%%%%%%%%%%%%%%%%%% SECTION 03 %%%%%%%%%%%%%%%%%%%%%%%%%%
%%%%%%%%%%%%%%%%%%%%%%%%%%%%%%%%%%%%%%%%%%%%%%%%%%%%%%%%%%%%%%%%%%

\section{Notation}
 \label{sec:notation}

We define the set of natural numbers as 
$\NN \defeq \{1,2,\ldots\}$.
The letter $p$ stands for a prime, $n$ for an integer.
We let $\sots\ $ stand for a generic element of $\SS$, possibly a 
different element each time.
Thus, for instance, $a + h\equiv \sots \, \bmod p^{\alpha}$ denotes 
that $a + h \equiv \sts \bmod p^{\alpha}$ for some $\sts \in \SS$.
We view $k$ as a fixed natural number, and $\bh$ as a nonempty, 
finite set of integers, with $\card \bh = k$ unless otherwise 
indicated.
We let $n + \bh \defeq \{n + h : h \in \bh\}$.
For $n \in \NN$, $\omega(n)$ denotes the number of distinct prime 
divisors of $n$, $\nu_p(n)$ the $p$-adic valuation of $n$.
(We also define $\nu_p(0) \defeq \infty$.) 
That $\nu_p(n) = \alpha$ may also be denoted by 
$p^{\alpha} \emid n$.
The radical of $n$ is $\rad(n) \defeq \prod_{p \mid n} p$, not to 
be confused with the squarefree part of $n$, viz.\  
$\sqfr(n) \defeq \prod_{p \emid n} p$.
By the least residue of an integer $a$ modulo $n$ we mean the 
integer $r$ such that $a \equiv r \bmod n$ and $0 \le r < n$.
When written in an exponent, $\alpha \bmod 2$ is to be interpreted 
as the least residue of $\alpha$ modulo $2$: for instance, 
$p^{\alpha \bmod 2} = 1$ if $\alpha$ is even.

We view $x$ as a real parameter tending to infinity.
Expressions of the form $A \sim B$ denote that $A/B \to 1$ as 
$x \to \infty$.
We also view $y$ as real parameter tending to infinity, 
typically in such a way that $y \sim x/\speccount(x)$.
We may assume that $x$ and $y$ are sufficiently large in terms of 
any fixed quantity.
Expressions of the form $A = O(B)$, $A \ll B$ and $B \gg A$ all 
denote that $|A| \le c|B|$, where $c$ is some positive constant, 
throughout the domain of the quantity $A$.
The constant $c$ is to be regarded as independent of any parameter 
unless indicated otherwise by subscripts, as in 
$A = O_{k}(B)$ ($c$ depends on $k$ only), $A \ll_{k,\lambda} B$ 
($c$ depends on $k$ and $\lambda$ only), etc.
By $o(1)$ we mean a quantity that tends to zero as $y \to \infty$.

%%%%%%%%%%%%%%%%%%%%%%%%%%%%%%%%%%%%%%%%%%%%%%%%%%%%%%%%%%%%%%%%%%
%%%%%%%%%%%%%%%%%%%%%%%%%% SECTION 04 %%%%%%%%%%%%%%%%%%%%%%%%%%%%
%%%%%%%%%%%%%%%%%%%%%%%%%%%%%%%%%%%%%%%%%%%%%%%%%%%%%%%%%%%%%%%%%%

\section{Deducing Theorem \ref{thm:main} from Proposition \ref{prop:sssa}}
 \label{sec:main}

Given $\vbi = (i_1,\ldots,i_r) \in \NN^r$ such that 
$i_1 + \cdots + i_r = k$, and 
$\vbl = (\lambda_1,\ldots,\lambda_r) \in \RR^r$, let 
\begin{equation}
 \label{eq:defThet}
 \Theta_{\vbi,\vbl}
  \defeq 
   \{(x_1,\ldots,x_k) \in \Delta^k : x_{i_1 + \cdots + i_j} - x_{i_1 + \cdots + i_{j - 1}} \le \lambda_j, j = 1,\ldots,r\},
\end{equation}
where for $j = 1$ we let $x_{i_1 + i_{j - 1}} = x_0 \defeq 0$. 
In the case where $r = 1$ and $\vbl = (\lambda)$, 
\begin{equation}
 \label{eq:defThetkl}
 \Theta_{\vbi,\vbl}
  = 
   \Theta_{k,\lambda}
    \defeq 
     \{(x_1,\ldots,x_k) \in \RR^k : 0 < x_1 < \cdots < x_k \le \lambda\}. 
\end{equation}
The following proof shows that Theorem \ref{thm:main} (a) and (b) 
hold under slightly weaker hypotheses than the ones stated:
for (a), it is enough to assume that 
Hypothesis \textup{(}$k,\Theta_{\vbi,\vbl},\emptyset$\textup{)}, 
where $\vbi = (i_1,\ldots,i_r)$ and 
$\vbl = (\lambda_1,\ldots,\lambda_r)$, holds for all $k \ge r$, 
and all $\vbi \in \NN^r$ satisfying $i_1 + \cdots + i_r = k$;
for (b), it is enough to assume that 
Hypothesis \textup{(}$k,\Theta_{k,\lambda},\emptyset$\textup{)} 
holds for all $k \ge 1$.

\begin{proof}[Deduction of Theorem \ref{thm:main}]
As this argument has appeared many times in the literature, we 
merely give an outline of it and provide references.
(a) 
To ease notation, we let $\vbi = (i_1,\ldots,i_r)$, 
$\vbh = (h_1,\ldots,h_k)$, $\bh = \{h_1,\ldots,h_k\}$, and 
\[
 \speccount(\{0\} \cup \bh; x)
  \defeq 
   \sum_{n \le x} 
    \ind{\SS}(n)\ind{\SS}(n + h_1)\cdots \ind{\SS}(n + h_k).
\]
% and, finally, 
% $
%  \gap{n} \defeq \sts_{n + 1} - \sts_n
% $.
%
Let $\ell \ge 0$ be an integer, arbitrarily large but fixed.
An inclusion-exclusion argument (see \cite{HOO:65iii}, 
\cite[Appendix A]{KR:99} or \cite[Key Lemma 2.4.12]{KS:99}) shows 
that 
\begin{align}
 \begin{split}
  \label{eq:bded1}
 & 
   \sum_{k = r}^{r + 2\ell + 1}
    (-1)^{k - r}
     \sum_{i_1 + \cdots + i_r = k}
      \hspace{5pt}
       \sum_{\vbh \, \in \, y\Theta_{\vbi,\vbl} \cap \, \ZZ^k}
%       \sum_{n \le x}
%        \ind{\SS}(n)\ind{\SS}(n + h_1)\cdots \ind{\SS}(n + h_k)
          \speccount(\{0\} \cup \bh; x)
 \\
 & \hspace{30pt} 
  \le 
%     \#\{\sts_n \le x : \sts_{n + j} - \sts_{n + j - 1} \le \lambda_j y, j = 1,\ldots,r\}
     \sums[\sts_n \le x]
          [\sts_{n + j} - \sts_{n + j - 1} \le \lambda_j y] %[\gap{n + j - 1} \le \lambda_j y]
          [j = 1,\ldots,r] 1 
%  \\
%  & \hspace{5pt}
   \le 
    \sum_{k = r}^{r + 2\ell}
     (-1)^{k - r}
      \sum_{i_1 + \cdots + i_r = k}
       \hspace{5pt}
        \sum_{\vbh \, \in \, y\Theta_{\vbi,\vbl} \cap \, \ZZ^k}
%        \sum_{n \le x}
%         \ind{\SS}(n)\ind{\SS}(n + h_1)\cdots \ind{\SS}(n + h_k).
           \speccount(\{0\} \cup \bh; x),
  \end{split}            
\end{align}
the sums over $i_1 + \cdots + i_r = k$, here and below, being over 
all $\vbi \in \NN^r$ for which $i_1 + \cdots + i_r = k$.
We make the substitution \eqref{eq:defEterm}, with 
$\{0\} \cup \bh$ and $k + 1$ in place of $\bh$ and $k$; 
we apply Hypothesis ($k,\Theta_{\vbi,\vbl},\{0\}$) for all 
$k$ and $\vbi$ satisfying $r \le k \le r + 2\ell + 1$ and  
$i_1 + \cdots + i_r = k$;
we use Proposition \ref{prop:sssa}, and our 
assumption that $yR_1(x) \sim 1$, i.e.\ $y \sim x/\speccount(x)$, 
as $x \to \infty$.
Thus, we deduce from \eqref{eq:bded1} that 
\begin{equation}
  \label{eq:bded2}
   \sum_{k = r}^{r + 2\ell + 1}
     (-1)^{k - r}
      \sum_{i_1 + \cdots + i_r = k}
       \vol(\Theta_{\vbi,\vbl})
        \le 
         \liminf_{x \to \infty}
          \frac{1}{\speccount(x)}
%           \#\{\sts_n \le x : \sts_{n + j} - \sts_{n + j - 1} \le \lambda_j y, j = 1,\ldots,r\}
           \sums[\sts_n \le x]
                [\sts_{n + j} - \sts_{n + j - 1} \le \lambda_j y] %[\gap{n + j - 1} \le \lambda_j y]
                [j = 1,\ldots,r] 1
\end{equation}
and 
\begin{equation}
 \label{eq:bded3}
          \limsup_{x \to \infty}
           \frac{1}{\speccount(x)}
%            \#\{\sts_n \le x : \sts_{n + j} - \sts_{n + j - 1} \le \lambda_j y, j = 1,\ldots,r\}
            \sums[\sts_n \le x]
                 [\sts_{n + j} - \sts_{n + j - 1} \le \lambda_j y] %[\gap{n + j - 1} \le \lambda_j y]
                 [j = 1,\ldots,r] 1
           \le 
            \sum_{k = r}^{r + 2\ell}
            (-1)^{k - r}
             \sum_{i_1 + \cdots + i_r = k}
              \vol(\Theta_{\vbi,\vbl}).            
\end{equation}
Since 
$
 \vol(\Theta_{\vbi,\vbl}) 
  = \lambda_1^{i_1}\cdots \lambda_r^{i_r}/(i_1!\cdots i_r!)
$, 
the sums on the left and right of \eqref{eq:bded2} and 
\eqref{eq:bded3} are truncations of the Taylor 
series for $(1 - \e^{-\lambda_1})\cdots (1 - \e^{-\lambda_r})$.
We have chosen $\ell$ arbitrarily large, so we may conclude that 
\eqref{eq:thm:mainc} holds, provided 
Hypothesis ($k,\Theta_{\vbi,\vbl},\{0\}$) does whenever 
$k \ge r$ and $i_1 + \cdots + i_r = k$.

(b)
We use an argument of Gallagher \cite{GAL:76}, who proved an 
analogous result for primes.
Let $\ell \ge 1$ be an integer, arbitrarily large but fixed.
We have   
\begin{align*}
  \sum_{n \le x}
   \big(\speccount(n + \lambda y) - \speccount(n) \big)^{\ell}
  & =
     \sum_{n \le x}
      \bigg(\sum_{0 < h \le \lambda y} \ind{\SS}(n + h)\bigg)^{\ell}
 \\
  & = 
     \sum_{n \le x}
      \sum_{0 < h_1,\ldots,h_{\ell} \le \lambda y} \ind{\SS}(n + h_1)\cdots \ind{\SS}(n + h_{\ell})
 \\
  & =
       \sum_{k = 1}^{\ell}
        \varrho(\ell,k)
         \hspace{-5.8pt}
         \sum_{0 < h_1 < \cdots < h_k \le \lambda y}
          \hspace{2pt}
           \sum_{n \le x}
            \ind{\SS}(n + h_1)\cdots \ind{\SS}(n + h_k),
\end{align*}
where $\varrho(\ell,k)$ denotes the number of maps from 
$\{1,\ldots,\ell\}$ onto $\{1,\ldots,k\}$.
Thus,  
\[
 \frac{1}{x}
  \sum_{n \le x}
   \big(\speccount(n + \lambda y) - \speccount(n) \big)^{\ell}
  =
    \sum_{k = 1}^{\ell}
     \bigg(\frac{\speccount(x)}{x}\bigg)^k
      \varrho(\ell,k)
       \sum_{0 < h_1 < \cdots < h_k \le \lambda y} 
        \big(\mathfrak{S}_{\bh} + \cE_{\bh}(x)\big),
\]
with $\bh = \{h_1,\ldots,h_k\}$ in the last summand.
To sum over $0 < h_1 < \cdots < h_k \le \lambda y$ is to sum over 
$(h_1,\ldots,h_k) \in y\Theta_{k,\lambda} \cap \ZZ^k$ (see 
\eqref{eq:defThetkl}).
If Hypothesis ($k,\Theta_{k,\lambda},\emptyset$) holds then for 
some function $\varepsilon(x)$ with $\varepsilon(x) \to 0$ 
($x \to \infty$), we have  
\[
 \sum_{0 < h_1 < \cdots < h_k \le \lambda y} 
  \big(\mathfrak{S}_{\bh} + \cE_{\bh}(x)\big)
   =
    \big(1 + O_{\lambda,k}(\varepsilon(x))\big)
     \sum_{0 < h_1 < \cdots < h_k \le \lambda y} \mathfrak{S}_{\bh}.
\]
Applying Proposition \ref{prop:sssa} (noting that 
$\vol(\Theta_{k,\lambda}) = \lambda^k/k!$), and our assumption 
that $yR_1(x) \sim 1$, i.e.\ $y \sim x/\speccount(x)$, as 
$x \to \infty$, we see that if 
Hypothesis ($k,\Theta_{k,\lambda},\emptyset$) holds for 
$1 \le k \le \ell$, then 
\begin{equation}
 \label{eq:ellthmoment}
 \frac{1}{x}
  \sum_{n \le x}
   \big(\speccount(n + \lambda y) - \speccount(n) \big)^{\ell}
    \sim 
     \sum_{k = 1}^{\ell}
      \varrho(\ell,k)
       \frac{\lambda^k}{k!} 
        \quad 
      (x \to \infty).
\end{equation}
Gallagher's calculation in \cite[Section 3]{GAL:76} shows that 
$
 \sum_{k = 1}^{\ell} \varrho(\ell,k)\lambda^k/k!
$
is the $\ell$th moment of the Poisson distribution with parameter 
$\lambda$, and that the corresponding moment generating function 
is entire.
Since a Poisson distribution is determined by its moments, it 
follows (see \cite[Section 30]{BIL:95}) that for any given 
$m \ge 0$, \eqref{eq:thm:main} holds as $x \to \infty$, provided 
Hypothesis ($k,\Theta_{k,\lambda},\emptyset$) holds for all 
$k \ge 1$.
\end{proof}

%%%%%%%%%%%%%%%%%%%%%%%%%%%%%%%%%%%%%%%%%%%%%%%%%%%%%%%%%%%%%%%%%%
%%%%%%%%%%%%%%%%%%%%%%%%%% SECTION 05 %%%%%%%%%%%%%%%%%%%%%%%%%%%%
%%%%%%%%%%%%%%%%%%%%%%%%%%%%%%%%%%%%%%%%%%%%%%%%%%%%%%%%%%%%%%%%%%
 
\section{Preliminaries}
 \label{sec:prelims}
 
A positive integer $n$ is a sum of two squares if and only if
\[
 n =
   2^{\beta_2}
    \prod_{p \equiv 1 \bmod 4}p^{\beta_p}
     \prod_{p \equiv 3 \bmod 4}p^{2\beta_p},
\]
where $\beta_2,\beta_p$ denote nonnegative integers.
(See \cite[Theorem 366]{HW:38}.)
In view of this and the next proposition, whose proof, being  
routine and elementary, is omitted, we have 
$\SS = \bcap_p S_p$, where 
$
 S_p 
  = 
   \bcap_{\alpha \ge 1} \{n \in \ZZ : n \equiv \sots \bmod p^{\alpha}\}.
$
%(see \eqref{eq:defSp}), as claimed in Section \ref{sec:ktups}.
%
Further, as $S_p = \ZZ$ for primes $p \equiv 1 \bmod 4$, we may 
write $\SS = \bcap_{p \not\equiv 1 \bmod 4} S_p$.

\begin{proposition}
 \label{prop:S2S3S1}
Let $n \in \ZZ$.
We have $n \in S_2$ if and only if either $n = 0$ or 
$n = 2^{\beta}m$ for some $\beta \ge 0$ and 
$m \equiv 1 \bmod 4$.
For $p \equiv 3 \bmod 4$, we have $n \in S_p$ if and only if 
either $n = 0$ or $n = p^{2\beta}m$ for some $\beta \ge 0$ and 
$m \not\equiv 0 \bmod p$.
For $p \equiv 1 \bmod 4$, we have $S_p = \ZZ$.
\end{proposition}

Let us introduce some notation in order to state further results.
Given a nonempty, finite set $\bh \subseteq \ZZ$, let
\begin{equation}
 \label{eq:defdethDh}
 \det(\bh)
  \defeq 
   \prods[h,h' \in \bh][h > h'](h - h') > 0.
\end{equation}
Note that if $p \le k - 1$, where $k = \card \bh$, then two 
elements of $\bh$ occupy the same congruence class modulo $p$, so 
$p \mid \det(\bh)$.
In other words, if $p \nmid \det(\bh)$ then $k \le p$.

Let 
\begin{equation}
 \label{eq:defhp}
  \bh_p \defeq \{h' \in \bh : -h' + \bh \subseteq S_p\}.
\end{equation}
Note that $\bh_2$ contains at most one element, for if 
$h,h' \in \bh_2$ then $\pm(h - h') \in S_2$, which by 
Proposition \ref{prop:S2S3S1} holds only if $h - h' = 0$.
Similarly, if $k = 1$ or $k = 2$, then $\card \bh_2 = 1$.
By Proposition \ref{prop:S2S3S1}, $\bh_p$ for $p \equiv 3 \bmod 4$ 
consists precisely of those elements $h'$ of $\bh$ for which 
$2 \mid \nu_p(h - h')$ for every $h \in \bh$ with $h \ne h'$.
(Recall that $\nu_p(n)$ denotes the $p$-adic valuation of $n$.)
For instance, if $p \nmid \det(\bh)$ then $\bh_p = \bh$.

Given $\alpha \ge 1$, let 
\begin{equation}
 \label{eq:defTh}
  \T_{\bh}(2^{\alpha + 1})
   \defeq 
    \{0 \le a < 2^{\alpha + 1} : a + \bh \subseteq S_2 
      \,\, \hbox{and} \,\, 
       {\textstyle \max_{h \in \bh}} \nu_2(a + h) < \alpha\}.
\end{equation}
By Proposition \ref{prop:S2S3S1}, this is the (possibly empty) set 
of least residues $a$ modulo $2^{\alpha + 1}$ such that, for each  
$h \in \bh$, there is some $\beta \le \alpha - 1$ and 
$m \equiv 1 \bmod 4$ such that $a + h = 2^{\beta}m$.
Finally, for $p \equiv 3 \bmod 4$, let 
\begin{equation}
 \label{eq:defVh}
  \V_{\bh}(p^{\alpha})
   \defeq 
    \{0 \le a < p^{\alpha} : a + \bh \subseteq S_p 
      \,\, \hbox{and} \,\, 
       {\textstyle \max_{h \in \bh}} \nu_p(a + h) < \alpha\}.
\end{equation}
This is the (possibly empty) set of least residues $a$ modulo 
$p^{\alpha}$ such that, for each $h \in \bh$, there exists 
$\beta \le (\alpha - 1)/2$ for which $p^{2\beta} \emid a + h$. 
Note that, for $\alpha \ge 2$ and odd $p$, the difference between 
$\T_{\bh}(2^{\alpha})$ and $\V_{\bh}(p^{\alpha})$ is that 
$\T_{\bh}(2^{\alpha})$ contains only integers $a$ for which 
$\max_{h \in \bh} \nu_2(a + h) \le \alpha - 2$, whereas 
$\V_{\bh}(p^{\alpha})$ contains $a$ for which 
$\max_{h \in \bh} \nu_p(a + h) \le \alpha - 1$.
As may be expected in view of Proposition \ref{prop:S2S3S1}, we 
will need to treat $p = 2$ as a special case throughout.

Recall from \eqref{eq:delthp} that  
$
 \delta_{\bh}(p) 
  \defeq 
   \lim_{\alpha \to \infty} \card S_{\bh}(p^{\alpha})/p^{\alpha}
$, where 
\[
 S_{\bh}(p^{\alpha}) 
  \defeq 
   \{0 \le a < p^{\alpha} : \forall h \in \bh, a + h \equiv \sots \bmod p^{\alpha}\}.
\]
We have introduced $\V_{\bh}(p^{\alpha})$ because it is more 
convenient than $S_{\bh}(p^{\alpha})$ to work with.
It is not difficult to see that, for $p \not\equiv 1 \bmod 4$, 
$
 0 
  \le 
   \card S_{\bh}(p^{\alpha}) - \card T_{\bh}(p^{\alpha}) 
    \le 1
$ 
once $\alpha$ is sufficiently large.
(One may verify Proposition \ref{prop:S2S3S1} by showing 
that $n \equiv \sots \bmod 2^{\alpha}$ if and only if 
$n \equiv 2^{\beta}m \bmod 2^{\alpha}$ for some $\beta \ge 0$ and 
odd $m$, and, for $p \equiv 3 \bmod 4$, that 
$n \equiv \sots \bmod p^{\alpha}$ if and only if
$n \equiv p^{2\beta}m \bmod p^{\alpha}$ for some $\beta \ge 0$ 
and $m \not\equiv 0 \bmod p$.) 
Thus, the limit $\delta_{\bh}(p)$ exists if and only if 
$\lim_{\alpha \to \infty} \card \V_{\bh}(p^{\alpha})/p^{\alpha}$ 
exists, in which case the two are equal.

In the next two propositions, and throughout, we allow for the 
possibility that $k = 1$.
In case $\bh = \{h_1\}$, we define 
$\max_{i \ne j} \nu_p(h_i - h_j)$ to be zero (and 
$\det(\bh) \defeq 1$). 

\begin{proposition}
 \label{prop:S2h}
Let $\bh = \{h_1,\ldots,h_k\}$ be a set of $k \ge 1$ distinct 
integers.

\textup{(}a\textup{)}
The limits $\delta_{\bh}(2)$ 
\textup{(}see \eqref{eq:delthp}\textup{)} and 
$
 \lim_{\alpha \to \infty} 
  \card \T_{\bh}(2^{\alpha + 1})/2^{\alpha + 1}
$ exist, and are equal:
\begin{equation}
 \label{eq:defdelth2}
  \delta_{\bh}(2)
   =
    \lim_{\alpha \to \infty}
     \frac{\card \T_{\bh}(2^{\alpha + 1})}{2^{\alpha + 1}}.
\end{equation}
Moreover, for all $\alpha \ge 1$, we have  
\begin{equation}
 \label{eq:Thdelth2bnd}
  \bigg|
   \frac{\card \T_{\bh}(2^{\alpha + 1})}{2^{\alpha + 1}}
     -
      \delta_{\bh}(2)
  \bigg|
   \le 
    \frac{k}{2^{\alpha}}. 
\end{equation}

\textup{(}b\textup{)}
For any $\alpha \ge 2 + \max_{i \ne j} \nu_2(h_i - h_j)$, we have 
\begin{equation}
 \label{eq:delth2}
  \delta_{\bh}(2)
   =
    \frac{\card \T_{\bh}(2^{\alpha + 1}) + \card \bh_2}{2^{\alpha + 1}}, 
\end{equation}
the right-hand side being constant for $\alpha$ in this range.

\textup{(}c\textup{)}
If $2 \nmid \det(\bh)$ 
\textup{(}in which case $k \le 2$\textup{)}, then 
$\delta_{\bh}(2) = (1/2)^k$.
As a special case, we record here that $\delta_{\bz}(2) = 1/2$.
\end{proposition}

\begin{proof}
In essence, we use a Hensel-type argument: for $\alpha \ge 1$, the 
condition that $n \equiv \sots \bmod 2^{\alpha}$ can be lifted
to $n \equiv \sots \bmod 2^{\alpha + 1}$, unless $n = 2^{\alpha}m$
for some $m \equiv 3 \bmod 4$.

(a)
As already noted, to show that $\delta_{\bh}(2)$ and the 
right-hand side of \eqref{eq:defdelth2} exist and are equal, it 
suffices to show that the right-hand side exists.
Let $\alpha \ge 1$ and let $0 \le b < 2^{\alpha + 2}$, so 
$b = a + 2^{\alpha + 1}q$, where 
$0 \le a < 2^{\alpha + 1}$ and either $q = 0$ or $q = 1$.
Suppose that, for each $i$, there exists $\beta_i \le \alpha - 1$ 
and $m_i \equiv \pm 1 \bmod 4$ such that 
$b + h_i = 2^{\beta_i}m_i$.
Then, for each $i$, $a + h_i = 2^{\beta_i}m'_i$ 
and $a + 2^{\alpha + 1} + h_i = 2^{\beta_i}m''_i$, where 
$m'_i \equiv m''_i \equiv m_i \bmod 4$.
Recalling Proposition \ref{prop:S2S3S1} and definition 
\eqref{eq:defTh}, we see that the following statements are 
equivalent: 
(i) $b \in \T_{\bh}(2^{\alpha + 2})$; 
(ii) both $a$ and $a + 2^{\alpha + 1}$ are in 
$\T_{\bh}(2^{\alpha + 2})$;
(iii) $a \in \T_{\bh}(2^{\alpha + 1})$.

We have shown that we have a partition 
\[
 \T_{\bh}(2^{\alpha + 2})
  =
   \{a,a + 2^{\alpha + 1} : a \in \T_{\bh}(2^{\alpha + 1})\}
    \cup
     \U_{\bh}(2^{\alpha + 2}),
\]
where 
\[
 \U_{\bh}(2^{\alpha + 2}) 
  \defeq 
   \{0 \le b < 2^{\alpha + 2} : b + \bh \subseteq S_2 
        \,\, \hbox{\textup{and}} \,\, 
         {\textstyle \max_{h \in\bh} } \nu_2(b + h) = \alpha\}  
\]
is the set of elements $b$ of $\T_{\bh}(2^{\alpha + 2})$ for which 
$\nu_2(b + h_j) = \alpha$ for some $h_j \in \bh$.
Any element of $\U_{\bh}(2^{\alpha + 2})$ is a least
residue of 
$\pm 2^{\alpha} - h_j$ for some $h_j \in \bh$, of which there are 
at most $2k$.
We see that 
\[
   \frac{\card \T_{\bh}(2^{\alpha + 2})}{2^{\alpha + 2}}
  -
    \frac{\card \T_{\bh}(2^{\alpha + 1})}{2^{\alpha + 1}}
   =
      \frac{\card \U_{\bh}(2^{\alpha + 2})}{2^{\alpha + 2}}
       \le 
        \frac{k}{2^{\alpha + 1}}.
\]
Consequently, for any $\beta$ with $\beta \ge \alpha$, we have 
\[
 0
  \le 
   \frac{\card \T_{\bh}(2^{\beta + 1})}{2^{\beta + 1}}
  -
    \frac{\card \T_{\bh}(2^{\alpha + 1})}{2^{\alpha + 1}}
     =
      \sum_{r = 1}^{\beta - \alpha}
       \frac{\card \U_{\bh}(2^{\alpha + r + 1})}{2^{\alpha + r + 1}}
        <
         \frac{k}{2^{\alpha}}.
\]
It follows that the limit on the right-hand side of 
\eqref{eq:defdelth2} exists, and that \eqref{eq:Thdelth2bnd} holds 
for all $\alpha \ge 1$.

(b)
Assume that $\alpha \ge 2 + \max_{i \ne j} \nu_2(h_i - h_j)$.
Suppose that, for some $j$, there exists $q$ such that 
$b + h_j = 2^{\alpha}(1 + 2q)$.
We have $b + h_j \in S_2$ if and only if 
$2 \mid q$, equivalently, 
$b + h_j \equiv 2^{\alpha} \bmod 2^{\alpha + 2}$.
For $i \ne j$ 
we may write $h_i - h_j = 2^{\beta_{ij}}m_{ij}$ with 
$\beta_{ij} \le \alpha - 2$ and $m_{ij} \equiv \pm 1 \bmod 4$.
Thus,  
\[
 b + h_i 
  = 2^{\beta_{ij}}(m_{ij} + 2^{\alpha - \beta_{ij}}(1 + 2q))
\]
is in $S_2$ if and only if $m_{ij} \equiv 1 \bmod 4$, 
equivalently, $h_i - h_j \in S_2$.
By definition of $\bh_2$, this holds for each $i \ne j$ if and 
only if $h_j \in \bh_2$.
We have shown that $b \in \T_{\bh}(2^{\alpha + 2})$ and 
$\nu_2(b + h_j) = \alpha$ for some $h_j \in \bh$ if and only if  
$\bh_2$ is nonempty, $h_j$ is the (necessarily unique) element of 
$\bh_2$, and $b + h_j \equiv 2^{\alpha} \bmod 2^{\alpha + 2}$. 
Thus,
\[
  \U_{\bh}(2^{\alpha + 2})
   =
    \{0 \le b < 2^{\alpha + 2} : \exists h' \in \bh_2, b \equiv 2^{\alpha} - h' \bmod 2^{\alpha + 2}\},
\]
and $\card \U_{\bh}(2^{\alpha + 2}) = \card \bh_2$.
Also, 
$
 \card \T_{\bh}(2^{\alpha + 2}) 
 = 2\card \T_{\bh}(2^{\alpha + 1}) + \card \bh_2
$.
Hence 
\[
 \frac{\card \T_{\bh}(2^{\alpha + 2}) + \card \bh_2}{2^{\alpha + 2}}
  =
   \frac{\card \T_{\bh}(2^{\alpha + 1}) + \card \bh_2}{2^{\alpha + 1}}.
\]

(c) 
Suppose $2 \nmid \det(\bh)$.
If $k = 1$, i.e.\ if $\bh = \{h_1\}$, then the elements of 
$\T_{\bh}(8)$ are precisely the least residues of 
$1 - h_1, 2 - h_1$ and $5 - h_1$ modulo $8$.
Also, $\bh_2 = \bh$.
If $k = 2$, i.e.\ if $\bh = \{h_1,h_2\}$, then either 
$h_2 - h_1 \equiv 1 \bmod 4$ or $h_1 - h_2 \equiv 1 \bmod 4$.
Without loss of generality, suppose $h_2 - h_1 \equiv 1 \bmod 4$.
Then the sole element of $\T_{\bh}(8)$ is the least residue of 
$h_2 - 2h_1$ modulo $8$.
Also, $\bh_2 = \{h_1\}$.
Therefore, by (b), $\delta_{\bh}(2) = (1/2)^k$.
\end{proof}

For the next proposition, recall that $\alpha \bmod 2$, when 
written in an exponent, denotes the least residue of $\alpha$ 
modulo $2$.
For instance, $p^{\alpha \bmod 2} = 1$ if $\alpha$ is even.

\begin{proposition}
 \label{prop:Sp3h}
Let $\bh = \{h_1,\ldots,h_k\}$ be a set of $k \ge 1$ distinct 
integers, and let $p$ be a prime with $p \equiv 3 \bmod 4$.

\textup{(}a\textup{)}
The limits $\delta_{\bh}(p)$ 
\textup{(}see \eqref{eq:delthp}\textup{)} and 
$\lim_{\alpha \to \infty} \card \V_{\bh}(p^{\alpha})/p^{\alpha}$ 
exist, and are equal:
\begin{equation}
 \label{eq:defdelthp3}
  \delta_{\bh}(p)
   = 
    \lim_{\alpha \to \infty}
     \frac{\card \V_{\bh}(p^{\alpha})}{p^{\alpha}}.
\end{equation}
Moreover, for all $\alpha \ge 1$, we have 
\begin{equation}
 \label{eq:Vhdeltp3bnd}
  \bigg|
   \frac{\card \V_{\bh}(p^{\alpha})}{p^{\alpha}}
     -
      \delta_{\bh}(p)
  \bigg|
   \le 
    \frac{k}{p^{\alpha }}
     \bigg(1 + \frac{1}{p}\bigg)^{-1} 
      \frac{1}{p^{\alpha \bmod 2}}.
\end{equation}

\textup{(}b\textup{)}
For any $\alpha \ge 1 + \max_{i \ne j} \nu_p(h_i - h_j)$, we have 
\begin{equation}
 \label{eq:delthp3}
  \delta_{\bh}(p)
   =
    \frac{1}{p^{\alpha}}
     \bigg(
      \card \V_{\bh}(p^{\alpha}) + \card \bh_p \bigg(1 + \frac{1}{p}\bigg)^{-1} \frac{1}{p^{\alpha \bmod 2}}
     \bigg),
\end{equation}
the right-hand side being constant for $\alpha$ in this range.

\textup{(}c\textup{)}
We have  
\begin{equation}
 \label{eq:delthpropsp3}
  \delta_{\bh}(p)
   \ge 
    \bigg(1 + \frac{1}{p}\bigg)^{-1}
     \bigg(1 - \frac{\min\{k - 1,p\}}{p}\bigg),
\end{equation}
with {\bfseries equality} attained if $p \nmid \det(\bh)$ 
\textup{(}in which case $k \le p$\textup{)}.
As a special case, we record here that 
$\delta_{\bz}(p) = (1 + 1/p)^{-1}$.
\end{proposition}

\begin{proof}
(a)
As noted above the statement of Proposition \ref{prop:S2h}, to 
show that $\delta_{\bh}(p)$ and the right-hand side of 
\eqref{eq:defdelthp3} exist and are equal, it suffices to show 
that the right-hand side exists.
Let $\alpha \ge 1$ and let $0 \le b < p^{\alpha + 1}$.
Thus, $b = a + p^{\alpha}q$, where $0 \le a < p^{\alpha}$ and 
$0 \le q < p$.
Suppose that, for each $i$, there exists $\beta_i \le \alpha - 1$ 
and $m_i \not\equiv 0 \bmod p$ such that 
$b + h_i = p^{\beta_i}m_i$.
Then, for each $i$ and each $q'$, $0 \le q' < p$, we have 
$a + p^{\alpha}q' + h_i = p^{\beta_i}m_i'$, where 
$m_i' \equiv m_i \not\equiv 0 \bmod p$.
Recalling Proposition \ref{prop:S2S3S1} and definition 
\eqref{eq:defVh}, we see that the following are equivalent: 
(i) $b \in \V_{\bh}(p^{\alpha + 1})$; 
(ii) $a + p^{\alpha}q' + h_i \in \V_{\bh}(p^{\alpha + 1})$ for 
$0 \le q' < p$;
(iii) $a \in \V_{\bh}(p^{\alpha})$.

We have shown that we have a partition
\[
  \V_{\bh}(p^{\alpha + 1})
   = 
    \{a + p^{\alpha}q : a \in \V_{\bh}(p^{\alpha}), 0 \le q < p\}
     \cup 
      \W_{\bh}(p^{\alpha + 1}),
\]
where
\[
 \W_{\bh}(p^{\alpha + 1}) 
  \defeq 
   \{0 \le b < p^{\alpha + 1} : b + \bh \subseteq S_p 
        \,\, \hbox{\textup{and}} \,\, 
         {\textstyle \max_{h \in\bh} } \nu_p(b + h) = \alpha\} 
\]
is the set of elements $b$ of $\V_{\bh}(p^{\alpha + 1})$ for which 
$\nu_p(b + h_j) = \alpha$ for some $h_j \in \bh$.
Plainly, $\W_{\bh}(p^{\alpha + 1})$ is empty if $\alpha$ is odd.
(If $b + \bh \subseteq S_p$ then, by 
Proposition \ref{prop:S2S3S1}, $\nu_p(b + h_j)$ is even and hence 
not equal to any odd $\alpha$.)
Also, any element of $\W_{\bh}(p^{\alpha + 1})$ is a least residue 
of $p^{\alpha}q - h_j \bmod p^{\alpha + 1}$, for some 
$0 < q < p$ and $h_j \in \bh$, of which there are at most 
$(p - 1)k$.
We see that 
\begin{equation}
 \label{eq:VW}
 \frac{\card \V_{\bh}(p^{\alpha + 1})}{p^{\alpha + 1}}
  -
   \frac{\card \V_{\bh}(p^{\alpha})}{p^{\alpha}}
    =
     \frac{\card \W_{\bh}(p^{\alpha + 1})}{p^{\alpha + 1}},
\end{equation}
and that 
\begin{equation}
 \label{eq:VWb}
 0 
  \le 
   \frac{\card \W_{\bh}(p^{\alpha + 1})}{p^{\alpha + 1}}
    \le 
     \bigg(1 - \frac{1}{p}\bigg)\frac{k}{p^{\alpha}},
\end{equation}
with {\em equality} on the left if $\alpha$ is {\em odd}.
Consequently, for any $\beta$ with $\beta \ge \alpha$, we have 
\begin{align*}
 0
 \le
  \frac{\card \V_{\bh}(p^{\beta})}{p^{\beta}}
   -
    \frac{\card \V_{\bh}(p^{\alpha})}{p^{\alpha}}
 =
  \sum_{r = 1}^{\beta - \alpha}
   \frac{\card \W_{\bh}(p^{\alpha + r})}{p^{\alpha + r}}
    <
     \bigg(1 - \frac{1}{p}\bigg)\frac{k}{p^{\alpha}} 
      \sums[r - 1 \ge 0][r - 1 \equiv \alpha \bmod 2]
       \frac{1}{p^{r - 1}}.
\end{align*}
Since this last sum is equal to $1/(1 - 1/p^2)$ if $\alpha$ is 
even, and to $1/(p(1 - 1/p^2))$ if $\alpha$ is odd, we have 
\[
  0
   \le
    \frac{\card \V_{\bh}(p^{\beta})}{p^{\beta}} 
   - \frac{\card \V_{\bh}(p^{\alpha})}{p^{\alpha}}
      <
       \frac{k}{p^{\alpha}}
        \bigg(1 + \frac{1}{p}\bigg)^{-1}
         \frac{1}{p^{\alpha \bmod 2}}.
\]
It follows that the limit on the right-hand side of 
\eqref{eq:defdelthp3} exists, and that \eqref{eq:Vhdeltp3bnd} 
holds for all $\alpha \ge 1$.

(b)
Let $0 \le b < p^{\alpha + 1}$, and assume now that 
$\alpha \ge 1 + \max_{i \ne j} \nu_p(h_i - h_j)$.
Suppose that, for some $j$, we have $b + h_j = p^{\alpha}m_j$ 
for some $m_j \not\equiv 0 \bmod p$.
We have $b + h_j \in S_p$ if and only if $\alpha$ is even.
Let $i \ne j$.
We may write $h_i - h_j = p^{\beta_{ij}}m_{ij}$ with 
$\beta_{ij} \le \alpha - 1$ and $m_{ij} \not\equiv 0 \bmod p$.
Thus, 
$b + h_i = p^{\beta_{ij}}(m_{ij} + p^{\alpha - \beta_{ij}}m_j)$ 
is in $S_p$ if and only if $\beta_{ij}$ is even, equivalently, 
$h_i - h_j \in S_p$.
By definition of $\bh_p$, this holds for each $i \ne j$ if and 
only if $h_j \in \bh_{p}$.
In that case, for $0 \le q' < p$ with 
$q' \not\equiv -m_j \bmod p$, we have 
$b + p^{\alpha}q' + h_i \in S_p$ and 
$\nu_p(b + p^{\alpha}q' + h_i) = \beta_{ij} < \alpha$ 
for $i \ne j$;    
$b + p^{\alpha}q' + h_j \in S_p$ if and only if $b + h_j \in S_p$, 
and $\nu_p(b + p^{\alpha}q' + h_j) = \alpha$.
For $q' \equiv -m_j \bmod p$, 
$\nu_p(b + p^{\alpha}q' + h_j) > \alpha$.

Thus, if $\W_{\bh}(p^{\alpha + 1}) \ne \emptyset$, then 
$\alpha$ is even and $\bh_p \ne \emptyset$; and if 
$b \in \W_{\bh}(p^{\alpha + 1})$, then the $h_j$ for which 
$\nu_p(b + h_j) = \alpha$ is uniquely determined by $b$ and must 
lie in $\bh_p$.
If $\alpha$ is even, then, writing $h_j = p^{\alpha}q_j + r_j$, 
with $0 \le r_j < p^{\alpha}$, we see that 
\[
 \W_{\bh}(p^{\alpha + 1}) 
  = 
   \bcup_{h_j \in \bh_p}
    \{p^{\alpha}(q' + 1) - r_j : 0 \le q' < p, q' \not\equiv -q_j \bmod p\}.
\]
Thus, 
$
 \card \V_{\bh}(p^{\alpha + 1})
  = 
   p\card \V_{\bh}(p^{\alpha}) 
$
if $\alpha$ is odd, and 
$
 \card \V_{\bh}(p^{\alpha + 1})
  = 
   p\card \V_{\bh}(p^{\alpha}) + (p - 1)\card \bh_p 
$
if $\alpha$ is even.
Consequently, if $\alpha$ is odd then 
\[
 \frac{1}{p^{\alpha + 1}}
  \bigg(
     \card \V_{\bh}(p^{\alpha + 1}) + \card\bh_p\frac{p}{p + 1}
  \bigg)
   =
 \frac{1}{p^{\alpha}}
  \bigg(
     \card \V_{\bh}(p^{\alpha}) + \card\bh_p\frac{1}{p + 1}
  \bigg),
\]
while if $\alpha$ is even then 
\[
 \frac{1}{p^{\alpha + 1}}
  \bigg(
     \card \V_{\bh}(p^{\alpha + 1}) + \card\bh_p\frac{1}{p + 1}
  \bigg)
   =
 \frac{1}{p^{\alpha}}
  \bigg(
     \card \V_{\bh}(p^{\alpha}) + \card\bh_p\frac{p}{p + 1}
  \bigg).
\]

(c) Note that 
$
 \V_{\bh}(p)
  =
   \{0 \le a < p : \forall i, a \not\equiv -h_i \bmod p\} 
$,
so $\card \V_{\bh}(p) = p - \kappa$ where $\kappa$ is the 
number of distinct congruence classes in  
$\{h_i \bmod p : h_i \in \bh\}$.
Thus, $\kappa = k$ if and only if $p \nmid \det(\bh)$.
First, consider the case $p \mid \det(\bh)$, 
i.e.\ $\kappa \le k - 1$.
As $\delta_{\bh}(p) \ge 0$, \eqref{eq:delthpropsp3} is trivial for 
$p \le k - 1$, so let us assume that $k \le p$.
The relation \eqref{eq:VW} shows that 
$
 \card \V_{\bh}(p^{\alpha + 1})/p^{\alpha + 1} 
  \ge 
   \card \V_{\bh}(p^{\alpha})/p^{\alpha}
$ 
for $\alpha \ge 1$, and hence 
\[
 \delta_{\bh}(p)
  \ge 
   \frac{\card \V_{\bh}(p)}{p}
    \ge 
     \frac{p - (k - 1)}{p}
       > 
        1 - \frac{k}{p + 1}.
\]
The right-hand side of \eqref{eq:delthpropsp3} is equal to 
$1 - k/(p + 1)$ when $\min\{k - 1,p\} = k - 1$, as we are 
currently assuming. 
Next, consider the case $p \nmid \det(\bh)$, i.e.\ $\kappa = k$.
In this case, we have $\card \bh = \card \bh_p = k$ and, by 
\eqref{eq:delthp3},
\[
 \delta_{\bh}(p)
  = 
   \frac{1}{p}
    \bigg(\card \V_{\bh}(p) + \card\bh_p\bigg(1 + \frac{1}{p}\bigg)^{-1}\frac{1}{p} \bigg)
     =
      \bigg(1 + \frac{1}{p}\bigg)^{-1}
       \bigg(1 - \frac{k - 1}{p}\bigg),
\]
which is equal to the right-hand side of \eqref{eq:delthpropsp3} 
(since $p \ge \kappa = k$).
\end{proof}

Notice that, for all $p \not\equiv 1 \bmod 4$, we have 
$0 \le \delta_{\bh}(p) \le 1$, by definition.
%
%From 
% definitions \eqref{eq:defdelth2}, \eqref{eq:defdelthp3} and 
% Propositions \ref{prop:S2h} (b) and \ref{prop:Sp3h} (b), it 
% follows
% that $\delta_{\bh}(p) > 0$ if and only if  
% $n + \bh \subseteq S_p$ for some $n$, and so $\bh$ is 
%that $\bh$ is 
% $\SS$-admissible %(see Definition \ref{def:Sadm})
% if and only if 
% $\delta_{\bh}(p) > 0$ for all $p$. 
%
By the following proposition,
%$\SS$-admissibility of $\bh$
%is 
%thus equivalent to the nonvanishing of its singular series
%i.e.,
the nonvanishing of its singular series
$
  \mathfrak{S}_{\bh}
   \defeq 
    \prod_{p \not\equiv 1 \bmod 4}
%      \big(
      \delta_{\bz}(p)^{-k}
       \delta_{\bh}(p),
%      \big)
$
%(see Definition \ref{def:Sss}), as claimed in 
%Section \ref{sec:ktups}.
is equivalent to $\delta_{\bh}(p) > 0$ for all $p$.

%****************************************************************%
%************************* START DETAIL *************************%
%****************************************************************%
%
\begin{nixnix}
Consider $p \equiv 3 \bmod 4$.
Suppose $n + \bh \subseteq S_p$ for some $n$.
By Proposition \ref{prop:S2S3S1}, for each $h_i \in \bh$ we have 
$n + h_i = p^{2\beta_i}m_i$ for some $\beta_i \ge 0$ and 
$m_i \not\equiv 0 \bmod p$.
Let $\alpha > \max_{h_i \in \bh} 2\beta_i$, and let $a$ be the 
least residue of $n$ modulo $p^{\alpha}$. 
Say $n = a + p^{\alpha}q$.
Then, for each $h_i \in \bh$, 
$
 a + h_i 
  = p^{2\beta_i}(m_i + p^{\alpha - 2\beta_i}q)
$, 
and 
$
 m_i + p^{\alpha - 2\beta_i}q \equiv m_i \not\equiv 0 \bmod p
$.
Thus, $a + \bh \subseteq S_p$ by Proposition \ref{prop:S2S3S1}, 
and $\max_{h \in \bh} \nu_p(a + h) < \alpha$.
Hence $a \in \V_{\bh}(p^{\alpha})$. 
Since $\V_{\bh}(p^{\alpha}) \ne \emptyset$, 
$\delta_{\bh}(p^{\alpha}) \ge 1/p^{\alpha} > 0$ by 
Proposition \ref{prop:Sp3h} (b).

Conversely, suppose $\delta_{\bh}(p) > 0$.
Then by Proposition \ref{prop:Sp3h} (b), either 
$\V_{\bh}(p^{\alpha}) \ne \emptyset$ or $\bh_p \ne \emptyset$, 
where $\alpha = 1 + \max_{i \ne j} \nu_p(h_i - h_j)$.
If $\V_{\bh}(p^{\alpha}) \ne \emptyset$ then there is some $a$, 
$0 \le a < p^{\alpha}$, such that $a + \bh \subseteq S_p$ and 
$\max_{h \in \bh} \nu_p(a + h) < \alpha$.
If $\bh_p \ne \emptyset$, we have $h_j \in \bh$ such that 
$-h_j + h_i \in S_p$ for each $h_i \in \bh$.
In the first case, take $n = a$; in the second case take 
$n = -h_j$.
In either case, we have some $n$ such that 
$n + \bh \subseteq S_p$.

Similarly, $n + \bh \subseteq S_2$ for some $n$ if and only if 
$\delta_{\bh}(2) > 0$.
\end{nixnix}
%
%****************************************************************%
%************************** END DETAIL **************************%
%****************************************************************%

\begin{proposition}
 \label{prop:sssc}
Let $\bh = \{h_1,\ldots,h_k\}$ be a set of $k \ge 1$ distinct 
integers.
We have 
\begin{equation}
 \label{eq:sssc1}
 \e^{-(k - 1)}
  \le 
   \prods[p \not\equiv 1 \bmod 4][p \nmid \det(\bh)] 
    \delta_{\bz}(p)^{-k}\delta_{\bh}(p)
     \le 
      1,
\end{equation}
and the product converges.
Consequently, 
{\small 
\begin{equation}
 \label{eq:sssc2}
  \frac{2^{k} \delta_{\bh}(2) }{\e^{k - 1} }
   \prods[p \equiv 3 \bmod 4][p \mid \det(\bh)]
    \bigg(   
     \bigg( 1 + \frac{1}{p} \bigg)^k 
      \delta_{\bh}(p) 
    \bigg)
  \le 
   \mathfrak{S}_{\bh}
    \le 
     2^k\delta_{\bh}(2)
      \prods[p \equiv 3 \bmod 4][p \mid \det(\bh)]
       \bigg(    
        \bigg( 1 + \frac{1}{p} \bigg)^k
         \delta_{\bh}(p) 
       \bigg).
\end{equation}
}
\end{proposition}

\begin{proof}
If $2 \nmid \det(\bh)$ then $k \le 2$ and 
$\delta_{\bz}(2)^{-k}\delta_{\bh}(2) = 1$ by 
Proposition \ref{prop:S2h} (c), so only the primes 
$p \equiv 3 \bmod 4$ have any bearing on the product in 
\eqref{eq:sssc1}.
Let $p \equiv 3 \bmod 4$, and suppose $p \nmid \det(\bh)$.
By Proposition \ref{prop:Sp3h} (c), $k \le p$ and 
\begin{equation}
 \label{eq:pnmiddeth}
  \delta_{\bz}(p)^{-k}\delta_{\bh}(p)
   =
    \bigg(1 + \frac{1}{p}\bigg)^{k - 1}
     \bigg(1 - \frac{k - 1}{p}\bigg).
\end{equation}
Thus, $\delta_{\bz}(p)^{-k}\delta_{\bh}(p) = 1 + O_k(1/p^2)$, and 
consequently the product in \eqref{eq:sssc1} converges.

More precisely, from \eqref{eq:pnmiddeth} we have, on the one 
hand,  
\[
  \delta_{\bz}(p)^{-k}\delta_{\bh}(p)
   =
    1 
    - 
     \sum_{j = 2}^k 
      \bigg\{
       (k - 1)\binom{k - 1}{j - 1} - \binom{k - 1}{j}
      \bigg\}
      p^{-j}
       \le 
        1,
\]
with equality attained if $k = 1$, which gives the upper bound in 
\eqref{eq:sssc1}, and also the lower bound for $k = 1$.
On the other hand we have
\[
 \delta_{\bz}(p)^{-k}\delta_{\bh}(p)
  \ge 
   1 - \frac{(k - 1)^2}{p^2}.
\]
For $k = 2$ we see that the product in \eqref{eq:sssc1} is at 
least $\prod_{p \equiv 3 \bmod 4}(1 - 1/p^2)$, which is equal to 
$1/(2C^2) = 0.856108\ldots$ (with $C$ being the Landau--Ramanujan 
constant; see \eqref{eq:sotsnt}), and is 
greater than $\e^{-1}$.
For $k \ge 3$ we apply the basic inequality 
$\log(1 - x) \ge -x/(1 - x)$ ($0 \le x < 1$) to the above, 
obtaining 
\[
 \textstyle 
 \log \delta_{\bz}(p)^{-k}\delta_{\bh}(p)
  \ge 
  -\frac{(k - 1)^2}{p^2}\Big(1 - \frac{(k - 1)^2}{p^2}\Big)^{-1}
    \ge
   - \frac{(k - 1)^2}{p^2}\Big(1 - \frac{(k - 1)^2}{k^2}\Big)^{-1}
\]
(since $k \le p$).
Noting that  
$
 -\sum_{p \nmid \det(\bh)} 1/p^2
  \ge 
   -\sum_{n \ge k} 1/n^2
    \ge 
     -1/(k - 1)^2
$,
and that  
$-(1 - (k - 1)^2/k^2)^{-1} = -k^2/(2k - 1) > -(k - 1)$, then 
exponentiating, we see that product in \eqref{eq:sssc1} is greater 
than $\e^{-(k - 1)}$.
The inequalities in \eqref{eq:sssc2} follow upon recalling that 
$\delta_{\bz}(p) = (1 + 1/p)^k$ for $p \equiv 3 \bmod 4$ (see 
Proposition \ref{prop:Sp3h} (c)), and again that 
$\delta_{\bz}(2)^{-k}\delta_{\bh}(2) = 1$ if $2 \nmid \det(\bh)$ 
(see Proposition \ref{prop:S2h} (c)).
\end{proof}

%%%%%%%%%%%%%%%%%%%%%%%%%%%%%%%%%%%%%%%%%%%%%%%%%%%%%%%%%%%%%%%%%%
%%%%%%%%%%%%%%%%%%%%%%%%%% SECTION 06 %%%%%%%%%%%%%%%%%%%%%%%%%%%%
%%%%%%%%%%%%%%%%%%%%%%%%%%%%%%%%%%%%%%%%%%%%%%%%%%%%%%%%%%%%%%%%%%

\section{Proof of Proposition \ref{prop:sssa}}
 \label{sec:keyprop}

We will make use of the following elementary bounds.
Recall that, for $n \in \NN$, 
$\omega(n) \defeq \#\{p : p \mid n\}$, 
$\rad(n) \defeq \prod_{p \mid n} p$, and 
$\sqfr(n) \defeq \prod_{p \emid n} p$.
 
\begin{lemma} 
 \label{lem:omegabnd}
Let 
\begin{equation}
 \label{eq:defcN}
  \cN 
   \defeq 
    \{ab^2\rad(b) : a,b \in \NN, (a,b) = 1, a \text{ squarefree}\}.
\end{equation}
Fix any number $A \ge 1$.
For $y \ge 1$ and integers $D \ge 1$, we have 
\begin{equation}
 \label{eq:realbnd}
  \sums[n \in \cN][n > y]
   A^{\omega(n)}
    \frac{(D,\rad(n))}{n\sqfr(n)}
     \ll_A
      (1 + A)^{2\omega(D)}
       \frac{y^{O(1/\log\log 3y)}}{y^{2/3}},  
\end{equation}
and 
\begin{equation}
 \label{eq:realbnd2}
  \sums[n \in \cN][n \le y]
   \frac{A^{\omega(n)}}{\sqfr(n)}
    \ll_A
     y^{1/3 + O(1/\log\log 3y)}.
\end{equation}
\end{lemma}

\begin{proof}
Let $y \ge 1$ and let $D \ge 1$.
We claim that the following four bounds hold:
\begin{equation}
  \label{eq:omegasfbnd}
   \sums[n > y][\text{squarefree}] %\sumss[\flat][n > y]
    A^{\omega(n)}
     \frac{(D,n)}{n^2}
      \ll_A
       (1 + A)^{\omega(D)} \frac{y^{O(1/\log\log 3y)}}{y};
\end{equation}
\begin{equation}
 \label{eq:omegasfbnda}
  \sums[n \le y][\text{squarefree}] %\sumss[\flat][n \le y]
   A^{\omega(n)}
    \frac{(D,n)}{n}
     \ll_A
      (1 + A)^{\omega(D)} y^{O(1/\log\log 3y)};
\end{equation}
\begin{equation}
   \label{eq:omegaradbnd}
    \sum_{n^2\rad(n) > y}
     \frac{A^{\omega(n)}(D,\rad(n))}{n^2\rad(n)}
      \ll_A
       (1 + A)^{\omega(D)}\frac{y^{O(1/\log\log 3y)}}{y^{2/3}};
\end{equation}
and 
\begin{equation}
   \label{eq:smth13}
    \sum_{n^2\rad(n) \le y}
     A^{\omega(n)} 
      \ll_A
       y^{1/3 + O(1/\log\log 3y)}.
\end{equation}
Let us deduce \eqref{eq:realbnd} and \eqref{eq:realbnd2}.
The left-hand side of \eqref{eq:realbnd} is at most 
\[
   \sums[a \le y^{2/3}][\text{squarefree}]  %\sumss[\flat][a \le y^{2/3}] 
    A^{\omega(a)}\frac{(D,a)}{a^2}
    \sum_{b^2\rad(b) > y/a}
     \frac{A^{\omega(b)}(D,\rad(b))}{b^2\rad(b)}
      +
       \sums[a > y^{2/3}][\text{squarefree}] %\sumss[\flat][a > y^{2/3}] 
        A^{\omega(a)}\frac{(D,a)}{a^2}
        \sum_{b \ge 1}
         \frac{A^{\omega(b)}}{b^2}.
\]
By \eqref{eq:omegasfbnda} and \eqref{eq:omegaradbnd}, the first 
double sum is 
\[
    \ll_A
     (1 + A)^{\omega(D)} y^{-2/3 + o(1)}  
      \sums[a \le y^{2/3}] [\text{squarefree}] %\sumss[\flat][a \le y^{2/3}] 
       A^{\omega(a)}\frac{(D,a)}{a^{4/3}}
       \ll_A
        (1 + A)^{2\omega(D)}\frac{y^{O(1/\log\log 3y)}}{y^{2/3}}.
\]
By \eqref{eq:omegasfbnd}, and since 
$\sum_{b \ge 1} (A^{\omega(b)}/b^2) \ll_A 1$, 
\[
 \sums[a > y^{2/3}][\text{squarefree}] %\sumss[\flat][a > y^{2/3}] 
  A^{\omega(a)}\frac{(D,a)}{a^2}
  \sum_{b \ge 1}
   \frac{A^{\omega(b)}}{b^2}
    \ll_A 
     (1 + A)^{\omega(D)}\frac{y^{O(1/\log\log 3y)}}{y^{2/3}}.
\]
Combining gives \eqref{eq:realbnd}.
The left-hand side of \eqref{eq:realbnd2} is at most
\[
 \sums[a \le y][\text{squarefree}] % \sumss[\flat][a \le y] 
  \frac{A^{\omega(a)}}{a}
   \sum_{b^2\rad(b) \le y} A^{\omega(b)};
\]
applying \eqref{eq:omegasfbnda} and \eqref{eq:smth13} gives 
\eqref{eq:realbnd2}.

We now prove our claim.
For \eqref{eq:omegasfbnd}, we first consider the case $D = 1$.
Note that   
\begin{equation}
 \label{eq:mert}
 \sums[n_1 \le y][\text{squarefree}] %\sumss[\flat][n_1 \le y] 
  \frac{(A - 1)^{\omega(n_1)}}{n_1}
   \le 
    \prod_{p \le y}
     \bigg(1 + \frac{A - 1}{p}\bigg)
      \le 
       \prod_{p \le y} 
        \bigg(1 + \frac{1}{p}\bigg)^{A - 1}
         \ll_A
          (\log 3y)^{A - 1},
\end{equation}
because $1 + 1/p < \e^{1/p}$ and 
$\sum_{p \le y} 1/p = \log\log 3y + O(1)$ Mertens' theorem.
Now, 
\[
 \sums[n > y][\text{squarefree}] %\sumss[\flat][n > y] 
  \frac{A^{\omega(n)}}{n^2}
  =
   \sums[n > y][\text{squarefree}] %\sumss[\flat][n > y]
    \frac{1}{n^2}
     \sum_{n_1 \mid n} (A - 1)^{\omega(n_1)}
      \le 
       \sums[n_1 \ge 1][\text{squarefree}] %\sumss[\flat][n_1 \ge 1]
        \frac{(A - 1)^{\omega(n_1)}}{n_1^2}
         \sums[m > y/n_1][\text{squarefree}] %\sumss[\flat][m > y/n_1]
          \frac{1}{m^2},
\]
the inner sum being $O(n_1/y)$ for $n_1 \le y$ and $O(1)$ for 
$n_1 > y$. 
Thus, 
\[
 \sums[n > y][\text{squarefree}]%\sumss[\flat][n > y] 
  \frac{A^{\omega(n)}}{n^2}
  \ll_A
   \frac{(\log 3y)^{A - 1}}{y}
    +
     \sums[n_1 > y][\text{squarefree}] %\sumss[\flat][n_1 > y]
      \frac{(A - 1)^{\omega(n_1)}}{n_1^2}.
\]
If $A \le 2$ then this last sum is $O(1/y)$; otherwise, 
repeating the argument as many times as necessary gives
\[
 \sums[n > y] [\text{squarefree}] % \sumss[\flat][n > y] 
  \frac{A^{\omega(n)}}{n^2}
  \ll_A
   \frac{(\log 3y)^{A - 1}}{y}.
\]
It follows that, for any integer $d \ge 1$, 
\[
  \sums[n > y, \, d \mid n][\text{squarefree}] %\sumss[\flat]
   \frac{A^{\omega(n)}}{n^2}
    \ll_A
     \frac{A^{\omega(d)}}{d}
      \cdot 
       \frac{(\log 3y)^{A - 1}}{y}.
\]
For any integer $D \ge 1$, we trivially have 
$(D,n) \le \sum_{d \mid D, \, d \mid n} d$, and hence
\[
 \sums[n > y][\text{squarefree}] %\sumss[\flat][n > y]
  A^{\omega(n)} \frac{(D,n)}{n^2}
   \le 
    \sums[d \mid D][\text{squarefree}] %\sumss[\flat][d \mid D] d 
     \sums[n > y, \, d \mid n][\text{squarefree}]%\sumss[\flat][n > y][d \mid n] 
      \frac{A^{\omega(n)}}{n^2} 
       \ll_A 
        \frac{(\log 3y)^{A - 1}}{y}
         \sums[d \mid D][\text{squarefree}] \frac{A^{\omega(d)}}{d}. %\sumss[\flat][d \mid D] A^{\omega(d)}.
\]
Since 
$
 \sums[d \mid D, \, \text{squarefree}] A^{\omega(d)} 
  = 
   (1 + A)^{\omega(D)}
$
and 
$(\log 3y)^{A - 1} \ll_A y^{O(1/\log\log 3y)}$, 
this gives \eqref{eq:omegasfbnd}.
The bound \eqref{eq:omegasfbnda} follows from \eqref{eq:mert} and 
$(D,n) \le \sum_{d \mid D, \, d \mid n} d$.

For \eqref{eq:omegaradbnd}, we use the following ancillary bound.
We have 
\begin{equation}
 \label{eq:omegabndanc2}
  \sums[n > y][\rad(n) = m] \frac{1}{n} 
   \ll
    \frac{y^{O(1/\log\log 3y)}}{y}, 
\end{equation}
uniformly for integers squarefree integers $m \ge 1$.
To establish \eqref{eq:omegabndanc2}, we use an estimate involving 
smooth numbers: for $y \ge z \ge 2$, let 
\[
 \Psi(y,z) \defeq \#\{n \le y : p \mid n \implies p \le z\} 
\]
denote the number of $z$-smooth positive integers $n \le y$.
The following can be found in \cite[(1.19)]{GRA:08a}: for 
$y \ge z \ge 2$,    
\begin{equation}
 \label{eq:smth}
  \log \Psi(y,z)
   =
    \bigg(\frac{\log y}{\log z}\bigg)
    g\bigg(\frac{z}{\log y}\bigg)
      \bigg(1 + O\bigg(\frac{1}{\log z} + \frac{1}{\log\log x}\bigg)\bigg),
\end{equation}
where $g(w) = \log(1 + w) + w\log(1 + 1/w) \le w + 1$ ($w > 0$).
Noting that 
\[
 \sums[n \ge 1][\rad(n) = m] 
  \frac{1}{n^{1/2}}
   =
    \frac{1}{m^{1/2}}
     \sums[n \ge 1][\rad(n) \mid m] \frac{1}{n^{1/2}}
      =
       \frac{1}{m^{1/2}}
        \prod_{p \mid m}
         \bigg(\sum_{a \ge 0} \frac{1}{p^{a/2}}\bigg)
         =
          \prod_{p \mid m} \bigg(\frac{1}{p^{1/2} - 1}\bigg),
\]
we see that 
\begin{equation}
 \label{eq:ngeqy2}
  \sums[n > y^2][\rad(n) = m] \frac{1}{n}
   \le 
    \sums[n > y^2][\rad(n) = m] 
     \frac{1}{n}\bigg(\frac{n}{y^{2}}\bigg)^{1/2}
      \le 
       \frac{1}{y}
        \sums[n \ge 1][\rad(n) = m] 
         \frac{1}{n^{1/2}}
          \ll
           \frac{1}{y}.
\end{equation}
If $m > y^2$ then 
$
 \sum_{n > y, \, \rad(n) = m} 1/n 
  = 
   \sum_{n > y^2, \, \rad(n) = m} 1/n
$, and we are done.
Let us assume, then, that $y^2 \ge m$.
Let $\ell_1,\ldots,\ell_r$ denote the prime divisors 
of $m$, and let $p_1 = 2 < p_2 = 3 < \cdots < p_r$ 
denote the $r$ smallest primes.
Note that 
$
 \#\{(\alpha_1,\ldots,\alpha_r) \in \NN^r :
       \ell_1^{\alpha_1}\cdots \ell_r^{\alpha_r} \le y^2\}
   \le 
    \#\{(\alpha_1,\ldots,\alpha_r) \in \NN^r :
       p_1^{\alpha_1}\cdots p_r^{\alpha_r} \le y^2\}
$, i.e.\ note that 
$
 \#\{n \le y^2 : \rad(n) = m\}
  \le 
   \#\{n \le y^2 : \rad(n) = p_1\cdots p_r\}
$.
Since $y^2 \ge m \ge p_1\cdots p_r$,  
we have $4\log y^2 \ge 4\log m \ge 4\log (p_1\cdots p_r) > p_r$ by 
one of Chebyshev's bounds for primes, so if 
$\rad(n) = p_1\cdots p_r$, then $n$ is $y$-smooth, where 
$y = 4\log y^2$.
Therefore, 
\begin{equation}
 \label{eq:ylenley2}
 \sums[y < n \le y^2][\rad(n) = m] \frac{1}{n}
  <
   \frac{1}{y}
    \sums[n \le y^2][\rad(n) = m] 1
     \le 
      \frac{1}{y}
       \sums[n \le y^2][\rad(n) = p_1\cdots p_r] 1
        \le 
         \frac{\Psi(y^2,4\log y^2)}{y}
          \ll 
           \frac{y^{O(1/\log\log 3y)}}{y},
\end{equation}
where the last bound follows, upon exponentiating, from 
\eqref{eq:smth}.
Combining \eqref{eq:ngeqy2} and \eqref{eq:ylenley2} gives 
\eqref{eq:omegabndanc2}.

The left-hand side of \eqref{eq:omegaradbnd} is at most
\[
     \sums[m \le y^{1/3}][\text{squarefree}] %\sumss[\flat][m \le y^{1/3}] 
      \frac{A^{\omega(m)}(D,m)}{m}
       \sums[n^2 > y^{2/3}][\rad(n) = m]
        \frac{1}{n^2}
       +
         \sums[m > y^{1/3}][\text{squarefree}]%\sumss[\flat][m > y^{1/3}] 
          \frac{A^{\omega(m)}(D,m)}{m}
           \sums[n \ge 1][\rad(n) = m] \frac{1}{n^2}.
\]
By \eqref{eq:omegasfbnda} and \eqref{eq:omegabndanc2} (note that 
$1/n^2 < 1/(y^{1/3}n)$ when $n^2 > y^{2/3}$), we have 
\[
  \sums[m \le y^{1/3}][\text{squarefree}] %\sumss[\flat][m \le y^{1/3}] 
   \frac{A^{\omega(m)}(D,m)}{m}
    \sums[n^2 > y^{2/3}][\rad(n) = m]
     \frac{1}{n^2}
      \ll_A
       (1 + A)^{\omega(m)}
        \frac{y^{O(1/\log\log 3y)}}{y^{2/3}};
\]
by \eqref{eq:omegasfbnd} (note that $1/m^3 < 1/(y^{1/3}m^2)$ when 
$m > y^{1/3}$), and since 
\[
  \sums[n \ge 1][\rad(n) = m] \frac{1}{n^2} 
   =
    \frac{1}{m^2}
     \sums[n \ge 1][\rad(n) \mid m] \frac{1}{n^2}
      =
       \frac{1}{m^2}
        \prod_{p \mid m}
         \bigg(\sum_{a \ge 0} \frac{1}{p^{2a}}\bigg)
          \ll
           \frac{1}{m^2},
\]
we have           
\[
  \sums[m > y^{1/3}][\text{squarefree}] %\sumss[\flat][m > y^{1/3}] 
   \frac{A^{\omega(m)}(D,m)}{m}
    \sums[n \ge 1][\rad(n) = m] \frac{1}{n^2}
     \ll
      \sums[m > y^{1/3}][\text{squarefree}] %\sumss[\flat][m > y^{1/3}]
       \frac{A^{\omega(m)}(D,m)}{m^3}
        \ll_A
         \frac{(1 + A)^{\omega(D)}}{y^{2/3}}.
\]
Combining gives \eqref{eq:omegaradbnd}.

For \eqref{eq:smth13}, we note that since 
$\rad(n)^3 \le n^2\rad(n)$ and 
$A^{\omega(n)} = A^{\omega(\rad(n))}$, 
\[
 \sum_{n^2\rad(n) \le y} 
  A^{\omega(n)} 
   \le 
    \sums[a \le y^{1/3}][\text{squarefree}] %\sumss[\flat][a \le y^{1/3}] 
     A^{\omega(a)}
     \sums[b^2 \le y][\rad(b) = a] 1.
\]
An argument similar to the one leading up to \eqref{eq:ylenley2} 
shows that, uniformly for $a \le y^{1/3}$, we have 
$\sum_{b^2 \le y, \, \rad(b) = a} 1 \ll y^{O(1/\log\log 3y)}$, and 
\[
 \sums[a \le y^{1/3}][\text{squarefree}] %\sumss[\flat][a \le y^{1/3}] 
  A^{\omega(a)} 
  \le 
   y^{1/3}
    \sums[a \le y^{1/3}][\text{squarefree}] %\sumss[\flat][a \le y^{1/3}] 
     \frac{A^{\omega(a)}}{a}
     \ll_A
      y^{1/3 + O(1/\log\log 3y)}
\]
by \eqref{eq:omegasfbnda}.
Combining gives \eqref{eq:smth13}.
\end{proof}

To prove Proposition \ref{prop:sssa}, we express 
$\mathfrak{S}_{\bh}$ as a series.
To this end, let us introduce some notation and establish some 
basic inequalities.
Let a nonempty, finite set $\bh \subseteq \ZZ$ be given, and let 
$k \defeq \card \bh$.
Recall that $\T_{\bh}(2^{\alpha})$ is defined (and nonempty when 
$\bh = \{0\}$) for $\alpha \ge 2$, and for $p \equiv 3 \bmod 4$, 
$\V_{\bh}(p^{\alpha})$ is defined (and nonempty when 
$\bh = \{0\}$) for $\alpha \ge 1$.
Let us set $\T_{\bh}(1) \defeq \{1\}$ and 
$\T_{\bh}(2) \defeq \{1,2\}$ for completeness.
For $p \not\equiv 1 \bmod 4$ and $\alpha \ge 1$, we may then 
define 
\begin{equation}
 \label{eq:defiota}
 \epsilon_{\bh}(p^{\alpha})
  \defeq 
   \bigg(\frac{\card \T_{\bz}(p^{\alpha})}{p^{\alpha}}\bigg)^{-k}
    \bigg(\frac{\card \T_{\bh}(p^{\alpha})}{p^{\alpha}}\bigg)
  -
     \bigg(\frac{\card \T_{\bz}(p^{\alpha - 1})}{p^{\alpha - 1}}\bigg)^{-k}
      \bigg(\frac{\card \T_{\bh}(p^{\alpha - 1})}{p^{\alpha - 1}}\bigg). 
\end{equation}
Note that $\epsilon_{\bh}(2^2) = 0$ by definition.

\begin{lemma}
 \label{lem:iotah}
Let $\bh$ be a nonempty, finite set of integers, and let 
$k \defeq \card \bh$.

\textup{(}a\textup{)}
For $p \equiv 3 \bmod 4$ and {\bfseries even} $\alpha \ge 2$, we 
have $\epsilon_{\bh}(p^{\alpha}) = 0$.

\textup{(}b\textup{)}
For $p \not\equiv 1 \bmod 4$, we have
\begin{equation}
 \label{eq:iotlem1}
  \epsilon_{\bh}(p)
   \ll_k
    \frac{(\det(\bh),p)}{p^2}.
\end{equation}

\textup{(}c\textup{)}
For $p \not\equiv 1 \bmod 4$ and $\alpha \ge 1$, we have 
\begin{equation}
 \label{eq:iotlem2}
  \epsilon_{\bh}(p^{\alpha})
   \ll_k 
    \frac{(\det(\bh),p)}{p^\alpha}.
\end{equation}

\textup{(}d\textup{)}
For $\beta \ge 1$, we have  
\begin{equation}
 \label{eq:iotlem3}
   \delta_{\bz}(2)^{-k}\delta_{\bh}(2)  
   =
    1 + \sum_{\alpha = 2}^{\beta} \epsilon_{\bh}(2^{\alpha})
      + O_k\bigg(\frac{1}{2^{\beta}}\bigg).
\end{equation}
For $p \equiv 3 \bmod 4$ and $\beta \ge 1$, we have 
\begin{equation}
 \label{eq:iotlem4}
   \delta_{\bz}(p)^{-k}\delta_{\bh}(p)  
   =
    1 + \sum_{\alpha = 1}^{\beta} \epsilon_{\bh}(p^{2\alpha - 1})
      + O_k\bigg(\frac{1}{p^{2\beta}}\bigg).
\end{equation}
\end{lemma}

\begin{proof}
(a)
Let $p \equiv 3 \bmod 4$ and let $\alpha \ge 1$.
As can be seen from Proposition \ref{prop:Sp3h}, 
\eqref{eq:delthp3} and part (c), we have 
\begin{equation}
 \label{eq:iotalempf1}
 \frac{\card \V_{\bz}(p^{\alpha})}{p^{\alpha}}
  =
   \bigg(1 + \frac{1}{p}\bigg)^{-1}
    \bigg(1 - \frac{1}{p^{\alpha + \alpha \bmod 2}}\bigg).
\end{equation}
For even $\alpha$ we therefore have  
\[
 \epsilon_{\bh}(p^{\alpha})
  =
   \bigg(1 + \frac{1}{p}\bigg)^k
    \bigg(1 - \frac{1}{p^{\alpha}}\bigg)^k
     \bigg(
      \frac{\card \V_{\bh}(p^{\alpha})}{p^{\alpha}}
     -
        \frac{\card \V_{\bh}(p^{\alpha - 1})}{p^{\alpha - 1}}
     \bigg),
\]
and as we noted following \eqref{eq:VW} and \eqref{eq:VWb}, 
$
 \card \V_{\bh}(p^{\alpha})/p^{\alpha} 
  - 
   \card \V_{\bh}(p^{\alpha - 1})/p^{\alpha - 1}
    = 
     0
$. 

(b)
Consider $p \equiv 3 \bmod 4$ (the case $p = 2$ is similar).
Let $\alpha \ge 1$.
Define $\eta_{\bh}(p^{\alpha})$ and $\kappa_{\bh}(p)$ as the 
numbers given by the relations
\begin{equation}
 \label{eq:defetakappa}
 \frac{\card \V_{\bh}(p^{\alpha})}{p^{\alpha}}
  \eqdef 
   \delta_{\bh}(p) + \eta_{\bh}(p^{\alpha}) 
    \quad 
     \text{and}
      \quad 
 \delta_{\bh}(p) 
  \eqdef
   \bigg(1 + \frac{1}{p}\bigg)^{-1}
    \bigg(1 - \frac{\kappa_{\bh}(p)}{p}\bigg).
\end{equation}
Note that by Proposition \ref{prop:Sp3h}, \eqref{eq:Vhdeltp3bnd} 
and part (c),  
$
 |\eta_{\bh}(p^{\alpha})| 
  < 
   k/p^{\alpha + (\alpha \bmod 2)}
$ and 
$\kappa_{\bh}(p) \le \min\{k - 1,p\}$, with 
$\kappa_{\bh}(p) = k - 1$ if $p \nmid \det(\bh)$.
Also, $\kappa_{\bh}(p) \ge -1$ (because $\delta_{\bh}(p) \le 1$).
Since $\alpha + (\alpha \bmod 2) \ge 2$, we have 
\[
 \frac{\card \V_{\bh}(p^{\alpha})}{p^{\alpha}}
  =
   \bigg(1 + \frac{1}{p}\bigg)^{-1}
    \bigg(1 - \frac{\kappa_{\bh}(p)}{p} + O\bigg(\frac{k}{p^2}\bigg)\bigg).
\]
In the special case $\bh = \{0\}$ we can take 
$\kappa_{\bh}(p) = 0$.
We therefore have  
\begin{align*}
 \bigg(\frac{\card \V_{\bz}(p^{\alpha})}{p^{\alpha}}\bigg)^{-k}
  \frac{\card \V_{\bh}(p^{\alpha})}{p^{\alpha}}
  &
   =
    \bigg(1 + \frac{1}{p}\bigg)^{k - 1}
       \bigg(1 - \frac{\kappa_{\bh}(p)}{p} + O_k\bigg(\frac{1}{p^2}\bigg)\bigg)
   \\
  &
   =
     \bigg(1 + \frac{k - 1}{p} +O_k\bigg(\frac{1}{p^2}\bigg)\bigg)
      \bigg(1 - \frac{\kappa_{\bh}(p)}{p} + O_k\bigg(\frac{1}{p^2}\bigg)\bigg)
 \\
   & 
    = 
      1 + \frac{k - 1 - \kappa_{\bh}(p)}{p} + O_k\bigg(\frac{1}{p^2}\bigg).
\end{align*}

Writing $\xi_{\bh}(p) \defeq k - 1 - \kappa_{\bh}(p)$, we have
\[
%  \bigg|
  \bigg(\frac{\card \V_{\bz}(p^{\alpha})}{p^{\alpha}}\bigg)^{-k}
   \frac{\card \V_{\bh}(p^{\alpha})}{p^{\alpha}}
  -
    1
%  \bigg|
  \ll_k 
%    \frac{k - 1 - \kappa_{\bh}(p)}{p} + \frac{A_k}{p^2}
%     =
     \frac{\xi_{\bh}(p)}{p} + \frac{1}{p^2}.
\]
If $p \mid \det(\bh)$ then 
$\xi_{\bh}(p)/p = \xi_{\bh}(p)(\det(\bh),p)/p^2$, and if 
$p \nmid \det(\bh)$ then, as already noted, 
$\kappa_{\bh}(p) = k - 1$, i.e.\ $\xi_{\bh}(p) = 0$, so 
$\xi_{\bh}(p)/p = \xi_{\bh}(p)(\det(\bh),p)/p^2$ in any case.
Since, as already noted, 
$-1 \le \kappa_{\bh}(p) \le k - 1$, 
we have $0 \le \xi_{\bh}(p) \le k$.
Thus, 
\[
%  \bigg|
  \bigg(\frac{\card \V_{\bz}(p^{\alpha})}{p^{\alpha}}\bigg)^{-k}
   \frac{\card \V_{\bh}(p^{\alpha})}{p^{\alpha}}
  -
    1
%  \bigg|
  \ll_k 
   \frac{(\det(\bh),p)}{p^2} + \frac{1}{p^2} 
    \ll
     \frac{(\det(\bh),p)}{p^2}.
\]
For $\alpha = 1$, the left-hand side is equal to $\epsilon_{\bh}(p)$ 
(see \eqref{eq:defiota}), so this gives \eqref{eq:iotlem1}.

(c)
Consider $p \equiv 3 \bmod 4$ (the case $p = 2$ is similar).
Let $\alpha \ge 1$.
By (a) and (b), the result holds for $\alpha = 1$ and 
$\alpha \ge 2$ even, so we may assume that $\alpha \ge 3$ is odd.
In that case, using \eqref{eq:iotalempf1} in the definition 
\eqref{eq:defiota} of $\epsilon_{\bh}(p^{\alpha})$, we see that 
\begin{align*}
 \epsilon_{\bh}(p^{\alpha})
 & 
  =
   \bigg(1 + \frac{1}{p}\bigg)^k
    \bigg\{
     \bigg(1 - \frac{1}{p^{\alpha + 1}}\bigg)^{-k}
      \frac{\card \V_{\bh}(p^{\alpha})}{p^{\alpha}}
       -
        \bigg(1 - \frac{1}{p^{\alpha - 1}}\bigg)^{-k}
         \frac{\card \V_{\bh}(p^{\alpha - 1})}{p^{\alpha - 1}}
    \bigg\}
 \\
 & 
  =
   \bigg(1 + \frac{1}{p}\bigg)^k
    \bigg\{
     \frac{\card \V_{\bh}(p^{\alpha})}{p^{\alpha}}
       -
       \frac{\card \V_{\bh}(p^{\alpha - 1})}{p^{\alpha - 1}}
        +
         O_k\bigg(\frac{1}{p^{\alpha - 1}}\bigg) 
    \bigg\},   
\end{align*}
since, for any $\alpha \ge 1$, 
$(1 - 1/p^{\alpha})^{-k} = 1 + O_k(1/p^{\alpha})$ and 
$\card \V_{\bh}(p^{\alpha})/p^{\alpha} = O(1)$.
We deduce, from \eqref{eq:VW} and \eqref{eq:VWb}, that  
$\epsilon_{\bh}(p^{\alpha}) \ll_k 1/p^{\alpha - 1}$, which is   
\eqref{eq:iotlem2} in the case $p \mid \det(\bh)$.

Now consider the case $p \nmid \det(\bh)$.
Note that, by Proposition \ref{prop:Sp3h}, \eqref{eq:Vhdeltp3bnd} 
and part (c), we have, for any $\alpha \ge 1$,   
\[
 \frac{\card \V_{\bh}(p^{\alpha})}{p^{\alpha}}
  =
   \bigg(1 + \frac{1}{p}\bigg)^{-1}
    \bigg(1 - \frac{k - 1}{p} - \frac{k}{p^{\alpha + \alpha \bmod 2}}\bigg). 
\]
In view of this and (the special case) \eqref{eq:iotalempf1}, we 
have, for odd $\alpha \ge 3$,
\begin{align*}
 &  
 \epsilon_{\bh}(p^{\alpha})
  =
   \bigg(1 + \frac{1}{p}\bigg)^{k - 1}
    \bigg\{
     \bigg(1 - \frac{1}{p^{\alpha + 1}}\bigg)^{-k}
      \bigg(1 - \frac{k - 1}{p} - \frac{k}{p^{\alpha + 1}}\bigg)
 \\
 & \hspace{180pt}
       -
        \bigg(1 - \frac{1}{p^{\alpha - 1}}\bigg)^{-k}
         \bigg(1 - \frac{k - 1}{p} - \frac{k}{p^{\alpha - 1}}\bigg)
     \bigg\}.
\end{align*}
Since 
$
 (1 - 1/p^{\alpha + 1})^{-k} 
  = 
   1 + k/p^{\alpha + 1} + O_k(1/p^{\alpha + 2})
$,
we have 
\[
 \bigg(1 - \frac{1}{p^{\alpha + 1}}\bigg)^{-k}
  \bigg(1 - \frac{k - 1}{p} - \frac{k}{p^{\alpha + 1}}\bigg)
   =
    1 - \frac{k - 1}{p} + O_k\bigg(\frac{1}{p^{\alpha + 2}}\bigg);
\]
similarly, 
\[
 \bigg(1 - \frac{1}{p^{\alpha - 1}}\bigg)^{-k}
  \bigg(1 - \frac{k - 1}{p} - \frac{k}{p^{\alpha - 1}}\bigg)
   =
    1 - \frac{k - 1}{p} + O_k\bigg(\frac{1}{p^{\alpha}}\bigg).
\]
Combining gives $\epsilon_{\bh}(p^{\alpha}) \ll_k 1/p^{\alpha}$, 
i.e.\ \eqref{eq:iotlem2}, for odd $\alpha \ge 3$.

(d)
Consider $p \equiv 3 \bmod 4$ (the case $p = 2$ is similar).
Let $\beta \ge 1$.
We have
\[
 1 + \sum_{\alpha = 1}^{\beta} \epsilon_{\bh}(p^{2\alpha - 1})
  =
   1 + \sum_{\alpha = 1}^{2\beta} \epsilon_{\bh}(p^{\alpha})
  =
    \bigg(\frac{\card \V_{\bz}(p^{2\beta})}{p^{2\beta}}\bigg)^{-k}
     \bigg(\frac{\card \V_{\bh}(p^{2\beta})}{p^{2\beta}}\bigg),
\]
because $\epsilon_{\bh}(p^{\alpha}) = 0$ for $\alpha$ even (by 
(a)), and the middle sum telescopes.
Now, Proposition \ref{prop:Sp3h} (c) gives
$\delta_{\bz}(p)^{-k} = (1 + 1/p)^k$, and by definition of 
$\eta_{\bh}(p^{2\beta})$ (see \eqref{eq:defetakappa}), 
$
 \delta_{\bh}(p) 
  =
   \big(\card \V_{\bh}(p^{2\beta})/p^{2\beta}\big) - \eta_{\bh}(p^{2\beta})
$.
With these substitutions, and \eqref{eq:iotalempf1}, we verify 
that  
\begin{align*}
 &   
  \delta_{\bz}(p)^{-k}\delta_{\bh}(p) 
 - 
    \bigg(\frac{\card \V_{\bz}(p^{2\beta})}{p^{2\beta}}\bigg)^{-k}
     \bigg(\frac{\card \V_{\bh}(p^{2\beta})}{p^{2\beta}}\bigg)
%   \\ 
%  & \hspace{15pt}
%   =   
%     \bigg(1 + \frac{1}{p}\bigg)^{k}
%      \bigg(\frac{\card \V_{\bh}(p^{2\beta})}{p^{2\beta}} - \eta_{\bh}(p^{2\beta})\bigg)
%       -
%        \bigg(1 + \frac{1}{p}\bigg)^{k}
%         \bigg(1 - \frac{1}{p^{2\beta}}\bigg)^{-k}
%          \frac{\card \V_{\bh}(p^{2\beta})}{p^{2\beta}}
  \\
 & \hspace{30pt}
  =   
    \frac{\card \V_{\bh}(p^{2\beta})}{p^{2\beta}}
     \bigg(1 + \frac{1}{p}\bigg)^{k}
      \bigg(1 - \bigg(1 - \frac{1}{p^{2\beta}}\bigg)^{-k} - \eta_{\bh}(p^{2\beta})\bigg).
\end{align*}
Now, $\card \V_{\bh}(p^{2\beta})/p^{2\beta} \le 1$, 
$(1 + 1/p)^{k} \ll_k 1$, 
$(1 - 1/p^{2\beta})^{-k} = 1 + O_k(1/p^{2\beta})$, and as noted in 
(b), Proposition \ref{prop:Sp3h}, \eqref{eq:Vhdeltp3bnd} and 
part (c) show that $|\eta_{\bh}(p^{2\beta})| < k/p^{2\beta}$.
Combining gives \eqref{eq:iotlem4}.
\end{proof}

For $n \in \NN$ such that $p \mid n$ implies 
$p \not\equiv 1 \bmod 4$, we extend \eqref{eq:defiota} by defining 
\[
 \epsilon_{\bh}(n)
  \defeq 
   \prod_{p^{\alpha} \emid \, n \,\,} \epsilon_{\bh}(p^{\alpha}).
\]
%
% (Note that $\epsilon_{\bh}(1) \defeq 1$ by convention.)
%
For such $n$, Lemma \ref{lem:iotah} (b) and (c) give
\begin{equation}
 \label{eq:iotad}
  |\epsilon_{\bh}(n)| 
   \le 
    A_k^{\omega(n)}\frac{(\det(\bh),\rad(n))}{n\sqfr(n)},   
\end{equation}
provided $A_k$ is sufficiently large in terms of $k$.
Since $\epsilon_{\bh}(2) = 0$ by definition, and by 
Lemma \ref{lem:iotah} (a), $\epsilon_{\bh}(n) = 0$ if either 
$\nu_2(n) = 1$ or $\nu_p(n)$ is even (and nonzero) for some 
$p \equiv 3 \bmod 4$.
Letting 
$
 \cN_1
  \defeq 
   \{n \in \cN : p \mid n \implies p \not\equiv 1 \bmod 4\}
$, 
where $\cN$ is as in \eqref{eq:defcN}, we define 
\begin{equation}
 \label{eq:defcD}
  \cD \defeq \cN_1 \cup \{2n : n \in \cN_1, 2 \mid n\}. 
\end{equation}
Thus, 
\[
  \cD 
   =
    \big\{2^{\alpha}p_1^{2\alpha_1 - 1}\cdots p_r^{2\alpha_r - 1} : \alpha \ge 0, \alpha \ne 1, r,\alpha_i \ge 1, p_i \equiv 3 \bmod 4 \text{ ($i \le r$)}\big\},
\]
and $\epsilon_{\bh}(n) = 0$ unless $n \in \cD$.
By definition \eqref{eq:defsssP} and Lemma \ref{lem:iotah} (d), 
\begin{equation}
 \label{eq:ssassum}
 \mathfrak{S}_{\bh}
  =
%  \prod_{p \not\equiv 1 \bmod 4}
%   \delta_{\bz}(p)^{-k}\delta_{\bh}(p)
%    =
    \bigg(1 + \sum_{\alpha \ge 2} \epsilon_{\bh}(2^{\alpha})\bigg)
     \prod_{p \not\equiv 1 \bmod 4}
      \bigg(1 + \sum_{\alpha \ge 1} \epsilon_{\bh}(p^{2\alpha - 1})\bigg)
   =
        1 + \sum_{d \in \cD} \epsilon_{\bh}(d),
\end{equation}
the last sum being absolutely convergent in view of 
Lemma \ref{lem:omegabnd} and \eqref{eq:iotad}.

For the purposes of stating and proving the next lemma, we define 
{\small 
% \begin{equation}
%  \label{eq:defiotaj}
\[
 \epsilon_{\bh}(p^{\alpha};j)
  \defeq 
   \bigg(\frac{\card \T_{\bz}(p^{\alpha})}{p^{\alpha}}\bigg)^{-j}
    \bigg(\frac{\card \T_{\bh}(p^{\alpha})}{p^{\alpha}}\bigg)
  -
     \bigg(\frac{\card \T_{\bz}(p^{\alpha - 1})}{p^{\alpha - 1}}\bigg)^{-j}
      \bigg(\frac{\card \T_{\bh}(p^{\alpha - 1})}{p^{\alpha - 1}}\bigg), 
% \end{equation}
\]
}

\noindent 
for $p \not\equiv 1 \bmod 4$, $\alpha \ge 1$, and $j \ge 1$; we 
then set 
$
 \epsilon_{\bh}(n;j) 
  \defeq \prod_{p^{\alpha} \emid n} 
   \epsilon_{\bh}(p^{\alpha};j)
$
for $n$ composed of primes $p \not\equiv 1 \bmod 4$.
Thus, $\epsilon_{\bh}(n) = \epsilon_{\bh}(n;j)$ when $j = \card \bh$.

\begin{lemma}
 \label{lem:cancel}
Set $\bo \defeq \emptyset$, or set $\bo \defeq \{0\}$.
Let $n \ge 2$ be such that $p \mid n$ implies 
$p \not\equiv 1 \bmod 4$, and let $R_1,\ldots,R_k$ be complete 
residue systems modulo $n$.
We have 
\[
 \sum_{h_1 \in R_1} 
  \cdots 
   \sum_{h_k \in R_k}
    \epsilon_{\bo \cup \bh}(n; \ocard \bo + k)
      =
       0,
\]
where $\bh = \{h_1,\ldots,h_k\}$ in the summand. 
\textup{(}Note that we may have $\card \bh < k$ here.\textup{)}
\end{lemma}

\begin{proof}
Let $p \not\equiv 1 \bmod 4$, $\alpha \ge 1$.
Suppose $\bh = \{h_1,\ldots,h_k\}$ and 
$\bh' = \{h_1',\ldots,h_k'\}$  
satisfy $h_i \equiv h_i' \bmod p^{\alpha}$, and hence 
$h_i \equiv h_i' \bmod p^{\alpha - 1}$ as well, for 
$i = 1,\ldots,k$.
For $p \equiv 3 \bmod 4$, it is clear from \eqref{eq:defVh} that
$
 \card \V_{\bo \cup \bh}(p^{\beta}) 
  = 
   \card \V_{\bo \cup \bh'}(p^{\beta})
$ 
for $\beta = \alpha$, and for $\beta = \alpha - 1$ as well.
Thus, 
$
 \epsilon_{\bo \cup \bh}(p^{\alpha}; \ocard \bo + k)
 = 
  \epsilon_{\bo \cup \bh'}(p^{\alpha}; \ocard \bo + k)
$.
Similarly, we have   
$
 \epsilon_{\bo \cup \bh}(2^{\alpha}; \ocard \bo + k) 
  = 
   \epsilon_{\bo \cup \bh'}(2^{\alpha}; \ocard \bo + k)
$
(see \eqref{eq:defTh}).
Therefore, by the Chinese remainder theorem, 
\[
 \sum_{h_1 \in R_1} 
  \cdots 
   \sum_{h_k \in R_k}
    \epsilon_{\bo \cup \bh}(n;\ocard \bo + k)
      =
       \prod_{p^{\alpha} \emid \, n \,} 
        \bigg(
         \sum_{h_1 \in \ZZ_{p^{\alpha}}}
          \cdots 
           \sum_{h_k \in \ZZ_{p^{\alpha}}}
            \epsilon_{\bo \cup \bh}(p^{\alpha};\ocard \bo + k) 
         \bigg),
\]
where $\bh = \{h_1,\ldots,h_k\}$ in both summands, and 
$\ZZ_{p^{\alpha}} \defeq \{0,\ldots,p^{\alpha} - 1\}$.
It therefore suffices to show that  
\begin{equation}
 \label{eq:cancelp}
  \sum_{h_1 \in \ZZ_{p^{\alpha}}}
   \cdots 
    \sum_{h_k \in \ZZ_{p^{\alpha}}}
     \epsilon_{\bo \cup \bh}(p^{\alpha};\ocard \bo + k) 
      =
       0
\end{equation}
for all $p \not\equiv 1 \bmod 4$ and $\alpha \ge 1$.

Consider the case $\bo = \emptyset$.
For $p \equiv 3 \bmod 4$ and $\alpha \ge 1$, we have 
\[
 \sum_{h_1 \in \ZZ_{p^{\alpha}}}
  \cdots 
   \sum_{h_k \in \ZZ_{p^{\alpha}}}
    \card \V_{\bh}(p^{\alpha})
   =
   \sum_{a \in \ZZ_{p^{\alpha}}}
    \sums[h_1 \in \ZZ_{p^{\alpha}}][a + h_1 \in S_p][\nu_p(a + h_1) < \alpha]
     \cdots 
      \sums[h_k \in \ZZ_{p^{\alpha}}][a + h_k \in S_p][\nu_p(a + h_k) < \alpha] 1,
\]
as can be seen by applying the definition \eqref{eq:defVh} of 
$\V_{\bh}(p^{\alpha})$ and changing the order of summation.
For $i = 1,\ldots,k$, each sum over $h_i$ on the right-hand side 
enumerates a translation of $\V_{\bz}(p^{\alpha})$, so the 
entire sum (i.e.\ the left-hand side) is equal to 
$p^{\alpha}(\card \V_{\bz}(p^{\alpha}))^k$.
Whence 
\[
 \sum_{h_1 \in \ZZ_{p^{\alpha}}}
  \cdots 
   \sum_{h_k \in \ZZ_{p^{\alpha}}}
    \bigg(\frac{\card \V_{\bz}(p^{\alpha})}{p^{\alpha}}\bigg)^{-k}
     \bigg(\frac{\card \V_{\bh}(p^{\alpha})}{p^{\alpha}}\bigg)
      =
       p^{k\alpha}.
\]
%
%****************************************************************%
%************************* START DETAIL *************************%
%****************************************************************%
%
\begin{nixnix}
\begin{align*}
  \sum_{h_1 \in \ZZ_{p^{\alpha}}}
   \cdots 
    \sum_{h_k \in \ZZ_{p^{\alpha}}}
     \card \V_{\bh}(p^{\alpha})
  & = 
      \sum_{h_1 \in \ZZ_{p^{\alpha}}}
       \cdots 
        \sum_{h_k \in \ZZ_{p^{\alpha}}}
         \sums[a \in \ZZ_{p^{\alpha}}][\forall i, a + h_i \in S_p][\forall i, \nu_p(a + h_i) < \alpha] 1
 \\ 
  & = 
   \sum_{a \in \ZZ_{p^{\alpha}}}
    \sums[h_1 \in \ZZ_{p^{\alpha}}][a + h_1 \in S_p][\nu_p(a + h_1) < \alpha]
     \cdots 
      \sums[h_k \in \ZZ_{p^{\alpha}}][a + h_k \in S_p][\nu_p(a + h_k) < \alpha] 1
 \\
 & = 
   \sum_{a \in \ZZ_{p^{\alpha}}}
    \sums[a + h_1 \in \ZZ_{p^{\alpha}}][a + h_1 \in S_p][\nu_p(a + h_1) < \alpha]
     \cdots 
      \sums[a + h_k \in \ZZ_{p^{\alpha}}][a + h_k \in S_p][\nu_p(a + h_k) < \alpha] 1   
 \\
 & = 
    \sum_{a \in \ZZ_{p^{\alpha}}}
     (\card \V_{\bz}(p^{\alpha}))
      \cdots 
       (\card \V_{\bz}(p^{\alpha}))
 \\
 & = 
    p^{\alpha}(\card \V_{\bz}(p^{\alpha}))^k
\end{align*}
\end{nixnix}
%
%****************************************************************%
%************************** END DETAIL **************************%
%****************************************************************%
%
Since 
\[
 \sum_{h_1 \in \ZZ_{p^{\alpha}}}
  \cdots 
   \sum_{h_k \in \ZZ_{p^{\alpha}}}
    \card \V_{\bh}(p^{\alpha - 1})
      =
      p^k
       \sum_{h_1 \in \ZZ_{p^{\alpha - 1}}}
        \cdots 
         \sum_{h_k \in \ZZ_{p^{\alpha - 1}}}
          \card \V_{\bh}(p^{\alpha - 1}),
\]
we similarly have 
\[
 \sum_{h_1 \in \ZZ_{p^{\alpha}}}
  \cdots 
   \sum_{h_k \in \ZZ_{p^{\alpha}}}
    \bigg(\frac{\card \V_{\bz}(p^{\alpha - 1})}{p^{\alpha - 1}}\bigg)^{-k}
     \bigg(\frac{\card \V_{\bh}(p^{\alpha - 1})}{p^{\alpha - 1}}\bigg)
      =
       p^kp^{k(\alpha - 1)} 
        =
         p^{k\alpha}.
\]
Subtracting gives \eqref{eq:cancelp} for $\alpha \ge 1$.
In a similar fashion, we obtain \eqref{eq:cancelp} in the case 
$\bo = \{0\}$.
An analogous argument gives the same results for $p = 2$.
%
%****************************************************************%
%************************* START DETAIL *************************%
%****************************************************************%
%
\begin{nixnix}
\[
 \sum_{h_1 \in \ZZ_{p^{\alpha}}}
  \cdots 
   \sum_{h_k \in \ZZ_{p^{\alpha}}}
    \card \V_{\{0\} \cup \bh}(p^{\alpha})
   =
   \sums[a \in \ZZ_{p^{\alpha}}][a \in S_p][\nu_p(a) < \alpha]
    \sums[h_1 \in \ZZ_{p^{\alpha}}][a + h_1 \in S_p][\nu_p(a + h_1) < \alpha]
     \cdots 
      \sums[h_k \in \ZZ_{p^{\alpha}}][a + h_k \in S_p][\nu_p(a + h_k) < \alpha] 1
       =
        \big(\card \V_{\bz}(p^{\alpha})\big)^{1 + k}.
\]
\end{nixnix}
%
%****************************************************************%
%************************** END DETAIL **************************%
%****************************************************************%
%
%****************************************************************%
%************************* START DETAIL *************************%
%****************************************************************%
%
\begin{nixnix}
Applying the definition \eqref{eq:defTh} of 
$\T_{\bh}(2^{\alpha + 1})$ and changing the order of summation 
yields 
\begin{align*}
  \sum_{h_1 \in \ZZ_{2^{\alpha + 1}}}
   \cdots 
    \sum_{h_k \in \ZZ_{2^{\alpha + 1}}}
     \card \T_{\bh}(2^{\alpha + 1})
  & = 
      \sum_{h_1 \in \ZZ_{2^{\alpha + 1}}}
       \cdots 
        \sum_{h_k \in \ZZ_{2^{\alpha + 1}}}
         \sums[a \in \ZZ_{2^{\alpha + 1}}][\forall i, a + h_i \in S_2][\forall i, \nu_2(a + h_i) < \alpha] 1
 \\ 
  & = 
   \sum_{a \in \ZZ_{2^{\alpha + 1}}}
    \sums[h_1 \in \ZZ_{2^{\alpha + 1}}][a + h_1 \in S_2][\nu_2(a + h_1) < \alpha]
     \cdots 
      \sums[h_k \in \ZZ_{2^{\alpha + 1}}][a + h_k \in S_2][\nu_2(a + h_k) < \alpha] 1
 \\
 & = 
   \sum_{a \in \ZZ_{2^{\alpha + 1}}}
    \sums[a + h_1 \in \ZZ_{2^{\alpha + 1}}][a + h_1 \in S_2][\nu_2(a + h_1) < \alpha]
     \cdots 
      \sums[a + h_k \in \ZZ_{2^{\alpha + 1}}][a + h_k \in S_2][\nu_2(a + h_k) < \alpha] 1   
 \\
 & = 
    \sum_{a \in \ZZ_{2^{\alpha + 1}}}
     (\card \T_{\bz}(2^{\alpha + 1}))
      \cdots 
       (\card \T_{\bz}(2^{\alpha + 1}))
 \\
 & = 
    2^{\alpha + 1}(\card \T_{\bz}(2^{\alpha + 1}))^k.
\end{align*}
\end{nixnix}
%
%****************************************************************%
%************************** END DETAIL **************************%
%****************************************************************%
\end{proof}

In the proof of Proposition \ref{prop:sssa}, we also make use of 
basic lattice point counting arguments, as in the final two 
lemmas below.

\begin{lemma}
 \label{lem:dethap}
Let $\cD$ be as in \eqref{eq:defcD}.
Set $\bo \defeq \emptyset$, or set $\bo \defeq \{0\}$.
Fix an integer $k \ge 1$, and a number $M_k \ge 1$ that depends on 
$k$ only.
Also, fix $B \ge 1$.
For $y \ge 1$, we have
\begin{equation}
 \label{eq:lemdetbnd}
  \sums[d \in \cD][d > y] 
   \frac{M_k^{\omega(d)}}{d\sqfr(d)}
    \sum_{0 < h_1 < \cdots < h_k \le By}
    (\det(\bo \cup \bh),\rad(d))
      \ll_{k,B}
       y^{k - 2/3 + O(1/\log\log 3y)},
\end{equation}
where $\bh = \{h_1,\ldots,h_k\}$ in the summand.
\end{lemma}

\begin{proof}
Let $y \ge 1$.
Let us first show that, for any squarefree integer $c \ge 1$, 
\begin{equation}
 \label{eq:lemdethappf1}
 \underset{c \mid \det(\{0,h_1,\ldots,h_k\})}
  { 
   \sum_{0 < h_1 < \cdots < h_k \le By}
  } 1
  \le  
   k^{2\omega(c)}
    \bigg(\frac{(By)^k}{c} + O_k\big((By)^{k - 1}\big)\bigg).
\end{equation}
Let $h_0 = 0,h_1,\ldots,h_k$ be pairwise distinct integers and  
suppose that $c$ divides $\prod_{0 \le i < j \le k}(h_i - h_j)$.
Then, since $c$ is squarefree, there exist pairwise coprime 
positive integers $c_{ij}$ such that 
$c = \prod_{0 \le i < j \le k} c_{ij}$ and 
$c_{ij} \mid h_i - h_j$, $0 \le i < j \le k$.
Therefore, 
\[
 \underset{c \mid \det(\{h_0,h_1,\ldots,h_k\})}
  { 
   \sum_{0 < h_1 < \cdots < h_k \le By}
  } 1
   \le 
    \sums[c = c_{01}\cdots c_{(k-1)k}] 
     \hspace{5pt}
     \underset{0 \le i < j \le k - 1 \implies c_{ij} \mid h_i - h_j}
    { 
       \sum_{h_1 \in I_{By}}
        \sum_{h_2 \in I_{By}}
         \cdots 
          \sum_{h_{k - 1} \in I_{By}}
    } 
     \hspace{5pt}
      \sums[h_k \in I_{By}][0 \le i \le k - 1 \implies c_{ik} \mid h_i - h_k] 1,
\]
where on the right-hand side, the outermost sum is over all 
decompositions of $c$ as a product of $\binom{k + 1}{2}$ positive 
integers, and $I_{By} \defeq (0,By]$.

Consider the decomposition $c = c_{01}\cdots c_{(k - 1)k}$.
Let us define $c_{j} \defeq \prod_{i = 0}^{j - 1} c_{ij}$ for 
$j = 1,\ldots,k$.
Notice that $c = \prod_{j = 1}^k c_j$.
By the Chinese remainder theorem, the condition on $h_k$ in the 
innermost sum above is equivalent to $h_k$ being in some 
congruence class modulo $c_k$, uniquely determined by 
$h_0,h_1,\ldots,h_{k - 1}$.
The sum is therefore equal to $By/c_k + O(1)$.
Iterating this argument $k$ times we see that the inner sum over 
$h_1,\ldots,h_k$ is equal to  
\[
  \prod_{j = 1}^k
   \bigg(\frac{By}{c_j} + O(1)\bigg)
    =
     \frac{(By)^k}{c} + O_k((By)^{k - 1}).
\]
The bound \eqref{eq:lemdethappf1} follows by combining and noting 
that, since $c$ is squarefree, the number of ways of writing $c$ 
as a product of $\binom{k + 1}{2}$ positive integers is 
$\binom{k + 1}{2}^{\omega(c)}$, and that 
$\binom{k + 1}{2} \le k^2$.

For $\bh = \{h_1,\ldots,h_k\}$, with $h_1,\ldots,h_k$ pairwise 
distinct, nonzero integers, and any $d \in \NN$, we trivially have 
$
 (\det(\bo \cup \bh),\rad(d)) 
   \le 
    \sum_{c \mid \det(\{0,h_1,\ldots,h_k\}), \, \rad(d)} c
$.
If $h_1,\ldots,h_k \le By$ as well, then $p \mid c$ implies 
$p \le By$.
From this and \eqref{eq:lemdethappf1}, it follows that 
\[
 \sum_{0 < h_1 < \cdots < h_k \le By} 
  (\det(\bo \cup \bh,\rad(d))
   \ll_{k,B}
    y^k
     \sum_{c \mid \rad(d)} k^{2\omega(c)}
      +
        y^{k - 1}
         \sums[c \mid \rad(d)][p \mid c \implies p \le By]
           ck^{2\omega(c)},
\]
where $\bh = \{h_1,\ldots,h_k\}$ in the summand on the left.
Now, for $c \mid \rad(d)$ we have 
$k^{2\omega(c)} \le k^{2\omega(d)}$, and 
$\sum_{c \mid \rad(d)} 1 = 2^{\omega(d)}$.
Applying these bounds to the left-hand side of 
\eqref{eq:lemdetbnd}, we see that it is 
\begin{equation}
 \label{eq:lemdetpf1}
%   \sums[d \in \cD][d > y] 
%    \frac{A_k^{\omega(d)}}{d\sqfr(d)}
%     \sum_{0 < h_1 < \cdots < h_k \le By}
%     (\det(\bo \cup \bh),\rad(d))
      \ll_{k,B}
       y^k
        \sums[d \in \cD][d > y] 
         \frac{A_k^{\omega(d)}}{d\sqfr(d)}
        +
           y^{k - 1} 
            \sum_{d \in \cD}
             \frac{A_k^{\omega(d)}}{d\sqfr(d)}
              \sums[c \mid \rad(d)][p \mid c \implies p \le By] c, 
\end{equation}
where $A_k$, here and below, denotes a sufficiently large number 
depending on $k$, which may be a different number at each 
occurrence.

By definition \eqref{eq:defcD} of $\cD$, for every $d \in \cD$, we 
have $d = n$ or $d = 2n$ for some $n \in \cN$, where $\cN$ is 
as in \eqref{eq:defcN}. 
Therefore, as a direct consequence of Lemma \ref{lem:omegabnd}, we 
have 
\begin{equation}
 \label{eq:lemdetpf2}
 \sums[d \in \cD][d > y] 
  \frac{A_k^{\omega(d)}}{d\sqfr(d)}
   \ll_k 
    \frac{y^{O(1/\log\log 3y)}}{y^{2/3}}.
\end{equation}
More specifically, for every $d \in \cD$, we have 
$d = ab^2\rad(b)$ or $d = 2ab^2\rad(b)$ for some uniquely 
determined $a,b \in \NN$, where $a$ is squarefree and $(a,b) = 1$.
Furthermore, $d$ is not exactly divisible by $2$, and so we have 
$2 \nmid a$ in the case $d = ab^2\rad(b)$, while $2 \mid ab$ in 
the case $d = 2ab^2\rad(b)$.
In either case, we have the following:
$A_k^{\omega(d)} = A_k^{\omega(a)}A_k^{\omega(b)}$;  
$d\sqfr(d) = a^2b^2\rad(b)$ or 
$d\sqfr(d) = 2a^2b^2\rad(b)$; and $\rad(d) = a\rad(b)$.
Thus, if $c \mid \rad(d)$, then $c = c_1c_2$, where $c_1 \mid a$
and $c_2 \mid \rad(b)$.
Consequently,  
\[
 \sum_{d \in \cD}
  \frac{A_k^{\omega(d)}}{d\sqfr(d)}
   \sums[c \mid \rad(d)][p \mid c \implies p \le By] c
    \ll
     \sums[a \ge 1][\text{squarefree}] %\sumss[\flat][a \ge 1]
      \frac{A_k^{\omega(a)}}{a^2}
       \sum_{b \ge 1} \frac{A_k^{\omega(b)}}{b^2\rad(b)}
        \sums[c_1 \mid a][p \mid c_1 \implies p \le By] c_1
         \sums[c_2 \mid \rad(b)][p \mid c_2 \implies p \le By] c_2.
\]

Now,  
\[
 \sums[a \ge 1][\text{squarefree}] %\sumss[\flat][a \ge 1]
  \frac{A_k^{\omega(a)}}{a^2}
   \sums[c_1 \mid a][p \mid c_1 \implies p \le By] c_1
    \le 
     \sums[c_1 \ge 1][\text{squarefree}][p \mid c_1 \implies p \le By] %\sumss[\flat][c_1 \ge 1][p \mid c_1 \implies p \le By] 
      \frac{A_k^{\omega(c_1)}}{c_1}
       \sums[a_1 \ge 1][\text{squarefree}] %\sumss[\flat][a_1 \ge 1] 
        \frac{A_k^{\omega(a_1)}}{a_1^2}
         \ll_k
          \sums[c_1 \ge 1][\text{squarefree}][p \mid c_1 \implies p \le By] %\sumss[\flat][c_1 \ge 1][p \mid c_1 \implies p \le By] 
           \frac{A_k^{\omega(c_1)}}{c_1}; 
\]
as can be seen by writing $a = a_1c_1$ and changing order of 
summation; also 
\[
 \sums[c_1 \ge 1][\text{squarefree}][p \mid c_1 \implies p \le By] %\sumss[\flat][c_1 \ge 1][p \mid c_1 \implies p \le By] 
  \frac{A_k^{\omega(c_1)}}{c_1}
  \le 
   \prod_{p \le By}
    \bigg(1 + \frac{A_k}{p}\bigg)
%      \le 
%       \prod_{p \le By}
%        \bigg(1 + \frac{1}{p}\bigg)^{A_k}
        \ll_{k,B} (\log 3y)^{A_k}. 
\]
(See \eqref{eq:mert}.)
Next, note that since 
$
 \sum_{c_2 \mid \rad(b)} c_2 
  \le 
   \rad(b) \sum_{c_2 \mid \rad(b)} 1
    \le 
     2^{\omega(b)}\rad(b)
$, 
\[
 \sum_{b \ge 1} \frac{A_k^{\omega(b)}}{b^2\rad(b)}
  \sums[c_2 \mid \rad(b)][p \mid c_2 \implies p \le By] c_2 
   \le 
    \sum_{b \ge 1} \frac{A_k^{\omega(b)}}{b^2}
     \sum_{c_2 \mid \rad(b)} 1
      \le 
       \sum_{b \ge 1} \frac{A_k^{\omega(b)}}{b^2}
        \ll_k 1.
\]
%
%(again replacing $A_{k}$ by a bigger constant.) 
Combining all of this gives 
\begin{equation}
 \label{eq:lemdetpf3}
  \sum_{d \in \cD}
   \frac{A_k^{\omega(d)}}{d\sqfr(d)}
    \sums[c \mid \rad(d)][p \mid c \implies p \le By] c
     \ll_{k,B}
      (\log 3y)^{A_k}.
\end{equation}
Finally, we obtain \eqref{eq:lemdetbnd} by combining 
\eqref{eq:lemdetpf1} with \eqref{eq:lemdetpf2} and 
\eqref{eq:lemdetpf3}.
\end{proof}

\begin{lemma}
 \label{lem:lip}
Fix an integer $k \ge 1$ and a bounded convex set 
$\sC \subseteq \RR^k$.
For $y \ge 1$ we have 
$
 \#(y\sC \cap \ZZ^k)
  =
   y^k\vol(\sC) + O_{k,\sC}(y^{k - 1}).
$
\end{lemma}

\begin{proof}
This is a special case of \cite[pp.\ 128--129]{LAN:94}.
\end{proof}

\begin{proof}[Proof of Proposition \ref{prop:sssa}]
Fix an integer $k \ge 1$ and a bounded convex set 
$\sC \subseteq \Delta^k$, where 
$
 \Delta^k \defeq \{(x_1,\ldots,x_k) \in \RR^k : 0 < x_1 < \cdots < x_k\}
$ 
(see \eqref{eq:defsimplex}). 
Set $\bo \defeq \emptyset$ or set $\bo \defeq \{0\}$.
Let $y \ge 1$.
To ease notation throughout, let $\cH \defeq y\sC \cap \ZZ^k$, 
$\vbh = (h_1,\ldots,h_k)$, and $\bh = \{h_1,\ldots,h_k\}$.
Note that $0 < h_1 < \cdots < h_k \ll_{\sC} y$ for $\vbh \in \cH$.
Also, let $A_k$ stand for a sufficiently large number depending on 
$k$, which may be a different number at each occurrence.

In view of \eqref{eq:ssassum} we see, upon partitioning the sum 
over $d$ and changing order of summation, that  
\begin{equation}
 \label{eq:lemsssapf1}
 \sum_{\vbh \in \cH} \mathfrak{S}_{\bo \cup \bh}
  =
   \sum_{\vbh \in \cH} 1
    +
      \sums[d \in \cD][d \le y]
       \sum_{\vbh \in \cH} \epsilon_{\bo \cup \bh}(d)
      +
       \sums[d \in \cD][d > y]
        \sum_{\vbh \in \cH} \epsilon_{\bo \cup \bh}(d),
\end{equation}
with $\cD$ as defined in \eqref{eq:defcD}.
By Lemma \ref{lem:lip}, we have 
\begin{equation}
 \label{eq:volH}
  \sum_{\vbh \in \cH} 1
   =
    y^k \vol(\sC) + O_{k,\sC}(y^{k - 1}).
\end{equation}
By \eqref{eq:iotad} and Lemma \ref{lem:dethap}, we have 
\begin{equation}
 \label{eq:sum3}
 \sums[d \in \cD][d > y]
  \sum_{\vbh \in \cH} |\epsilon_{\bo \cup \bh}(d)|
   \le 
    \sums[d \in \cD][d > y]
     \sum_{\vbh \in \cH} A_k^{\omega(d)}\frac{(\det(\bo \cup \bh),\rad(d))}{d\sqfr(d)}
      \ll_{k,\sC} 
       y^{k - 1}\frac{y^{O(1/\log\log 3y)}}{y^{2/3}}.
\end{equation}

Consider the middle sum on the right-hand side of 
\eqref{eq:lemsssapf1}.
Let $d$ be any element of $\cD$ with $d \le y$, and partition 
$\RR^k$ into cubes 
\[
 C_{d,\vbt} 
  \defeq 
   \{(x_1,\ldots,x_k) \in \RR^k : t_id \le x_i < (t_i + 1)d, i = 1,\ldots,k\},
\]
with $\vbt \defeq (t_1,\ldots,t_k)$ running over $\ZZ^k$.
Each $\vbh \in \cH$ is a point in a unique cube of this form: we 
call $\vbh$ a {\em $d$-interior} point if this cube is entirely 
contained in $y\sC$, and $\vbh$ a {\em $d$-boundary} point if this 
cube has a nonempty intersection with the boundary of $y\sC$.
We partition $\cH$ into $d$-interior points and $d$-boundary 
points.
As $\vbh$ runs over all $d$-interior points of $\cH$, $h_i$ 
($i = 1,\ldots,k$) runs over a pairwise disjoint union of complete 
residue systems modulo $d$, none of which contain $0$.
By Lemma \ref{lem:cancel} (we have 
$\card(\bo \cup \bh) = \ocard \bo + k$ for each $\vbh \in \cH$), 
it follows that 
\begin{equation}
 \label{eq:2ndlast}
  \sums[d \in \cD][d \le y] 
   \sum_{\vbh \in \cH}
    \epsilon_{\bo \cup \bh}(d)
  =
   \sums[d \in \cD][d \le y]
    \sums[\vbh \in \cH][\text{$d$-boundary}]
     \epsilon_{\bo \cup \bh}(d).
\end{equation}

By \eqref{eq:iotad}, and the aforementioned trivial bound for 
$(\det(\bo \cup \bh),\rad(d))$,  
\begin{align*} 
 \sums[d \in \cD][d \le y]
  \sums[\vbh \in \cH][\text{$d$-boundary}]
  |\epsilon_{\bo \cup \bh}(d)|
 & 
    \le 
     \sums[d \in \cD][d \le y] \frac{A_k^{\omega(d)}}{d\sqfr(d)}
      \sums[\vbh \in \cH][\text{$d$-boundary}] 
      (\det(\bo \cup \bh),\rad(d))
 \\
 & 
         \le 
          \sums[d \in \cD][d \le y] \frac{A_k^{\omega(d)}}{d\sqfr(d)}
           \sum_{c \mid \rad(d)} c
            \sums[\vbh \in \cH][\text{$d$-boundary}][c \mid \det(\{0,h_1,\ldots,h_k\})] 1.
\end{align*}
For each $d \in \cD$ with $y/d \ge 1$, the proof of 
Lemma \ref{lem:lip} (see \cite[pp.\ 128--129]{LAN:94}) shows that 
there are $\ll_{k,\sC} (y/d)^{k - 1}$ cubes $C_{d,\vbt}$ that have 
a nonempty intersection with the boundary of $y\sC$.
For each such boundary cube $C_{d,\vbt}$, the corresponding 
$d$-boundary points are all in $C_{d,\vbt} \cap \ZZ^k$, which is a 
product of complete residue systems modulo $d$, and, given that 
$c \mid \rad(d)$ (and hence $c \mid d$), the condition 
$c \mid \det(\{0,h_1,\ldots,h_k\})$ is equivalent to 
$c \mid \det(\{0,h'_1,\ldots,h'_k\})$ when 
$h_i \equiv h'_i \bmod d$, $i = 1,\ldots,k$.

If follows that, for $d \in \cD$ with $d \le y$, and for 
$c \mid \rad(d)$, we have 
\[
 \sums[\vbh \in \cH][\text{$d$-boundary}][c \mid \det(\{0,h_1,\ldots,h_k\})] 1
  \ll_{k,\sC}
   \frac{y^{k - 1}}{d^{k - 1}}
    \sums[0 < h_1 < \cdots < h_k \le d][c \mid \det(\{0,h_1,\ldots,h_k\})] 1
     \ll_k 
      y^{k - 1} d \bigg(\frac{A_k^{\omega(c)}}{c}\bigg)
\]
by \eqref{eq:lemdethappf1}. 
Whence 
\[
 \sums[d \in \cD][d \le y]
  \sums[\vbh \in \cH][\text{$d$-boundary}]
  |\epsilon_{\bo \cup \bh}(d)|
   \ll_{k,\sC}
    y^{k - 1} 
      \sums[d \in \cD][d \le y] \frac{A_k^{\omega(d)}}{\sqfr(d)}
       \sum_{c \mid \rad(d)} A_k^{\omega(c)}
        \le 
         y^{k - 1} 
          \sums[d \in \cD][d \le y] \frac{A_k^{\omega(d)}}{\sqfr(d)},
\]
since $\sum_{c \mid \rad(d)} A_k^{\omega(c)}$ is at most 
$
  A_k^{\omega(d)} \sum_{c \mid \rad(d)} 1
   = 
    (2A_k)^{\omega(d)}
$.
By \eqref{eq:realbnd2}, this last sum is 
$\ll_k y^{1/3 + O(1/\log\log 3y)}$.
Combining, we obtain 
\begin{equation}
 \label{eq:last}
  \sums[d \in \cD][d \le y] 
   \sum_{\vbh \in \cH}
    \epsilon_{\bo \cup \bh}(d)
     \ll_{k,\sC}
      y^{k - 1}y^{1/3 + O(1/\log\log 3y)}.
\end{equation}
Combining \eqref{eq:lemsssapf1} with \eqref{eq:volH}, 
\eqref{eq:sum3}, and \eqref{eq:last} gives \eqref{eq:sssa}.
\end{proof}

%%%%%%%%%%%%%%%%%%%%%%%%%%%%%%%%%%%%%%%%%%%%%%%%%%%%%%%%%%%%%%%%%%
%%%%%%%%%%%%%%%%%%%%%%%%%%%% APPENDIX %%%%%%%%%%%%%%%%%%%%%%%%%%%%
%%%%%%%%%%%%%%%%%%%%%%%%%%%%%%%%%%%%%%%%%%%%%%%%%%%%%%%%%%%%%%%%%%

\begin{nix}
\appendix 

\section{Some elementary verifications}
 \label{sec:A1}
\end{nix}
 
%****************************************************************%
%************************* START DETAIL *************************%
%****************************************************************%
%
\begin{nixnix}
\begin{proof}[Deduction of Theorem \ref{thm:main} (b) in detail]
(b) 
To ease notation, we let $\vbi = (i_1,\ldots,i_r)$, 
$\vba = (a_1,\ldots,a_r)$, $\vbh = (h_1,\ldots,h_k)$,  
$\bh = \{h_1,\ldots,h_k\}$, and 
\[
 \speccount(\{0\} \cup \bh; x)
  \defeq 
   \sum_{n \le x} 
    \ind{\SS}(n)\ind{\SS}(n + h_1)\cdots \ind{\SS}(n + h_k).
\]
Given $\vbi,\vba \in \NN^r$, let 
\[
 N_{\vbi,\vba}(x)
  \defeq 
   \sums[0 < h_1 < \cdots < h_{i_1 + \cdots + i_r}]
        [h_{i_1 + \cdots + i_j} = a_j, \, j = 1,\ldots,r]
     \hspace{5pt}
      \sum_{n \le x}
       \ind{\SS}(n)\ind{\SS}(n + h_1)\cdots \ind{\SS}(n + h_{i_1 + \cdots + i_r}).
\]
Let $\ell \ge 0$ be an integer, arbitrarily large but fixed.
We claim that 
\begin{equation}
 \label{eq:inc-exc}
 \sum_{k = r}^{r + 2\ell + 1}
  (-1)^{k - r}
   \sum_{i_1 + \cdots + i_r = k}
    N_{\vbi,\vba}(x)
     \le 
      \sums[\sts_n \le x]
           [\sts_{n + j} - \sts_n = a_j]
           [j = 1,\ldots,r] 1
       \le
        \sum_{k = r}^{r + 2\ell}
         (-1)^{k - r}
          \sum_{i_1 + \cdots + i_r = k}
           N_{\vbi,\vba}(x) 
\end{equation}
for any $\vba \in \NN^r$ with $a_1 < \cdots < a_r$, the inner 
sums (here and below) are over all $\vbi \in \NN^r$ such that 
$i_1 + \cdots + i_r = k$.
Now, 
\begin{equation}
 \label{eq:inc-exc2}
  \sums[\sts_n \le x]
       [\sts_{n + j} - \sts_{n + j - 1} \le \lambda_j y]
       [j = 1,\ldots,r] 1
   =
     \sums[\vba \, \in \, \NN^r]
          [0 < a_j - a_{j - 1} \le \lambda_j y]
          [j = 1,\ldots,r]  
       \hspace{5pt} 
        \sums[\sts_n \le x]
             [\sts_{n + j} - \sts_n = a_j]
             [j = 1,\ldots,r] 1,
\end{equation}
where $a_0 \defeq 0$.
Given $\vbi \in \NN^r$ with $i_1 + \cdots + i_r = k$ we 
have, with $\Theta_{\vbi,\vbl}$ as in \eqref{eq:defThet}, 
\begin{equation}
 \label{eq:inc-exc3}
 \sums[\vba \, \in \, \NN^r]
      [0 < a_j - a_{j - 1} \le \lambda_j y]
      [j = 1,\ldots,r]
  N_{\vbi,\vba}(x)
   =
    \sum_{\vbh \, \in \, y\Theta_{\vbi,\vbl} \cap \, \ZZ^k}
%        \sum_{n \le x}
%         \ind{\SS}(n)\ind{\SS}(n + h_1)\cdots \ind{\SS}(n + h_k).
        \speccount(\{0\} \cup \bh; x).
\end{equation}
Combining \eqref{eq:inc-exc}, \eqref{eq:inc-exc2} and 
\eqref{eq:inc-exc3}, then changing order of summation, we obtain 
\begin{align}
 \begin{split}
  \label{eq:inc-exc4}
 & 
   \sum_{k = r}^{r + 2\ell + 1}
    (-1)^{k - r}
     \sum_{i_1 + \cdots + i_r = k}
      \sum_{\vbh \, \in \, y\Theta_{\vbi,\vbl} \cap \, \ZZ^k}
%        \sum_{n \le x}
%         \ind{\SS}(n)\ind{\SS}(n + h_1)\cdots \ind{\SS}(n + h_k)
        \speccount(\{0\} \cup \bh; x)
 \\
 & \hspace{30pt} 
  \le 
%     \#\{\sts_n \le x : \sts_{n + j} - \sts_{n + j - 1} \le \lambda_j y, j = 1,\ldots,r\}
     \sums[\sts_n \le x]
          [\sts_{n + j} - \sts_{n + j - 1} \le \lambda_j y]
          [j = 1,\ldots,r] 1 
%  \\
%  & \hspace{5pt}
   \le 
    \sum_{k = r}^{r + 2\ell}
     (-1)^{k - r}
      \sum_{i_1 + \cdots + i_r = k}
       \sum_{\vbh \, \in \, y\Theta_{\vbi,\vbl} \cap \, \ZZ^k}
%        \sum_{n \le x}
%         \ind{\SS}(n)\ind{\SS}(n + h_1)\cdots \ind{\SS}(n + h_k).
        \speccount(\{0\} \cup \bh; x).
  \end{split}            
\end{align}

The substitution \eqref{eq:defEterm}, with $\{0\} \cup \bh$ and 
$k + 1$ in place of $\bh$ and $k$, yields 
\begin{align}
 \begin{split}
  \label{eq:inc-exc5}
 & 
   \sum_{k = r}^{r + 2\ell + 1}
    (-1)^{k - r}
     \bigg(\frac{\speccount(x)}{x}\bigg)^k
     \sum_{i_1 + \cdots + i_r = k}
      \sum_{\vbh \, \in \, y\Theta_{\vbi,\vbl} \cap \, \ZZ^k}
       \bigg(\mathfrak{S}_{\{0\} \cup \bh} + \cE_{\{0\} \cup \bh}(x)\bigg)
 \\
 & \hspace{15pt} 
  \le 
   \frac{1}{\speccount(x)}
%     \#\{\sts_n \le x : \sts_{n + j} - \sts_{n + j - 1} \le \lambda_j y, j = 1,\ldots,r\}
     \sums[\sts_n \le x]
          [\sts_{n + j} - \sts_{n + j - 1} \le \lambda_j y]
          [j = 1,\ldots,r] 1    
 \\
 & \hspace{30pt}
   \le 
    \sum_{k = r}^{r + 2\ell}
     (-1)^{k - r}
      \bigg(\frac{\speccount(x)}{x}\bigg)^k
       \sum_{i_1 + \cdots + i_r = k}
        \sum_{\vbh \, \in \, y\Theta_{\vbi,\vbl} \cap \, \ZZ^k}
         \bigg(\mathfrak{S}_{\{0\} \cup \bh} + \cE_{\{0\} \cup \bh}(x)\bigg).
  \end{split}            
\end{align}
By applying Hypothesis ($k,\Theta_{\vbi,\vbl},\{0\}$) for all 
$k$ and $\vbi$ satisfying $r \le k \le r + 2\ell + 1$ and  
$i_1 + \cdots + i_r = k$, Proposition \ref{prop:sssa}, and our 
assumption that $y \sim x/\speccount(x)$ as $x \to \infty$, it is 
straightforward to deduce from \eqref{eq:inc-exc5} that 
\begin{align*}
%  \begin{split}
%   \label{eq:inc-exc6}
 & 
   \Big(1 + O_{r,\ell,\vbl}\big(\varepsilon(x)\big)\Big)
    \sum_{k = r}^{r + 2\ell + 1}
     (-1)^{k - r}
      \sum_{i_1 + \cdots + i_r = k}
       \vol(\Theta_{\vbi,\vbl})
 \\
 &  \hspace{15pt} 
      \le 
       \frac{1}{\speccount(x)}
%           \#\{\sts_n \le x : \sts_{n + j} - \sts_{n + j - 1} \le \lambda_j y, j = 1,\ldots,r\}
          \sums[\sts_n \le x]
               [\sts_{n + j} - \sts_{n + j - 1} \le \lambda_j y]
               [j = 1,\ldots,r] 1
%  \\
%  &  \hspace{60pt}
         \le 
          \Big(1 + O_{r,\ell,\vbl}\big(\varepsilon(x)\big)\Big) 
           \sum_{k = r}^{r + 2\ell}
           (-1)^{k - r}
            \sum_{i_1 + \cdots + i_r = k}
             \vol(\Theta_{\vbi,\vbl}), 
%   \end{split}             
\end{align*}
where $\varepsilon(x)$ is some function, not necessarily the same 
as in \eqref{eq:hyp}, such that $\varepsilon(x) \to 0$ as 
$x \to \infty$. 
Consequently, 
\begin{equation}
  \label{eq:inc-exc7}
   \sum_{k = r}^{r + 2\ell + 1}
     (-1)^{k - r}
      \sum_{i_1 + \cdots + i_r = k}
       \vol(\Theta_{\vbi,\vbl})
        \le 
         \liminf_{x \to \infty}
          \frac{1}{\speccount(x)}
%           \#\{\sts_n \le x : \sts_{n + j} - \sts_{n + j - 1} \le \lambda_j y, j = 1,\ldots,r\}
           \sums[\sts_n \le x]
                [\sts_{n + j} - \sts_{n + j - 1} \le \lambda_j y]
                [j = 1,\ldots,r] 1
\end{equation}
and 
\begin{equation}
 \label{eq:inc-exc8}
          \limsup_{x \to \infty}
           \frac{1}{\speccount(x)}
%            \#\{\sts_n \le x : \sts_{n + j} - \sts_{n + j - 1} \le \lambda_j y, j = 1,\ldots,r\}
            \sums[\sts_n \le x]
                 [\sts_{n + j} - \sts_{n + j - 1} \le \lambda_j y]
                 [j = 1,\ldots,r] 1
           \le 
            \sum_{k = r}^{r + 2\ell}
            (-1)^{k - r}
             \sum_{i_1 + \cdots + i_r = k}
              \vol(\Theta_{\vbi,\vbl}).            
\end{equation}
Since 
$
 \vol(\Theta_{\vbi,\vbl}) 
  = \lambda_1^{i_1}\cdots \lambda_r^{i_r}/(i_1!\cdots i_r!)
$, 
the sums on the left and right of \eqref{eq:inc-exc7} and 
\eqref{eq:inc-exc8} are truncations of the Taylor 
series for $(1 - \e^{-\lambda_1})\cdots (1 - \e^{-\lambda_r})$.
We have chosen $\ell$ arbitrarily large, so we may conclude that 
\eqref{eq:thm:mainc} holds, provided 
Hypothesis ($k,\Theta_{\vbi,\vbl},\{0\}$) does whenever 
$k \ge r$ and $i_1 + \cdots + i_r = k$.

It remains only to prove our claim \eqref{eq:inc-exc}.
Let $\vba \in \NN^r$ with $a_1 < \cdots < a_r$ be given.
First of all note that, for any $\vbi \in \NN^r$, 
\begin{equation}
 \label{Aeq:inc-excb1}
 N_{\vbi,\vba}(x)
  =
   \sum_{n \le x} \ind{\SS}(n)
    \sums[0 < h_1 < \cdots < h_{i_1 + \cdots + i_r}]
         [h_{i_1 + \cdots + i_j} = a_j, \, j = 1,\ldots r]
          \ind{\SS}(n + h_1)\cdots \ind{\SS}(n + h_{i_1 + \cdots + i_r}).
\end{equation}
For $n \in \ZZ$, let $M_j(n,a_j)$ be the number of elements of 
$\SS$ in-between $n + a_{j - 1}$ and $n + a_j$, i.e.\ 
\[
 M_j(n)
  =
   M_j(n,a_j)
    \defeq 
      \card \SS \cap (n + a_{j - 1},n + a_j),    
\] 
for $j = 1,\ldots,r$, where $a_0 \defeq 0$.
From \eqref{Aeq:inc-excb1} it is not difficult to see that, for 
any $\vbi \in \NN^r$, 
\begin{equation}
 \label{Aeq:inc-excb2}
 N_{\vbi,\vba}(x)
  =
   \sum_{n \le x}
    \bigg\{
     \ind{\SS}(n)
      \ind{\SS}(n + a_1)\cdots \ind{\SS}(n + a_r)
       \binom{M_1(n)}{i_1 - 1}
        \cdots 
         \binom{M_r(n)}{i_r - 1} 
    \bigg\},
\end{equation}
where, as usual, $\binom{0}{j} = 1$ for $j \ge 0$ and 
$\binom{m}{j} = 0$ for $j > m$.
Note that, given $n,n + a_1,\ldots,n + a_r \in \SS$, these are 
{\em consecutive} elements of $\SS$, i.e.\ 
$n + a_j = \sts_{t + j}$, $j = 0,\ldots,r$, for some $t$, if and 
only if $M_1(n) = \cdots = M_r(n) = 0$.

Next, for any integer $k \ge r$ and any nonnegative integers 
$M_1,\ldots,M_r$, we have 
\begin{equation}
 \label{Aeq:inc-excb3}
 \sum_{i_1 + \cdots + i_r = k}
  \binom{M_1}{i_1 - 1}
   \cdots 
    \binom{M_r}{i_r - 1}
     =
      \binom{M_1 + \cdots + M_r}{k - r}.
\end{equation}
To see this, note that each summand on the left-hand side is the 
number of ways of choosing $k - r$ objects from a set $X$ of size 
$M_1 + \cdots + M_r$, in such a way that $i_j - 1$ objects are 
chosen from a subset $X_j \subseteq X$ of size $M_j$, and where 
$X = X_1 \cup \cdots \cup X_r$ is a partition.
Summing over all partitions of $k$ into $r$ positive integers, we 
end up with the total number of ways to choose $k - r$ objects 
from $X$, viz.\ the right-hand side.
Also, 
\begin{align}
 \label{Aeq:inc-excb4}
 \sum_{k - r \ge 0} 
  (-1)^{k - r}
   \binom{M_1 + \cdots + M_r}{k - r}
    =
     \begin{cases}
      1 & \text{$M_1 + \cdots + M_r = 0$,} \\
      0 & \text{otherwise,}
     \end{cases}
\end{align}
the left-hand side being the binomial expansion of 
$(1 - 1)^{M_1 + \cdots + M_r}$ in the second case.
Furthermore, for any nonnegative integers $M_1,\ldots,M_r$ we have  
\begin{align}
 \begin{split}
  \label{Aeq:inc-excb5}
 & 
 \sum_{k - r = 0}^{2\ell + 1}
  (-1)^{k - r}
   \binom{M_1 + \cdots + M_r}{k - r}
 \\
 & \hspace{30pt}
    \le 
     \sum_{k - r \ge 0} 
      (-1)^{k - r}
       \binom{M_1 + \cdots + M_r}{k - r}
        \le 
         \sum_{k - r = 0}^{2\ell}
          (-1)^{k - r}
           \binom{M_1 + \cdots + M_r}{k - r},
  \end{split}
\end{align}
as can be verified by using the recurrence relation
$
 \binom{m}{i} = \binom{m - 1}{i} + \binom{m - 1}{i - 1}
$.

Combining \eqref{Aeq:inc-excb2}, \eqref{Aeq:inc-excb3} and 
\eqref{Aeq:inc-excb4} we find, after changing order of summation, 
that 
\begin{equation}
 \label{Aeq:inc-excb6}
  \sums[n \le x][\text{consecutive}] 
   \ind{\SS}(n)\ind{\SS}(n + a_1)\cdots \ind{\SS}(n + a_r)
    =
  \sum_{k - r \ge 0} (-1)^{k - r}
   \sum_{i_1 + \cdots + i_r = k} 
    N_{\vbi,\vba}(x),
\end{equation}
where in the summand on the left-hand side, ``consecutive'' 
indicates summation restricted to those $n$ for which 
$n + a_j = \sts_{t + j}$, $j = 0,\ldots,r$, for some $t$.
Combining \eqref{Aeq:inc-excb2}, \eqref{Aeq:inc-excb3} and 
\eqref{Aeq:inc-excb5} we similarly find that 
\begin{align}
 \begin{split}
  \label{Aeq:inc-excb7}
 & 
  \sum_{k - r = 0}^{2\ell + 1} (-1)^{k - r}
   \sum_{i_1 + \cdots + i_r = k} 
    N_{\vbi,\vba}(x)
 \\
 & \hspace{15pt} 
  \le 
   \sums[n \le x][\text{consecutive}] 
    \ind{\SS}(n)\ind{\SS}(n + a_1)\cdots \ind{\SS}(n + a_r)
  \le 
   \sum_{k - r = 0}^{2\ell} (-1)^{k - r}
    \sum_{i_1 + \cdots + i_r = k} 
     N_{\vbi,\vba}(x).
 \end{split}
\end{align}
These are the claimed inequalities in \eqref{eq:inc-exc}, since 
\[
 \sums[n \le x][\text{consecutive}] 
  \ind{\SS}(n)\ind{\SS}(n + a_1)\cdots \ind{\SS}(n + a_r)
   =
    \sums[\sts_t \le x][\sts_{t + j} - \sts_t = a_j][j = 1,\ldots,r] 1.
\]

\end{proof}
\end{nixnix}
%
%****************************************************************%
%************************** END DETAIL **************************%
%****************************************************************%
%

\begin{nix}
\begin{proof}[Proof of Proposition \ref{prop:S2S3S1}] 
(a)
Trivially, $0 \equiv \sots \bmod 2^{\nu}$ for all $\nu \ge 1$.
Let $\beta \ge 0$ and $m \equiv 1 \bmod 4$.
We have $2^{\beta} = \sots$, and we claim that 
$m \equiv \sots \bmod 2^{\nu}$ for all $\nu \ge 1$.
By Brahmagupta's identity, it follows that 
$2^{\beta}m \equiv \sots \bmod 2^{\nu}$ for all $\nu \ge 1$.
We trivially have $m \equiv \sots \bmod 2^{\nu}$ for $\nu = 1,2$.
For $\nu \ge 2$, if $m \equiv a^2 + b^2 \bmod 2^{\nu}$ then 
$a + b \equiv 1 \bmod 2$ and either 
$m \equiv a^2 + b^2 \bmod 2^{\nu + 1}$ or 
\[
 m \equiv a^2 + b^2 + 2^{\nu}
    \equiv (a + 2^{\nu - 1})^2 + (b + 2^{\nu - 1})^2
     \bmod 2^{\nu + 1}.
\]

Next, suppose $n \ne 0$.
Then $n = 2^{\beta}m$ for some $\beta \ge 0$ and 
$m \equiv \pm 1 \bmod 4$.
We claim that if $2^{\beta}m \equiv \sots \bmod 2^{\beta + 2}$ 
then $m \equiv 1 \bmod 4$.
This holds trivially for $\beta = 0$.
For $\beta \ge 0$, if 
$2^{\beta + 1}m \equiv a^2 + b^2 \bmod 2^{\beta + 3}$ then 
$a \equiv b \bmod 2$ and, letting $c = (a + b)/2$ and 
$d = (a - b)/2$, we see that  
$2^{\beta}m \equiv c^2 + d^2 \bmod 2^{\beta + 2}$.

(b)
Trivially, $0 \equiv \sots \bmod p^{\nu}$ for all $\nu \ge 1$.
Let $\beta \ge 0$ and $m \not\equiv 0 \bmod p$.
We have $p^{2\beta} = \sots$, and we claim that 
$m \equiv \sots \bmod p^{\nu}$ for all $\nu \ge 1$.
In view of Brahmagupta's identity, it follows that 
$p^{2\beta}m \equiv \sots \bmod p^{\nu}$ for all $\nu \ge 1$.
For $\nu = 1$ note that since, by Euclid's lemma, the sets 
\[
 \{m - a^2 \bmod p : a = 0,\ldots,p - 1\}
  \quad 
   \text{and} 
    \quad 
     \{b^2 \bmod p : b = 0,\ldots,p - 1\}
\]
both contain $(p + 1)/2$ congruence classes, their intersection 
must be nonempty.
Hence $m \equiv a^2 + b^2 \bmod p$ for some $a$ and $b$.
For $\nu \ge 1$, if $m \equiv a^2 + b^2 \bmod p^{\nu}$ then 
$m \equiv a^2 + b^2 + p^{\nu}r \bmod p^{\nu + 1}$ for some integer 
$r$ and, without loss of generality, $a \not\equiv 0 \bmod p$ 
(because  $m \not\equiv 0 \bmod p$).
In that case we have $2aa' \equiv 1 \bmod p^{\nu}$ for some 
integer $a'$, and so
$
  m \equiv a^2 + b^2 + p^{\nu}r
     \equiv (a + p^{\nu})^2a'r + b^2 
      \bmod p^{\nu + 1}
$.

Next, suppose $n \ne 0$.
Then $n = p^{\alpha}m$ for some $\alpha \ge 0$ and 
$m \not\equiv 0 \bmod p$.
We claim that if $p^{\alpha}m \equiv \sots \bmod p^{\alpha + 1}$ 
then $\alpha$ is even.
Suppose for a contradiction that 
$p^{\alpha}m \equiv a^2 + b^2 \bmod p^{\alpha + 1}$ but 
$\alpha$ is odd.
Then $a^2 \equiv -b^2 \bmod p$ and so, since $(p - 1)/2$ is even 
(as $p \equiv 3 \bmod 4$), $a^{p - 1} \equiv -b^{p - 1} \bmod p$.
In view of Fermat's little theorem we must have 
$a \equiv b \equiv 0 \bmod p$.
Letting $c = a/p$ and $d = b/p$, we see that 
$p^{\alpha}m \equiv p^2(c^2 + d^2) \bmod p^{\alpha + 1}$.
This gives a contradiction for $\alpha = 1$, and for 
$\alpha \ge 3$ implies that 
$p^{\alpha - 2}m \equiv c^2 + d^2 \bmod p^{\alpha - 1}$.

(c) 
We have $p^{\beta} = \sots$ for all $\beta \ge 0$ by Fermat's 
theorem on sums of two squares.
If $m \not\equiv 0 \bmod p$ then $m \equiv \sots \bmod p$ by 
the argument in the first paragraph of (b).
\end{proof}
\end{nix}

\begin{nix}
\begin{proposition}
 \label{prop:tripadmiss}
Let $h_1,h_2,h_3 \in \ZZ$.
The set $\bh = \{h_1,h_2,h_3\}$ is $\SS$-admissible.
\end{proposition}

\begin{proof}
If $p > 3$ then $\{-h_1,-h_2,-h_3\}$ is not a complete set of 
residues modulo $p$, and by Proposition \ref{prop:S2S3S1}, we have 
$n + \bh \subseteq S_p$ whenever $n + h_i \not\equiv 0 \bmod p$ 
for $i = 1,2,3$.
For $p = 3$, note that the least residues $b$ modulo $3^3$ for 
which $3 \emid b$ are $3,6,12,15,21$ and $24$, and consider the 
set 
\[
 \{-h_i,3 - h_i,6 - h_i,12 - h_i,15 - h_i,21 - h_i,24 - h_i : i = 1,2,3\}.
\]
This consists of representatives of at most 21 distinct 
congruence classes modulo $3^3$.
Thus, there exist integers $n$ such that, for $i = 1,2,3$, we have 
\[
 \text{$n + h_i \not\equiv 0,3,6,12,15,21$ or $24 \bmod 3^3$},
\]
meaning that $\nu_3(n + h_i) = 0$ or $\nu_3(n + h_i) = 2$.
By Proposition \ref{prop:S2S3S1}, for such $n$ we have 
$n + \bh \subseteq S_p$. 

Finally, consider $p = 2$.
By Proposition \ref{prop:S2S3S1}, $n \in S_2$ if and only if 
$n = 0$ or $n = 2^{\beta}m$ with $\beta \ge 0$ and 
$m \equiv 1 \bmod 4$.
Equivalently, $n \in S_2$ if and only if either $n = 0$ or there 
is some $\alpha \ge 0$ such that 
$n \equiv 2^{\beta}m \bmod 2^{\alpha + 2}$, where  
$0 \le \beta \le \alpha$ and $m \equiv 1 \bmod 4$.
Note that for $n \ne 0$, we have $n \in S_2$ if and only if 
$-n \not\in S_2$.

If there is some $h_i \in \bh$ such that 
$-h_i + \bh \subseteq S_2$, then there is nothing more to prove, 
so assume this is not the case.
This means that for each $h_i \in \bh$ there is some $h_j \in \bh$ 
such that $h_j - h_i = 2^{\beta_{ji}}m_{ji}$ with 
$\beta_{ji} \ge 0$ and $m_{ji} \equiv 3 \bmod 4$.
Without loss of generality, suppose $h_3 - h_1 \not\in S_2$.
Then $h_1 - h_3 \in S_2$, so it must be that 
$h_2 - h_3 \not\in S_2$.
Then $h_3 - h_2 \in S_2$, so it must be that 
$h_1 - h_2 \not\in S_2$.
Then $h_2 - h_1 \in S_2$.
We have: 
$h_2 - h_1 = 2^{\beta_{21}}m_{21}$,
$h_3 - h_2 = 2^{\beta_{32}}m_{32}$, and 
\[
 h_3 - h_1 = 2^{\beta_{32}}m_{32} + 2^{\beta_{21}}m_{21} 
           = 2^{\beta_{31}}m_{31},
\]
where $m_{21} \equiv m_{32} \equiv 1 \bmod 4$ and 
$m_{31} \equiv 3 \bmod 4$.
Suppose further, without loss of generality, that 
$\beta_{21} \le \beta_{32}$.

Consider the case $\beta_{21} = \beta_{32} = \beta$ (say).
We have   
$
  2^{\beta}(m_{32} + m_{21}) = 2^{\beta_{31}}m_{31} 
$,
and as $m_{32} + m_{21} \equiv 2 \bmod 4$, we must have 
$\beta_{31} = \beta + 1$. 
Thus, for any $n$ satisfying 
$n + h_1 \equiv 2^{\beta + 2} \bmod 2^{\beta + 4}$, we have  
\[
 n + h_2 \equiv 2^{\beta + 2} + 2^{\beta}m_{21}
          \equiv 2^{\beta}(4 + m_{21})
           \bmod 2^{\beta + 4},
\]
and 
\[
 n + h_3 \equiv 2^{\beta + 2} + 2^{\beta + 1}m_{31}
          \equiv 2^{\beta + 1}(2 + m_{31})
           \bmod 2^{\beta + 4}.
\]
It follows that $n + \bh \subseteq S_2$.

Now consider the case $\beta_{21} < \beta_{32}$.
We have 
$
 2^{\beta_{21}}(2^{\beta_{32} - \beta_{21}}m_{32} + m_{21}) 
  = 2^{\beta_{31}}m_{31} 
$.
Since 
$
 m_{31} \equiv 2^{\beta_{32} - \beta_{21}}m_{32} + m_{21} 
         \equiv 1 \bmod 2
$, 
we must have $\beta_{21} = \beta_{31} = \beta$ (say).
Then 
$
 2^{\beta_{32} - \beta}m_{32} 
  \equiv m_{31} - m_{21} 
   \equiv 3 - 1 
    \equiv 2 \bmod 4
$, 
so in fact $\beta_{32} = \beta + 1$.
In summary:  
$
 h_2 - h_1 = 2^{\beta}m_{21}
$, 
$
 h_3 - h_2 = 2^{\beta + 1}m_{32}
$,
and 
$
 h_3 - h_1 = 2^{\beta}(2m_{32} + m_{21})
$, 
where $m_{21} \equiv m_{32} \equiv 1 \bmod 4$.
If $\beta \ge 2$ take $n$ satisfying 
$n + h_1 \equiv 1 \bmod 2^{\beta + 3}$;
if $\beta = 1$ take $n$ satisfying 
$n + h_1 \equiv 2m_{21} \bmod 32$; 
if $\beta = 0$ take $n$ satisfying 
$n + h_1 \equiv m_{21} \bmod 16$.
In each case we have $n + \bh \subseteq S_2$.
\end{proof}
\end{nix}

\begin{nix}
\begin{proposition}
 \label{prop:S2hest}
Let $\bh = \{h_1,\ldots,h_k\}$ be a set of $k \ge 1$ integers, and 
assume that $0 \le h_1 < \cdots < h_k$.
Let $\alpha = 2 + \max_{i \ne j} \nu_2(h_i - h_j)$.
For $x \ge 1$ we have  
\begin{equation}
 \label{eq:S2hest}
  \frac{1}{x}
   \sums[n \le x][\forall i, n + h_i \in S_2] 1
  =
    \delta_{\bh}(2)
     \bigg\{
      1 + O\bigg(\frac{2^{\alpha}(1 + \card \bh_2\log(x + h_k))}{x}\bigg)
     \bigg\}.
\end{equation}
\end{proposition}

\begin{proof} 
Let $x \ge 1$.
We partition the sum in \eqref{eq:S2hest} as follows:
\begin{equation}
 \label{eq:S2hestpf1}
 \sums[n \le x][\forall i, n + h_i \in S_2] 1
  =
  \sums[n \le x][\forall i, n + h_i \in S_2][\forall i, \nu_2(n + h_i) < \alpha] 1
  \hspace{15pt} + 
    \sums[n \le x][\forall i, n + h_i \in S_2][\exists j, \nu_2(n + h_j) = \alpha] 1
    \hspace{15pt} +
      \sums[n \le x][\forall i, n + h_i \in S_2][\exists j, \nu_2(n + h_j) > \alpha] 1.
\end{equation}
(For a given $n$, if $\nu_2(n + h_i) = \beta$ and 
$\nu_2(n + h_j) > \beta$, then $\nu_2(h_i - h_j) = \beta$, 
implying 
$\beta \le \max_{i \ne j} \nu_2(h_i - h_j) \le \alpha - 2$, so 
there is certainly no overlap between the last two sums.)
By Proposition \ref{prop:S2S3S1}, $n + h_i \in S_2$ 
(for $n + h_i \ne 0$) if and only if there exists $\beta \ge 0$ 
and $m \equiv 1 \bmod 4$ such that $n + h_i = 2^{\beta}m$.
Using this, we verify that on the right-hand side of 
\eqref{eq:S2hestpf1}, the condition on $n$ in the first sum holds 
if and only if $n \equiv a \bmod 2^{\alpha + 1}$ for some $a$ in 
the (possibly empty) set 
\[
 \T_{\bh}(2^{\alpha + 1})
  =
   \{0 \le a < 2^{\alpha + 1} : \forall i, a + h_i \in S_2
      \,\, \hbox{and} \, \, \nu_2(a + h_i) < \alpha\}.
\]
Now,
\[
 \sum_{a \in \T_{\bh}(2^{\alpha + 1})}
  \sums[n \le x][n \equiv a \bmod 2^{\alpha + 1}] 1
   =
    x\frac{\card \T_{\bh}(2^{\alpha + 1})}{2^{\alpha + 1}} 
    + 
     O(\card \T_{\bh}(2^{\alpha + 1})),
\]
but the left-hand side here is also equal to 
\[
  \sum_{n \le x} 1
   -
    \sums[0 \le b < 2^{\alpha + 1}][b \not\in \T_{\bh}(2^{\alpha + 1})] 
     \sums[n \le x][n \equiv b \bmod 2^{\alpha + 1}] 1
   =
    x\frac{\card \T_{\bh}(2^{\alpha + 1})}{2^{\alpha + 1}} 
    + 
     O(2^{\alpha + 1} - \card \T_{\bh}(2^{\alpha + 1})).
\]
Hence 
\begin{equation}
 \label{eq:S2hestpf2}
 \sums[n \le x][\forall i, n + h_i \in S_2][\forall i, \nu_2(n + h_i) < \alpha] 1
  =
   x 
    \frac{\card \T_{\bh}(2^{\alpha + 1})}{2^{\alpha + 1}}
    + 
     O(\min\{\card \T_{\bh}(2^{\alpha + 1}),2^{\alpha + 1} - \T_{\bh}(2^{\alpha + 1})\}).
\end{equation}

In the second sum on the right-hand side of \eqref{eq:S2hestpf1}, 
we claim that the condition on $n$ holds if and only if the set 
$
 \bh_2 = \{h_j \in \bh : \forall i, h_j - h_i \in S_2\} 
$
is nonempty and 
$n + h_j \equiv 2^{\alpha} \bmod 2^{\alpha + 2}$ for the 
(necessarily unique) $h_j$ in $\bh_2$.
Thus,  
\begin{equation}
 \label{eq:S2hestpf3}
 \sums[n \le x][\forall i, n + h_i \in S_2][\exists j, \nu_2(n + h_j) = \alpha] 1
  =
   \card\bh_2
    \bigg(\frac{x}{2^{\alpha + 2}} + O(1)\bigg).
\end{equation}
To verify the claim, suppose that for some $j$ we have 
$n + h_j = 2^{\alpha}(1 + 2q)$.
By Proposition \ref{prop:S2S3S1}, this is in $S_2$ if and only if 
$2 \mid q$.
Now let $i \ne j$, so that there exists 
$\beta_{ij} \le \max_{i \ne j} \nu_2(h_i - h_j) \le \alpha - 2$ 
and $m_{ij} \equiv \pm 1 \bmod 4$ such that 
$h_i - h_j = 2^{\beta_i}m_i$.
Thus, 
$
 n + h_i 
  = 2^{\beta_{ij}}
   (m_{ij} + 2^{\alpha - \beta_{ij}}(1 + 2q)) 
$ 
is in $S_2$ if and only if $m_{ij} \equiv 1 \bmod 4$, 
equivalently, $h_i - h_j \in S_2$.
By definition of $\bh_2$, this holds for each $i \ne j$ if and 
only if $h_j \in \bh_2$.

By the same argument, the condition on $n$ in the third sum on the 
right-hand side of \eqref{eq:S2hestpf1} holds if and only if 
$n + h_j = 2^{\beta}m$ for some $\beta > \alpha$, 
$m \equiv 1 \bmod 4$, and $h_j$ in the (possibly empty) set 
$\bh_2$.
Thus,  
\[
 \sums[n \le x][\forall i, n + h_i \in S_2][\exists j, \nu_2(n + h_j) > \alpha] 1 
  =
   \sum_{h_j \in \bh_2}
    \sum_{\beta > \alpha}
     \sums[n \le x][n + h_j = 2^{\beta}m][m \equiv 1 \bmod 4] 1.
\] 
Assume $x + h_k \ge 2^{\alpha + 1}$, and let $\gamma$ be the 
integer such that $2^{\gamma} \le x + h_k < 2^{\gamma + 1}$.
For any $h_j \in \bh$ we have 
\[
 \sum_{\beta > \alpha}
  \sums[n \le x][n + h_j = 2^{\beta}m][m \equiv 1 \bmod 4] 1
   =
    \sum_{\beta = \alpha + 1}^{\gamma}
     \sums[h_j < 2^{\beta}m \le x + h_j][m \equiv 1 \bmod 4] 1
      =
      \frac{x}{4}
       \sum_{\beta = \alpha + 1}^{\gamma}
        \frac{1}{2^{\beta}} 
         + 
          O(\gamma).
\]
Since
\[
 \sum_{\beta = \alpha + 1}^{\gamma}
  \frac{1}{2^{\beta}}
   =
    \frac{1}{2^{\alpha}}
     -
      \frac{1}{2^{\gamma}}
      =
       \frac{1}{2^{\alpha}} + O\bigg(\frac{1}{x}\bigg), 
\]
and since $\gamma \le \log(x + h_k)/\log 2$, we see that 
\[
 \sum_{\beta > \alpha}
  \sums[n \le x][n + h_j = 2^{\beta}m][m \equiv 1 \bmod 4] 1 
   =
    \frac{x}{2^{\alpha + 2}} + O(\log(x + h_k)).
\]
This also holds trivially if $x + h_k < 2^{\alpha + 1}$, so in 
any case we have 
\begin{equation}
 \label{eq:S2hestpf4}
 \sums[n \le x][\forall i, n + h_i \in S_2][\exists j, \nu_2(n + h_j) > \alpha] 1 
  =
   \card\bh_2\frac{x}{2^{\alpha + 2}} + O(\card\bh_2\log x).
\end{equation}

Combining \eqref{eq:S2hestpf1} with \eqref{eq:S2hestpf2}, 
\eqref{eq:S2hestpf3} and \eqref{eq:S2hestpf4} gives 
\[
  \sums[n \le x][\forall i, n + h_i \in S_2] 1
  =
    \frac{x}{2^{\alpha + 1}}
     \big(\card \T_{\bh}(2^{\alpha + 1}) + \card \bh_2\big)
     + 
      O\big(
        \card \T_{\bh}(2^{\alpha + 1}) 
       + \card\bh_2 \log(x + h_k) 
       \big).
\]
(This estimate also holds with 
$2^{\alpha + 1} - \card \T_{\bh}(2^{\alpha + 1})$ in place of 
$\card \T_{\bh}(2^{\alpha + 1})$ in the $O$-term.)
By Proposition \ref{prop:S2h} (b), 
$
  \delta_{\bh}(2)
   =
    (\card \T_{\bh}(2^{\alpha + 1}) + \card \bh_2)/2^{\alpha + 1} 
$, 
giving the main term in \eqref{eq:S2hest}.
For the $O$-term, note that 
\[
 \card \T_{\bh}(2^{\alpha + 1}) + \card\bh_2 \log(x + h_k)
  \le 
   2^{\alpha + 1}\delta_{\bh}(2)\big(1 + \card \bh_2\log(x + h_k)\big). 
\]
\end{proof}

\begin{proposition}
 \label{prop:Sp3hest}
Let $\bh = \{h_1,\ldots,h_k\}$ be a set of $k \ge 1$ integers, and 
assume that $0 \le h_1 < \cdots < h_k$.
Let $p \equiv 3 \bmod 4$, and let 
$\alpha = 1 + \max_{i \ne j} \nu_p(h_i - h_j)$.
For $x \ge 1$ we have  
\begin{equation}
 \label{eq:Sp3hest}
  \frac{1}{x}
   \sums[n \le x][\forall i, n + h_i \in S_p] 1
  =
   \delta_{\bh}(p)
    \bigg\{ 
     1 + O\bigg(\frac{p^{\alpha}}{x} + \frac{p^{\alpha + \alpha \bmod 2} \log(x + h_k)}{x\log p}\bigg)
    \bigg\}.
\end{equation}
In the $O$-term, $\log(x + h_k)$ may be replaced by $0$ if 
$\card \bh_p = \emptyset$.
Also, if $\alpha = 1$ then, in the $O$-term, $p^{\alpha}$ and 
$p^{\alpha + \alpha \bmod 2}$ may be replaced by $k$ and $k^2$ 
respectively.
\end{proposition}

\begin{proof} 
By Proposition \ref{prop:S2S3S1}, $n + h_i \in S_p$ (for 
$n + h_i \ne 0$) if and only if $2 \mid \nu_p(n + h_i)$.
Let $x \ge 1$.
We partition the sum in \eqref{eq:Sp3hest} as follows:
\begin{equation}
 \label{eq:Sp3hestpf1}
 \sums[n \le x][\forall i, n + h_i \in S_p] 1
 =
  \sums[n \le x][\forall i, 2 \mid \nu_p(n + h_i)][\forall i, \nu_p(n + h_i) < \alpha] 1
   \hspace{15pt}
  + 
   \sums[n \le x][\forall i, 2 \mid \nu_p(n + h_i)][\exists j, \nu_p(n + h_j) \ge \alpha] 1.
\end{equation}
%
%****************************************************************%
%************************* START DETAIL *************************%
%****************************************************************%
%
\begin{nixnix} 
On the right-hand side of \eqref{eq:Sp3hestpf1}, the condition on 
$n$ in the first sum holds if and only if 
$n \equiv a \bmod p^{\alpha}$ for some $a$ in the (possibly 
empty) set 
\[
 \V_{\bh}(p^{\alpha})
  =
   \{0 \le a < p^{\alpha} : \forall i, 2 \mid \nu_p(a + h_i) \,\, \hbox{and} \,\, \nu_p(a + h_i) < \alpha\}.
\]
Now,
\[
 \sum_{a \in \V_{\bh}(p^{\alpha})}
  \sums[n \le x][n \equiv a \bmod p^{\alpha}] 1
   =
    x\frac{\card \V_{\bh}(p^{\alpha})}{p^{\alpha}} 
    + 
     O(\card \V_{\bh}(p^{\alpha})),
\]
but the left-hand side here is also equal to 
\[
  \sum_{n \le x} 1
   -
    \sums[0 \le b < p^{\alpha}][b \not\in \V_{\bh}(p^{\alpha})] 
     \sums[n \le x][n \equiv b \bmod p^{\alpha}] 1
   =
    x\frac{\card \V_{\bh}(p^{\alpha})}{p^{\alpha}} 
    + 
     O(p^{\alpha} - \card \V_{\bh}(p^{\alpha})).
\]
Hence 
\begin{equation}
 \label{eq:Sp3hestpf2}
 \sums[n \le x][\forall i, 2 \mid \nu_p(n + h_i)][\forall i, \nu_p(n + h_i) < \alpha] 1
  =
   x\frac{\card \V_{\bh}(p^{\alpha})}{p^{\alpha}} 
    + 
     O(\min\{\card \V_{\bh}(p^{\alpha}),p^{\alpha} - \card \V_{\bh}(p^{\alpha})\}).
\end{equation}

Consider the second sum on the right-hand side of 
\eqref{eq:Sp3hestpf1}. 
For a given $n$, there is at most one $h_j$ in $\bh$ such that 
$\nu_p(n + h_j) \ge \alpha$. 
Therefore, the condition on $n$ in the second sum holds if and 
only if $2 \mid \nu_p(n + h_j)$ and $\nu_p(n + h_j) \ge \alpha$ 
for some $h_j$ in the (possibly empty) set 
$
 \bh_p 
  =
   \{h_j \in \bh : \forall i \ne j, 2 \mid \nu_p(h_i - h_j)\} 
$.
Thus,  
\[
 \sums[n \le x][\forall i, 2 \mid \nu_p(n + h_i)][\exists j, \nu_p(n + h_j) \ge \alpha] 1 
  =
   \sum_{h_j \in \bh_p}
    \sum_{\beta \ge \frac{\alpha}{2}}
     \sums[n \le x][p^{2\beta} \emid n + h_j] 1.
\]

Assume $x + h_k \ge p^{\alpha}$, and let $\gamma$ be the integer 
such that $p^{\gamma} \le x + h_k < p^{\gamma + 1}$.
For any $h_j \in \bh$ we have 
\begin{align*}
  \sum_{\beta \ge \frac{\alpha}{2}}
   \sums[n \le x][p^{2\beta} \emid n + h_j] 1
 & =
     \sum_{\frac{\alpha}{2} \le \beta \le \frac{\gamma}{2}} 
      \Big\{ 
            \sums[n \le x][p^{2\beta} \mid n + h_j] 1
           - 
             \sums[n \le x][p^{2\beta + 1} \nmid n + h_j] 1
      \Big\}
 \\
 & = 
   x\bigg(1 - \frac{1}{p}\bigg)
      \sum_{\frac{\alpha}{2} \le \beta \le \frac{\gamma}{2}} 
       \frac{1}{p^{2\beta}}
       + 
        O(\gamma).
\end{align*}
Since  
\begin{align*}
 \sum_{\frac{\alpha}{2} \le \beta \le \frac{\gamma}{2}} 
  \frac{1}{p^{2\beta}}
   & 
   =
    \bigg(1 - \frac{1}{p^2}\bigg)^{-1}
     \bigg(\frac{1}{p^{\alpha + (\alpha \bmod 2)}} - \frac{p^{\gamma \bmod 2}}{p^{\gamma + 2}}\bigg)
  \\
   & 
    =
       \bigg(1 - \frac{1}{p}\bigg)^{-1}
        \bigg(1 + \frac{1}{p}\bigg)^{-1}
         \frac{1}{p^{\alpha + (\alpha \bmod 2)}}
       + O\bigg(\frac{1}{x}\bigg),
\end{align*}
and since $\gamma \le \log(x + h_k)/\log p$, we see that 
\[
 \sum_{\beta \ge \frac{\alpha}{2}}
  \sums[n \le x][p^{2\beta} \emid n + h_j] 1
   =
    x\bigg(1 + \frac{1}{p}\bigg)^{-1}
      \frac{1}{p^{\alpha + (\alpha \bmod 2)}}
       +
       O\bigg(\frac{\log(x + h_k)}{\log p}\bigg).  
\]
This also holds trivially if $x + h_k < p^{\alpha}$, so in any 
case we have 
\begin{equation}
 \label{eq:Sp3hestpf3}
  \sums[n \le x][\forall i, 2 \mid \nu_p(n + h_i)][\exists j, \nu_p(n + h_j) \ge \alpha] 1
   =
    \card\bh_p \,
     x\bigg(1 + \frac{1}{p}\bigg)^{-1}
       \frac{1}{p^{\alpha + (\alpha \bmod 2)}}
        +
         O\bigg(\card \bh_p \frac{\log(x + h_k)}{\log p}\bigg).
\end{equation}

Combining \eqref{eq:Sp3hestpf1} with \eqref{eq:Sp3hestpf2} and 
\eqref{eq:Sp3hestpf3} gives
\begin{align*}
   \sums[n \le x][\forall i, n + h_i \in S_p] \hspace{-5pt} 1
  = 
  \textstyle 
   \frac{x}{p^{\alpha}}
    \Big(\card \V_{\bh}(p^{\alpha}) + \card \bh_p\big(1 + \frac{1}{p}\big)^{-1}\frac{1}{p^{\alpha \bmod 2}}\Big)
     + 
      O\Big(
        \card \V_{\bh}(p^{\alpha})+ \card\bh_p \frac{\log(x + h_k)}{\log p} 
       \Big).
\end{align*}
\end{nixnix}
%
%****************************************************************%
%************************** END DETAIL **************************%
%****************************************************************%
%
We argue along the lines of the proof of 
Proposition \ref{prop:S2hest}, ending up with  
\begin{align*}
   \sums[n \le x][\forall i, n + h_i \in S_p] \hspace{-5pt} 1
  = 
  \textstyle 
   \frac{x}{p^{\alpha}}
    \Big(\card \V_{\bh}(p^{\alpha}) + \card \bh_p\big(1 + \frac{1}{p}\big)^{-1}\frac{1}{p^{\alpha \bmod 2}}\Big)
     + 
      O\Big(
        \card \V_{\bh}(p^{\alpha})+ \card\bh_p \frac{\log(x + h_k)}{\log p} 
       \Big),
\end{align*}
which also holds with $p^{\alpha} - \card \V_{\bh}(p^{\alpha})$ in 
place of $\card \V_{\bh}(p^{\alpha})$ in the $O$-term. 
Proposition \ref{prop:Sp3h} (b) then gives the main term in 
\eqref{eq:Sp3hest}.
For the $O$-term, note that 
$
 p^{\alpha}\delta_{\bh}(p)
  \ge 
   \card \V_{\bh}(p^{\alpha})  
$
and 
$
 p^{\alpha + \alpha \bmod 2}\delta_{\bh}(p) 
  \ge 
   \card \bh_p
$.
If $\bh_p = \emptyset$ then the term with $\log(x + h_k)$ may be 
omitted.

To deal with the special case where $\alpha = 1$, i.e.\ 
$p \nmid \det(\bh)$, we note that $p - \card \V_{\bh}(p) = 1$ and 
$\bh_p = \bh$.
Since we may take $p^{\alpha} - \card \V_{\bh}(p^{\alpha})$ in place 
of $\card \V_{\bh}(p^{\alpha})$ in the $O$-term above, 
\[
 \sums[n \le x][\forall i, n + h_i \in S_p] 1
  =
   x\delta_{\bh}(p)
     + 
      O\bigg(
        1 + k\frac{\log(x + h_k)}{\log p} 
       \bigg),
\]
where $\delta_{\bh}(p) = 1 - k/(p + 1) \ge 1/(k + 1)$ (see 
Proposition \ref{prop:Sp3h} (c)).
\end{proof}

\end{nix}

\begin{nix}
\begin{proposition}
 \label{Aprop:sssk=2}
Let $h$ be a nonzero integer.
We have 
\begin{equation}
 \label{Aeq:delthk=2}
  \delta_{\{0,h\}}(2) = \frac{1}{4}
   \quad 
    (\nu_2(h) = 0), 
     \quad 
      \delta_{\{0,h\}}(2)  
       = 
        \frac{2^{\nu_2(h) + 1} - 3}{2^{\nu_2(h) + 2}}
         \quad 
          (\nu_2(h) \ge 1), 
\end{equation}
and for $p \equiv 3 \bmod 4$ we have 
\begin{equation}
 \label{Aeq:delthp3k=2}
 \delta_{\{0,h\}}(p) 
  =
   \bigg(1 + \frac{1}{p}\bigg)^{-1}
    \bigg(1 - \frac{1}{p^{\nu_p(h) + 1}}\bigg).
\end{equation}
Consequently, if 
$h = 2^{\alpha}p_1^{\alpha_1}\cdots p_r^{\alpha_r}q$, 
where $\alpha \ge 0$, $p_i \equiv 3 \bmod 4$ and $\alpha_i \ge 1$ 
for $i = 1,\ldots,r$, and $q$ is composed only of primes congruent 
to $1$ modulo $4$, then  
\begin{equation}
 \label{Aeq:sssk=2}
 \mathfrak{S}_{\{0,h\}}
  =
   \frac{2\delta_{\{0,h\}}(2)}{C^2}
    \prod_{i = 1}^r
     \bigg(1 - \frac{1}{p_i}\bigg)^{-1}
      \bigg(1 - \frac{1}{p_i^{\alpha_i + 1}}\bigg),
\end{equation}
where $C$ is the Landau--Ramanujan constant 
\textup{(}see \eqref{eq:defLanRamconst}\textup{)}.
\end{proposition}

\begin{proof}
Let $\bh = \{0,h\}$.
Let $\nu_2(h) = \alpha$, so that $h = 2^{\alpha}h'$, say.
If $h' \equiv 1 \bmod 4$ then $\bh_2 = \{0\}$, and if 
$h' \equiv 3 \bmod 4$ then $\bh_2 = \{h\}$, so in any case 
$\card \bh_2 = 1$.
First consider the case where $\alpha = 0$, i.e.\ 
$\bh = \{0,h'\}$.
We have 
\[
 \T_{\bh}(8)
  =
   \{0 \le a < 8 : a,a + h' \equiv 1,2 \, \hbox{or} \, 5 \bmod 8\}.
\]
Since $h' \equiv 1,3,5$ or $7 \bmod 8$, we verify that  
$\T_{\bh}(8) = \{1\},\{2\},\{5\}$ or $\{2\}$ respectively.
Thus, $\card \T_{\bh}(8) = 1$ and 
\[
 \delta_{\bh}(2) 
  =
   \frac{\card \T_{\bh}(8) + \card \bh_2}{8}
    =
     \frac{1 + 1}{8}
      = 
       \frac{1}{4}.
\]
\begin{nixnix}
Consider the case where $\alpha = 1$, i.e.\ 
$\bh = \{0,2h'\}$.
We have 
\[
 \T_{\bh}(16) 
  =
   \{0 \le a < 16 : a, a + 2h' \equiv 1,2,4,5,9,10 \, \hbox{or} \, 13 \bmod 16\}.
\]
As $2h' \equiv 2,6,10$ or $14 \bmod 16$, we verify that 
$\T_{\bh}(16) = \{2\}$, $\T_{\bh}(16) = \{4\}$, 
$\T_{\bh}(16) = \{10\}$ or $\T_{\bh}(16) = \{4\}$ respectively.
Thus, $\card \T_{\bh}(16) = 1$ and 
\[
 \delta_{\bh}(2) 
  =
   \frac{\card \T_{\bh}(16) + \card \bh_2}{16}
    =
     \frac{1 + 1}{16}
      = 
       \frac{1}{8}.
\]
\end{nixnix}

Next, consider the case where $\alpha \ge 1$.
Recall from Proposition \ref{prop:S2h} that in general, 
\[
 \T_{\bh}(2^{\alpha + 3}) 
  =
   \{a,a + 2^{\alpha + 2} : a \in \T_{\bh}(2^{\alpha + 2})\}
    \cup 
     \U_{\bh}(2^{\alpha + 3}),
\]
where $\card \U_{\bh}(2^{\alpha + 3}) = \card \bh_2$.
By definition,  
\[
 \T_{\bh}(2^{\alpha + 2}) 
  =
   \{0 \le a < 2^{\alpha + 2} : a,a + h \in S_2, \nu_2(a),\nu_2(a + h) \le \alpha\},
\]
We claim that if $a \in S_2$ and $\nu_2(a) = \alpha$, then either 
$a + h \not\in S_2$ or $\nu_2(a + h) > \alpha$.
For if $a = 2^{\alpha}m$ with $m \equiv 1 \bmod 4$, then   
$a + h = 2^{\alpha}(m + h')$ and $m + h' \equiv 2$ or 
$0 \bmod 4$ ($h' \equiv 1$ or $3 \bmod 4$).
Likewise, if $a + h \in S_2$ and $\nu_2(a + h) \le \alpha$, then 
either $a \not\in S_2$ or $\nu_2(a) > \alpha$.
Similarly, if $a \in S_2$ and $\nu_2(a) = \alpha - 1$ then 
$a + h \not\in S_2$, and likewise with $a$ and $a + h$ 
interchanged.
If $\nu_2(a) \le \alpha - 2$ or $\nu_2(a + h) \le \alpha - 2$, 
then $a,a + h \in S_2$ if and only if $a \in S_2$.

In view of all of this, 
$
 \T_{\bh}(2^{\alpha + 2})
  =
   \{0 \le a < 2^{\alpha + 2} : a \in S_2 \,\, \hbox{and} \,\, \nu_2(a) \le \alpha - 2\} 
$, 
and so 
\[
 \card \T_{\bh}(2^{\alpha + 2})
  =
   \sum_{0 \le \beta \le \alpha - 2}
    \hspace{5pt}
     \sums[0 \le m < 2^{\alpha + 2 - \beta} ][m \equiv 1 \bmod 4] 1
      =
       \sum_{0 \le \beta \le \alpha - 2}
         2^{\alpha - \beta}
        =
         2^{\alpha + 1} - 4.
\]
Since $\card \U_{\bh}(2^{\alpha + 3}) = \card \bh_2 = 1$, it 
follows that 
\[
 \card \T_{\bh}(2^{\alpha + 3}) 
  = 2(2^{\alpha + 1} - 4) + 1 = 2^{\alpha + 2} - 7, 
\]
and in turn that 
\[
 \delta_{\bh}(2)
  = 
   \frac{\card \T_{\bh}(2^{\alpha + 3}) + \card \bh_2}{2^{\alpha + 3}}
    =
     \frac{2^{\alpha + 2} - 7 + 1}{2^{\alpha + 3}}
      =
       \frac{2^{\alpha + 1} - 3}{2^{\alpha + 2}}.
\]

Let $p \equiv 3 \bmod 4$ and let $\gamma = \nu_p(h)$.
If $\nu_p(a) \le \gamma - 1$ then $\nu_p(a + h) = \nu_p(a)$, so 
$a,a + h \in S_p$ and $\nu_p(a),\nu_p(a + h) \le \gamma - 1$ if 
and only if $2 \mid \nu_p(a)$ and $\nu_p(a) \le \gamma - 1$.
Therefore,
\begin{align*}
 & \#\{0 \le a < p^{\gamma + 1} : a,a + h \in S_p, \nu_p(a),\nu_p(a + h) \le \gamma - 1\}
  \\
 & \hspace{30pt}  = 
    \sum_{0 \le \beta \le \frac{\gamma - 1}{2}}
     \sums[0 \le m < p^{\gamma + 1}][m \not\equiv 0 \bmod p] 1
  \\
  & \hspace{30pt} = 
     p^{\gamma + 1}
      \bigg(1 - \frac{1}{p}\bigg)
       \sum_{0 \le \beta \le \frac{\gamma - 1}{2}}
        \frac{1}{p^{2\beta}}
  \\
  & \hspace{30pt} = 
   p^{\gamma + 1}
    \bigg(1 + \frac{1}{p}\bigg)^{-1}
     \bigg(1 - \frac{1}{p^{\gamma + (\gamma \bmod 2)}}\bigg).
\end{align*}
If $\gamma$ is odd, this accounts for all of 
$\V_{\bh}(p^{\gamma + 1})$.
Also, $\bh_p = \emptyset$.
In that case we have  
\[
 \delta_{\bh}(p)
  =
   \frac{1}{p^{\gamma + 1}}
    \bigg(
     \card \V_{\bh}(p^{\gamma + 1}) + \card \bh_p\bigg(1 + \frac{1}{p}\bigg)^{-1} 
    \bigg)
     =
      p^{\gamma + 1}
    \bigg(1 + \frac{1}{p}\bigg)^{-1}
     \bigg(1 - \frac{1}{p^{\gamma + 1}}\bigg).
\]
If $\gamma$ is even then there may exist 
$a \in \V_{\bh}(p^{\gamma + 1})$ such that  
$\max\{\nu_p(a),\nu_p(a + h)\} = \gamma$.
This holds if and only if 
$\nu_p(a) = \nu_p(a + h) = \gamma$.
Writing $h = p^{\gamma}h''$ (so that $h'' \not\equiv 0 \bmod p$), 
we have $\nu_p(a) = \nu_p(a + h) = \gamma$ if and only if 
$a = p^{\gamma}m$, where 
$m \not\equiv 0 \bmod p$ and $m \not\equiv - h'' \bmod p$.

Therefore, if $\gamma$ is even then  
\[
 \#\{0 \le a < p^{\gamma + 1} : a, a + h \in S_p, \max\{\nu_p(a),\nu_p(a + h)\} = \gamma\}
  =
   p - 2.
\]
Also, $\bh_p = \bh$.
In that case we have 
\begin{align*}
 \delta_{\bh}(p) 
  & =
   \frac{1}{p^{\gamma + 1}}
    \bigg(
     \card \V_{\bh}(p^{\gamma + 1}) + \card \bh_p\bigg(1 + \frac{1}{p}\bigg)^{-1}\frac{1}{p}
    \bigg)
  \\
  & =
      \frac{1}{p^{\gamma + 1}}
       \bigg(
        p^{\gamma + 1}
         \bigg(1 + \frac{1}{p}\bigg)^{-1}
          \bigg(1 - \frac{1}{p^{\gamma}}\bigg)
           +
            p - 2
             + 
              \bigg(1 + \frac{1}{p}\bigg)^{-1}
               \frac{2}{p}
       \bigg)
  \\
  & = 
   \bigg(1 + \frac{1}{p}\bigg)^{-1}
    \bigg(1 - \frac{1}{p^{\gamma + 1}}\bigg),
\end{align*}
as before.

Writing $h = 2^{\alpha}p_1^{\alpha_1}\cdots p_r^{\alpha_r}q$ as in 
the statement of the proposition, we see that 
\begin{align*}
 \prod_{p \equiv 3 \bmod 4}
  \bigg(1 + \frac{1}{p}\bigg)^2
   \delta_{\bh}(p)
 & = 
  \prod_{i = 1}^{r}
   \bigg(1 + \frac{1}{p_i}\bigg)
    \bigg(1 - \frac{1}{p_i^{\alpha_i + 1}}\bigg)
     \prods[p \equiv 3 \bmod 4][p \nmid h]
      \bigg(1 + \frac{1}{p}\bigg)
       \bigg(1 - \frac{1}{p}\bigg)
 \\
 & = 
  \prod_{i = 1}^{r}
   \bigg(1 - \frac{1}{p_i}\bigg)^{-1}
    \bigg(1 - \frac{1}{p_i^{\alpha_i + 1}}\bigg)
     \prod_{p \equiv 3 \bmod 4}
      \bigg(1 - \frac{1}{p^2}\bigg).
\end{align*}
This last product is equal to $1/(2C^2)$ (see 
\eqref{eq:defLanRamconst}).
The left-hand side is $\mathfrak{S}_{\bh}$ without the factor 
of $2^2\delta_{\bh}(2)$.
\end{proof}
\end{nix}

\begin{nixnix}
\begin{proposition}
 \label{Sprop:sssc}
Let $\bh$ be a set of $k \ge 1$ distinct integers.
For $z \ge \max\{2,k\}$ we have 
\begin{equation}
 \label{Seq:sssc1}
  {\textstyle 
  \exp 
   \Big(
    \frac{-c_1(k - 1)^2 z}{(z^2 - (k - 1)^2)\log z}
   \Big)
  }
    \le 
     \prods[p \not\equiv 1 \bmod 4][p \ge z, \, p \nmid \det(\bh)] 
      \delta_{\bz}(p)^{-k}\delta_{\bh}(p)
     \le 
      1,
\end{equation}
where $c_1$ is an absolute positive constant.
In fact, the upper bound in \eqref{Seq:sssc1} holds for 
$z \ge 2$.
For $z \ge \min\{3,k\}$ we have 
\begin{equation}
 \label{Seq:sssc2}
 {\textstyle 
  \exp 
   \Big(
    \frac{-c_2(k - 1)^2\log |\det(\bh)|}{(z^2 - (k - 1)^2)\log\log 3|\det(\bh)|}
   \Big)
 }
    \le 
     \prods[p \not\equiv 1 \bmod 4][p \ge z, \, p \mid \det(\bh)]
      \delta_{\bz}(p)^{-k}\delta_{\bh}(p)
     \le  
       {\textstyle 
         \exp 
          \Big(
           \frac{c_3k\log |\det(\bh)|}{z\log\log 3|\det(\bh)|}
          \Big) 
 },
\end{equation}
where $c_2$ and $c_3$ are absolute positive constants.
\end{proposition}

\begin{proof}
If $k = 1$ then $\mathfrak{S}_{\bh} = 1$, so the estimates 
\eqref{Seq:sssc1} and  \eqref{Seq:sssc2} are 
trivial in this case.
Let us assume for the rest of the proof that $k \ge 2$.

If $2 \nmid \det(\bh)$, then $k = 2$ and 
$\delta_{\bz}(2)^{-2}\delta_{\bh}(2) = (1/2)^{-2}(1/4) = 1$ (see 
Proposition \ref{prop:S2h} (c)), so only the primes 
$p \equiv 3 \bmod 4$ have any bearing on the product in 
\eqref{Seq:sssc1}.
Let $p \equiv 3 \bmod 4$.
If $p \nmid \det(\bh)$ then $k \ge p$ and, by 
Proposition \ref{prop:Sp3h} (c),  
\[
%  \textstyle 
  \delta_{\bz}(p)^{-k}\delta_{\bh}(p)
%   =
%    \bigg(
%     1 - \frac{1}{p + 1}
%    \bigg)^{-k}
%    \bigg(
%     1 - \frac{k}{p + 1}
%    \bigg)   
  =
   \bigg(
    1 + \frac{1}{p}
   \bigg)^{k - 1}
   \bigg(
    1 - \frac{k - 1}{p}
   \bigg).   
\]
Therefore, 
\[
%  \textstyle 
  \delta_{\bz}(p)^{-k}\delta_{\bh}(p)
   =
    1 
    - 
     \sum_{j = 2}^k 
      \bigg\{
       (k - 1)\binom{k - 1}{j - 1} - \binom{k - 1}{j}
      \bigg\}
      p^{-j}
       \le 
        1,
\]
so we see that the upper bound in \eqref{Seq:sssc1} 
holds for any $z$.
Whether or not $p$ divides $\det(\bh)$ we have, by 
\eqref{eq:delthpropsp3}, 
\begin{equation}
 \label{Seq:ssscpf1}
%  \textstyle 
 \delta_{\bz}(p)^{-k}\delta_{\bh}(p)
  \ge
   \bigg(1 + \frac{k - 1}{p}\bigg)
    \bigg(1 - \frac{\min\{k - 1,p\}}{p}\bigg).
\end{equation}
For $z \ge k$ we therefore have 
\begin{align}
 \begin{split}
  \label{Seq:ssscpf2} 
%   \textstyle 
  -\log \prods[p \not\equiv 1 \bmod 4][p \ge z \, p \nmid \det(\bh)]
    \delta_{\bz}(p)^{-k}\delta_{\bh}(p) 
%    \textstyle
 &
      \le 
      -\sum_{p \ge z} \log\bigg(1 - \frac{(k - 1)^2}{p^2}\bigg)
 \\
 & 
%     \textstyle 
     \le 
      (k - 1)^2\bigg(1 - \frac{(k - 1)^2}{z^2}\bigg)^{-1}
       \sum_{p \ge z}
        \frac{1}{p^2}
 \\
 & 
%   \textstyle 
    \ll 
     (k - 1)^2\bigg(1 - \frac{(k - 1)^2}{z^2}\bigg)^{-1}  
      \frac{1}{z\log z}.
 \end{split}
\end{align}
(The last bound follows from the bound 
$\sum_{p \le x} 1 \ll x/\log x$, $x \ge 2$, via partial 
summation.)
Upon exponentiating, we obtain the lower bound in 
\eqref{Seq:sssc1}.

Next, for $p \ge z \ge 3$ we have $\delta_{\bz}(p)^{-1} = 1 + 1/p$ 
and $\delta_{\bh}(p) \le 1$, so 
\[
%  \textstyle 
  \log \prods[p \not\equiv 1 \bmod 4][p \ge z, \, p \mid \det(\bh)] 
   \delta_{\bz}(p)^{-k}\delta_{\bh}(p) 
    \le 
     k
      \sums[p \ge z][p \mid \det(\bh)] \log\bigg(1 + \frac{1}{p}\bigg)
       \le
        \frac{k}{z}
         \sum_{p \mid \det(\bh)} 1. 
\]
The upper bound in \eqref{Seq:sssc2} follows upon exponentiating, 
after applying the bound 
$\sum_{p \mid n} 1 \ll \log n/\log\log n$, $n \ge 3$.
For $p \ge z \ge k$, we have, by \eqref{Seq:ssscpf1},  
\[
%  \textstyle 
  \delta_{\bz}(p)^{-k}\delta_{\bh}(p) 
   \ge 1 - \frac{(k - 1)^2}{p^2}
    \ge 1 - \frac{(k - 1)^2}{z^2} 
    > 0.
\]
For $z \ge \min\{3,k\}$ we therefore have, similarly to 
\eqref{Seq:ssscpf2}, 
\[
%  \textstyle 
 -\log \prods[p \not\equiv 1 \bmod 4][p \ge z, \, p \mid \det(\bh)]
   \delta_{\bz}(p)^{-k}\delta_{\bh}(p) 
    \le 
     \frac{(k - 1)^2}{z^2}
      \bigg(1 - \frac{(k - 1)^2}{z^2}\bigg)^{-1}
       \sum_{p \mid \det(\bh)} 1,
\]
giving the lower bound in \eqref{Seq:sssc2}.
\end{proof}
\end{nixnix}

%%%%%%%%%%%%%%%%%%%%%%%%%%%%%%%%%%%%%%%%%%%%%%%%%%%%%%%%%%%%%%%%%%
%%%%%%%%%%%%%%%%%%%%%%%%%%% REFERENCES %%%%%%%%%%%%%%%%%%%%%%%%%%%
%%%%%%%%%%%%%%%%%%%%%%%%%%%%%%%%%%%%%%%%%%%%%%%%%%%%%%%%%%%%%%%%%%

%****************************************************************%
%****************************************************************%
%**************************** JETSAM ****************************%
%****************************************************************%
%****************************************************************%

\begin{jetsam}

\section{Jetsam}
 \label{sec:jetsam}
 
%%%%%%%%%%%%%%%%%%%%%%%%%%%%%%%%%%%%%%%%%%%%%%%%%%%%%%%%%%%%%%%%%%
%%%%%%%%%%%%%%%%%%%%%%%%%%% HISTORY %%%%%%%%%%%%%%%%%%%%%%%%%%%%%%
%%%%%%%%%%%%%%%%%%%%%%%%%%%%%%%%%%%%%%%%%%%%%%%%%%%%%%%%%%%%%%%%%%
 
Sums of two squares are historically, perhaps, the most studied 
integers after the primes.
(See Volume II, Chapter VI of Dickson's {\em History of the Theory 
of Numbers} \cite{DIC:19}.)
The special case of Brahmagupta's identity, 
\[
 (a^2 + b^2)(c^2 + d^2) 
  =
   (ac + bd)^2 + (ad - bc)^2,
\]
is implicit in Diophantus' {\em Arithmetica}, and was also derived 
by Fibonacci (in his {\em Liber Quadratorum}), Bachet, and 
doubtless many others.
Fermat's theorem on sums of two squares, which was anticipated by 
Girard and first completely proved by Euler, states that all 
primes $p \equiv 1 \bmod 4$ are sums of two squares.
Because sums of two squares are not congruent to $3$ modulo $4$, 
we are led to the well-known characterization of nonzero sums of 
two squares as the integers whose canonical factorization is of 
the form  
\begin{equation}
 \label{eq:canfac}
   2^{\beta_2}
    \prod_{p \equiv 1 \bmod 4}p^{\beta_p}
     \prod_{p \equiv 3 \bmod 4}p^{2\beta_p}.
\end{equation}

This is a starting point for a proof of Landau's result 
\eqref{eq:sotsnt}.
A refinement of Landau's proof (see \cite[(4.6.4)]{HAR:40} or 
\cite[II.5, Theorem 3]{TEN:95})) shows that there exists a 
sequence of (explicitly given) numbers $D_1,D_2,\ldots$ such that, 
for all $N \ge 0$, we have
\begin{equation}
 \label{eq:sotsntii}
 \speccount(x)
  =
   \frac{Cx}{\sqrt{\log x}}
    \bigg\{
     1 + \frac{D_1}{\log x} + \cdots + \frac{D_N}{(\log x)^N} + 
       O_N\bigg(\frac{1}{(\log x)^{N + 1}}\bigg) 
    \bigg\}. 
\end{equation}
Unfortunately, $D_1,D_2,\ldots$ are difficult to evaluate 
numerically. 
As stated by Shanks \cite{SHA:64}, ``an unsolved problem of 
interest is to find [an approximation to $\speccount(x)$] that could 
be computed without undue difficulty by a {\em convergent} 
process, and which would be accurate to $O_N(x(\log x)^{-N})$ for 
[any given] $N$''.
Apropos of this, a rare ``mistake'' of Ramanujan is often 
recounted (see \cite[Chapter IV]{HAR:40}, \cite{STA:28, SHA:64}, 
or \cite[pp.\ xxiv--xxviii]{RAM:00} for instance). 

Unaware of Landau's work, Ramanujan claimed, in his very first   
correspondence with Hardy dated 1913 
(see \cite[pp.\ xxiv--xxviii]{RAM:00}), that 
\[
 \speccount(x)
  = 
   C\int_2^x \frac{\dd t}{\sqrt{\log t}} + \text{small error}, 
\]
where the ``small error'' is of order $\sqrt{x/\log x}$.
After partial integration, Ramanujan's claim would imply that 
$D_1 = 1/2$ in \eqref{eq:sotsntii}.
However, it turns out that $D_1 = 0.58194\ldots$, as was 
shown by Shanks \cite{SHA:64} (by correcting earlier work of  
Hardy's student Stanley \cite{STA:28}).
Even so, Ramanujan's incorrect estimate for $\speccount(x)$ is, as 
pointed out by Shanks \cite{SHA:64}, numerically a better 
approximation than $Cx/\sqrt{\log x}$ (because $D_1$ is close to 
$1/2$). 
We mention this lest the numerical evidence in support of 
Conjectures \ref{con:poisdist} and \ref{con:sotsktups} be not all 
that striking.
Also, we are circumspect in making more exact conjectures as to the 
nature of the error term implicit in \eqref{eq:sotsktups}.

In general, sequences like the sequence of sums of two squares, 
which are ``survivors'' of a sieving process of positive 
dimension, are expected to have many properties in common with the 
sequence of primes.
``Folklore'' conjectures entail analogs of the Hardy--Littlewood 
prime $k$-tuples conjecture for such sequences. 
Given the history outlined above, there is particular interest in 
the specific case of sums of two squares, yet it seems difficult 
to find an explicit formulation of a $k$-tuples conjecture for 
sums of two squares in the literature.

%%%%%%%%%%%%%%%%%%%%%%%%%%%%%%%%%%%%%%%%%%%%%%%%%%%%%%%%%%%%%%%%%%
%%%%%%%%%%%%%%%%%%%%%%% BASIC ESTIMATES %%%%%%%%%%%%%%%%%%%%%%%%%%
%%%%%%%%%%%%%%%%%%%%%%%%%%%%%%%%%%%%%%%%%%%%%%%%%%%%%%%%%%%%%%%%%%

\begin{proposition}
 \label{prop:S2est}
For $x \ge 2$ we have 
\[
  \sums[n \le x][n \in S_2] 1 = \frac{x}{2} + O(\log x).
\]
\end{proposition}

\begin{proof}
Let $x \ge 2$, and let $\alpha$ be the integer such that 
$2^{\alpha} \le x < 2^{\alpha + 1}$.
By Proposition \ref{prop:S2S3S1}, the positive integers $n$ in 
$S_2$ are precisely those of the form $2^{\beta}m$, where 
$\beta \ge 0$ and $m \equiv 1 \bmod 4$.
Thus, 
\begin{align*}
  \sums[n \le x][n \in S_2] 1
  & = 
   \sum_{\beta = 0}^{\alpha}
    \sums[m \le x/2^{\beta}][m \equiv 1 \bmod 4] 1
     =
      \frac{x}{4}\sum_{\beta = 0}^{\alpha} \frac{1}{2^{\beta}}
       + O(\alpha)
     =
      \frac{x}{2} - \frac{x}{2^{\alpha + 2}} + O(\alpha).
\end{align*}
Since $x/2^{\alpha + 2} < 1/2$ and $\alpha \le \log x/\log 2$, the 
result follows.
\end{proof}

\begin{proposition}
 \label{prop:Sp3est}
Let $p \equiv 3 \bmod 4$.
For $x \ge 1$ we have 
\[
   \sums[n \le x][n \in S_p] 1
    =
     x\bigg(1 + \frac{1}{p}\bigg)^{-1}
      + O\bigg(1 + \frac{\log x}{\log p}\bigg).
\]
\end{proposition}

\begin{proof}
First, suppose $1 \le x < p$.
By Proposition \ref{prop:S2S3S1}, 
$\{1,\ldots,p - 1\} \subseteq S_p$, so 
\[
  \sums[n \le x][n \in S_p] 1
  = 
   x + O(1)
    =
     x\bigg(1 + \frac{1}{p}\bigg)^{-1} + O(1).
\]
Now let $x \ge p$, and let $\alpha$ be the integer such that 
$p^{\alpha} \le x < p^{\alpha + 1}$.
By Proposition \ref{prop:S2S3S1}, the positive integers $n$ in 
$S_p$ are precisely those for which $p^{2\beta} \emid n$ for some 
$\beta \ge 0$.
Thus, 
\begin{align*}
 \sums[n \le x][n \in S_p] 1
  =
   \sum_{0 \le \beta \le \frac{\alpha}{2}}
    \Big\{
     \sums[n \le x][p^{2\beta} \mid n] 1
      -
       \sums[n \le x][p^{2\beta + 1} \nmid n] 1
    \Big\}
     =
      x\bigg(1 - \frac{1}{p}\bigg)
        \sum_{0 \le \beta \le \frac{\alpha}{2}}
         \frac{1}{p^{2\beta}}
         +
          O(\alpha).
\end{align*}
Since 
\[
 \sum_{0 \le \beta \le \frac{\alpha}{2}} \frac{1}{p^{2\beta}}
  =
   \bigg(1 - \frac{1}{p^2}\bigg)^{-1}
    \bigg(1 - \frac{p^{\alpha \bmod 2}}{p^{\alpha + 2}}\bigg)
   =
    \bigg(1 - \frac{1}{p}\bigg)^{-1}
     \bigg(1 + \frac{1}{p}\bigg)^{-1} 
   + O\bigg(\frac{1}{x}\bigg), 
\]
and since $\alpha \le \log x/\log p$, we see that  
\[
 \sums[n \le x][n \in S_p] 1
  = 
   x\bigg(1 + \frac{1}{p}\bigg)^{-1} 
   + O\bigg(\frac{\log x}{\log p}\bigg).
\]
The result follows by combining the estimates for $1 \le x < p$ 
and $x \ge p$.
\end{proof}

%%%%%%%%%%%%%%%%%%%%%%%%%%%%%%%%%%%%%%%%%%%%%%%%%%%%%%%%%%%%%%%%%%
%%%%%%%%%%%% AVERAGE OF SINGULAR SERIES FOR k = 2 %%%%%%%%%%%%%%%%
%%%%%%%%%%%%%%%%%%%%%%%%%%%%%%%%%%%%%%%%%%%%%%%%%%%%%%%%%%%%%%%%%%

\begin{lemma}
 \label{lem:sssak=2aux1}
Let $\bh = \{h_1,\ldots,h_k\}$ be a set of integers with 
$0 \le h_1 < \cdots < h_k \le y$.
If $y \ge \min\{\e^3,\e^k\}$ then 
\begin{equation}
 \label{eq:sotssingserconv3}
  \mathfrak{S}_{\bh}
   =
    \bigg(1 + O\bigg(\frac{k^3}{\log\log y}\bigg)\bigg)
     \prods[p \not\equiv 1 \bmod 4][p \le \log y] 
      \delta_{\bz}(p)^{-k}\delta_{\bh}(p).
\end{equation}
\end{lemma}

\begin{proof}
Let $z \ge \min\{3,k\}$ and suppose $y \le \e^z$.
Since $|\det(\bh)| \le y^{k^2} \le \e^{zk^2}$, we have 
$\log |\det(\bh)|/\log\log 3|\det(\bh)| \ll k^2z/\log z$.
Combining with \eqref{Seq:sssc1} and 
\eqref{Seq:sssc2} of 
Proposition \ref{Sprop:sssc}, we obtain 
\[
 \prods[p \equiv 3 \bmod 4][p > z]
  \delta_{\bz}(p)^{-k}\delta_{\bh}(p)
   =
    1 + O\bigg(\frac{k^3}{\log z}\bigg),
\]
from which \eqref{eq:sotssingserconv3} follows.
\end{proof}
 
Given a number $t \ge 1$ and a set of primes $\mathscr{P}$, let 
\begin{equation}
 \label{eq:defPsi}
  \Psi(t,\mathscr{P})
   \defeq 
%     \sums[n \le t][p \mid n \implies p \in \mathscr{P}] 1.
     \#\{n \le t : p \mid n \implies p \in \mathscr{P}\}.
\end{equation}
If $\mathscr{P}$ is the set of all primes less than or equal to 
some number $z$, then $\Psi(t,\mathscr{P})$ is the number of 
$z$-smooth numbers less than or equal to $t$, denoted $\Psi(t,z)$.

\begin{proposition}
 \label{prop:ssak=2pre}
Let $\mathscr{P}$ be any nonempty, finite subset of the primes 
that are not congruent to $1$ modulo $4$.
For $t \ge 1$, we have 
\[
  \sum_{g \le t}
   {\textstyle \prod_{p \in \mathscr{P}} } \delta_{\bz}(p)^{-2}\delta_{\{0,g\}}(p)
    =
     t 
      + O\Big(
              2^{\ecard\mathscr{P}}
               \Psi(t,\mathscr{P})  
                {\textstyle \prod_{p \in \mathscr{P}}\big(1 + \frac{1}{p}\big)}
         \Big).
\]
\end{proposition}

\begin{proof}
We begin with two observations, which we will use in the course of 
the proof.
First, let $g_1$ and $g_2$ be integers, and set $g = g_1g_2$.
Let $p$ be any prime and suppose $p \nmid g_1$, so that 
$\alpha \defeq \nu_p(g) = \nu_p(g_2)$.
Recall from Proposition \ref{Aprop:sssk=2} that in the case 
where $p \equiv 3 \bmod 4$, we have
\begin{equation*}
  \delta_{\bz}(p)^{-1}\delta_{\{0,g\}}(p)
   = 
    \delta_{\bz}(p)^{-1}\delta_{\{0,g_2\}}(p)
     =
       \bigg(1 - \frac{1}{p^{\alpha + 1}}\bigg),
\end{equation*}
while in the case where $p = 2$, we have 
\begin{equation*}
  \delta_{\bz}(2)^{-2}\delta_{\{0,g\}}(2)
   = 
    \delta_{\bz}(2)^{-2}\delta_{\{0,g_2\}}(2)
     =
        \begin{cases}
         1                                     & \alpha = 0    \\
         \frac{2^{\alpha + 1} - 3}{2^{\alpha}} & \alpha \ge 1.
        \end{cases}
\end{equation*}
Second, for $\beta \ge \alpha \ge 0$ (and any prime $p$) we have 
\begin{equation}
 \label{eq:ssak=2prepf1}
  \sum_{\alpha = 0}^{\beta}
   \frac{1}{p^{\alpha}}
    \bigg(1 - \frac{1}{p^{\alpha + 1}}\bigg)
     =
      \bigg(1 - \frac{1}{p^2}\bigg)^{-1}
      +
       O\bigg(\frac{1}{p^{\beta + 1}}\bigg);
\end{equation}
\begin{nixnix}
\begin{align*}
   \sum_{\alpha = 0}^{\beta}
    \frac{1}{p^{\alpha}}
     \bigg(1 - \frac{1}{p^{\alpha + 1}}\bigg)
 & = 
  \sum_{\alpha = 0}^{\infty} \frac{1}{p^{\alpha}}
   -
    \frac{1}{p}
     \sum_{\alpha = 0}^{\infty} \frac{1}{p^{2\alpha}}
      +
       O\bigg(\sum_{\alpha > \beta} \frac{1}{p^{\alpha}}\bigg)
 \\
 & = 
  \bigg(1 - \frac{1}{p}\bigg)^{-1}
   -
    \frac{1}{p}
     \bigg(1 - \frac{1}{p^2}\bigg)^{-1}
     +
      O\bigg(\frac{1}{p^{\beta + 1}}\bigg)
 \\
 & = 
  \bigg(1 - \frac{1}{p^2}\bigg)^{-1}
     +
      O\bigg(\frac{1}{p^{\beta + 1}}\bigg) 
\end{align*}
\end{nixnix}
we also have 
\begin{equation}
 \label{eq:ssak=2prepf2}
  \sum_{\alpha = 1}^{\beta}
   \frac{1}{2^{\alpha}}
    \cdot 
     \frac{2^{\alpha + 1} - 3}{2^{\alpha}}
    =
      1 + O\bigg(\frac{1}{2^{\beta}}\bigg).
\end{equation}
\begin{nixnix}
\begin{align*}
 \sum_{\alpha = 1}^{\beta}
  \frac{1}{2^{\alpha}}
   \cdot \frac{2^{\alpha + 1} - 3}{2^{\alpha}}
 & =
    \sum_{\alpha = 1}^{\beta}
     \bigg(\frac{2}{2^{\alpha}} - \frac{3}{2^{2\alpha}}\bigg)
 \\
 & = 
    \sum_{\alpha - 1 = 0}^{\infty} \frac{1}{2^{\alpha - 1}}
     -
     \frac{3}{4}
      \sum_{\alpha - 1 = 0}^{\infty} \frac{1}{4^{\alpha - 1}}
       +
        O\bigg(\sum_{\alpha > \beta} \frac{1}{2^{\alpha}}\bigg)
 \\
 & = 
    2 - \frac{3}{4}\bigg(\frac{3}{4}\bigg)^{-1} + O\bigg(\frac{1}{2^{\beta + 1}}\bigg) 
 \\
 & = 
     1 + O\bigg(\frac{1}{2^{\beta}}\bigg)
\end{align*}
\end{nixnix}

Suppose $\mathscr{P} = \{p_0,p_1,\ldots,p_r\}$, where $r \ge 0$ 
and $p_0,p_1,\ldots,p_r$ are distinct primes.
Suppose further that 
$p_1 \equiv \cdots \equiv p_r \equiv 3 \bmod 4$, while either 
$p_0 = 2$ or $p_0 \equiv 3 \bmod 4$ (to be specified).
If $g_1$ is any integer coprime with $p_0\cdots p_r$ then, by our 
first observation, 
\begin{equation*}
  \prod_{i = 0}^r
   \delta(p_i)^{-2}\delta_{\{0,g_1p_0^{\alpha_0}\cdots p_r^{\alpha_r}\}}(p_i)
    =
     \delta(p_0)^{-2}\delta_{\{0,p_0^{\alpha_0}\}}(p_0)
      \prod_{i = 1}^r\delta(p_i)^{-1}\bigg(1 - \frac{1}{p_i^{\alpha_i + 1}}\bigg).
\end{equation*}
\begin{nixnix}
\begin{equation*}
  \prod_{i = 0}^r
   \delta_{\{0,g_1p_0^{\alpha_0}\cdots p_r^{\alpha_r}\}}(p_i)
   =
     \prod_{i = 0}^r\delta_{\{0,p_i^{\alpha_i}\}}(p_i)
   =
        \delta_{\{0,p_0^{\alpha_0}\}}(p_0)
         \prod_{i = 1}^r\delta(p_i)\bigg(1 - \frac{1}{p_i^{\alpha_i + 1}}\bigg) 
\end{equation*}
\end{nixnix}
Now let $t \ge 1$, and for $0 \le j \le r$ let 
\[
 A_{r - j}
  \defeq 
   \{(\alpha_0,\ldots,\alpha_{r - j}) : p_0^{\alpha_0}\cdots p_{r - j}^{\alpha_{r - j}} \le t\}.
\]
We have 
\begin{align*}
  &
    \sum_{g \le t}
    {\textstyle \prod_{i = 0}^r } \delta(p_i)^{-2}\delta_{\{0,g\}}(p_i) 
  \\
  & = 
    \sum_{(\alpha_0,\ldots,\alpha_r) \in A_r}
     \sums[g_1 \le t/p_0^{\alpha_0}\cdots p_r^{\alpha_r}]
      {\textstyle \prod_{i = 0}^r } \delta(p_i)^{-2}\delta_{\{0,g_1p_0^{\alpha_0}\cdots p_r^{\alpha_r}\}}(p_i)
 \\
  & = 
   \bigg(\prod_{i = 1}^r \delta(p_i)^{-1}\bigg)
    \sum_{(\alpha_0,\ldots,\alpha_r) \in A_r}
     \delta(p_0)^{-2}\delta_{\{0,p_0^{\alpha_0}\}}(p_0)
      {\textstyle \prod_{i = 1}^r\big(1 - \frac{1}{p_i^{\alpha_i + 1}}\big)}
       \sums[g_1 \le t/p_0^{\alpha_0}\cdots p_r^{\alpha_r}] 1.
\end{align*}
To estimate this last sum we apply the sieve of 
Eratosthenes--Legendre: 
\begin{equation*}
  \sums[g_1 \le t/p_0^{\alpha_0}\cdots p_r^{\alpha_r}][(g_1,p_0\cdots p_r) = 1] 1
 =
   \frac{t}{p_0^{\alpha_0}\cdots p_r^{\alpha_r}}
   \prod_{i = 0}^r\bigg(1 - \frac{1}{p_i}\bigg)
    + 
     O(2^r).
\end{equation*}
\begin{nixnix}
\begin{align*}
 \sums[g_1 \le t/p_0^{\alpha_0}\cdots p_r^{\alpha_r}][(g_1,p_0\cdots p_r) = 1] 1
  & =
   \sum_{g_1 \le t/p_0^{\alpha_0}\cdots p_r^{\alpha_r}}
    \sums[d \mid g_1][d \mid p_0\cdots p_r] \mu(d)
 \\
  & = 
   \sum_{d \mid p_0\cdots p_r}
    \mu(d)
     \sums[g_1 \le t/p_0^{\alpha_0}\cdots p_r^{\alpha_r}][g_1 \equiv 0 \bmod d] 1
 \\
 & = 
   \sum_{d \mid p_0\cdots p_r}
    \mu(d)
     \bigg(\frac{t}{dp_0^{\alpha_0}\cdots p_r^{\alpha_r}} + O(1)\bigg) 
 \\
  & = 
   \frac{t}{p_0^{\alpha_0}\cdots p_r^{\alpha_r}}
    \sum_{d \mid p_0\cdots p_r} \frac{\mu(d)}{d}
    +
     O\bigg(\sum_{d \mid p_0\cdots p_r} |\mu(d)|\bigg)
 \\
 & = 
  \frac{t}{p_0^{\alpha_0}\cdots p_r^{\alpha_r}}
   \prod_{i = 0}^r\bigg(1 - \frac{1}{p_i}\bigg)
    + 
     O(2^r) 
\end{align*}
\end{nixnix}
Combining and recalling that 
$\delta(p_i)^{-1} = 1 + 1/p_i$ for $i = 1,\ldots,r$, we obtain 
\begin{align}
 \begin{split}
  \label{eq:ssak=2prepf3} 
 & 
   \sum_{g \le t}
   {\textstyle \prod_{i = 0}^r } \delta(p_i)^{-2}\delta_{\{0,g\}}(p_i) 
 \\
 & = 
  t\bigg(1 - \frac{1}{p_0}\bigg)
   \prod_{i = 1}^r \bigg(1 - \frac{1}{p_i^2}\bigg) 
    \sum_{(\alpha_0,\ldots,\alpha_r) \in A_r}
     \frac{\delta(p_0)^{-2}\delta_{\{0,p_0^{\alpha_0}\}}(p_0)}{p_0^{\alpha_0}}
      \prod_{i = 1}^r\frac{1}{p_i^{\alpha_i}}\bigg(1 - \frac{1}{p_i^{\alpha_i + 1}}\bigg)
 \\
 & \hspace{120pt} + 
     O\bigg(
       2^r(\card A_r)\prod_{i = 1}^r\bigg(1 + \frac{1}{p_i}\bigg) 
      \bigg).
  \end{split}
\end{align}
By repeating ($r$ times) an argument that uses 
\eqref{eq:ssak=2prepf1}, we verify that 
\begin{align}
 \begin{split}
  \label{eq:ssak=2prepf4}
 & 
  \prod_{i = 1}^r\bigg(1 - \frac{1}{p_i^2}\bigg)
   \sum_{(\alpha_0,\ldots,\alpha_r) \in A_r}
    \frac{\delta(p_0)^{-2}\delta_{\{0,p_0^{\alpha_0}\}}(p_0)}{p_0^{\alpha_0}}
     \prod_{i = 1}^r 
      \frac{1}{p_i^{\alpha_i}}
       \bigg(1 - \frac{1}{p_i^{\alpha_i + 1}}\bigg) 
  \\
 & \hspace{30pt} =
   \sum_{\alpha_0 \in A_0}
    \frac{\delta(p_0)^{-2}\delta_{\{0,p_0^{\alpha_0}\}}(p_0)}{p_0^{\alpha_0}}
     +
      O\bigg(\frac{1}{t}\sum_{j = 1}^{r - 1} \card A_{r - j}\bigg).
 \end{split}
\end{align}
\begin{nixnix}
Given $(\alpha_0,\ldots,\alpha_{r - 1}) \in A_{r - 1}$, let 
$\beta_r$ be the integer such that 
\[
 p^{\beta_r} 
  \le t/(p_0^{\alpha_0}\cdots p_{r - 1}^{\alpha_{r-1}}) 
   < p_r^{\beta_r + 1}.
\]
Below, the second equality follows from 
\eqref{eq:ssak=2prepf1}:  
\begin{align*}
 & 
 \sum_{(\alpha_0,\ldots,\alpha_r) \in A_r}
  \frac{\delta(p_0)^{-2}\delta_{\{0,p_0^{\alpha_0}\}}(p_0)}{p_0^{\alpha_0}}
   \prod_{i = 1}^r 
    \frac{1}{p_i^{\alpha_i}}
     \bigg(1 - \frac{1}{p_i^{\alpha_i + 1}}\bigg)
 \\
 & \hspace{10pt} = 
  \sum_{(\alpha_0,\ldots,\alpha_{r-1}) \in A_{r-1}}
   \frac{\delta(p_0)^{-2}\delta_{\{0,p_0^{\alpha_0}\}}(p_0)}{p_0^{\alpha_0}}
    \prod_{i = 1}^{r - 1} 
     \frac{1}{p_i^{\alpha_i}}
      \bigg(1 - \frac{1}{p_i^{\alpha_i + 1}}\bigg)
       \sum_{\alpha_r = 0}^{\beta_r}
        \frac{1}{p_r^{\alpha_r}}
         \bigg(1 - \frac{1}{p_r^{\alpha_r}}\bigg)
 \\
 & \hspace{10pt} = 
  \sum_{(\alpha_0,\ldots,\alpha_{r-1}) \in A_{r-1}}
   \frac{\delta(p_0)^{-2}\delta_{\{0,p_0^{\alpha_0}\}}(p_0)}{p_0^{\alpha_0}}
    \prod_{i = 1}^{r - 1} 
     \frac{1}{p_i^{\alpha_i}}
      \bigg(1 - \frac{1}{p_i^{\alpha_i + 1}}\bigg)
 \\
 & \hspace{180pt} \times 
       \bigg\{
        \bigg(1 - \frac{1}{p_r^2}\bigg)^{-1}
          + O\bigg(\frac{p_0^{\alpha_0}\cdots p_{r-1}^{\alpha_{r-1}}}{t} \bigg)
       \bigg\}
 \\
 & \hspace{10pt} =
    \bigg(1 - \frac{1}{p_r^2}\bigg)^{-1}
     \sum_{(\alpha_0,\ldots,\alpha_{r-1}) \in A_{r-1}}
      \frac{\delta(p_0)^{-2}\delta_{\{0,p_0^{\alpha_0}\}}(p_0)}{p_0^{\alpha_0}}
       \prod_{i = 1}^{r - 1} 
        \frac{1}{p_i^{\alpha_i}}
         \bigg(1 - \frac{1}{p_i^{\alpha_i + 1}}\bigg)
 \\
 & \hspace{210pt} + 
           O\bigg(\frac{\card A_{r-1}}{t}\bigg).
\end{align*}
Repeating this argument $r - 1$ more times, we obtain
\begin{align*}
 & 
  \prod_{i = 1}^r\bigg(1 - \frac{1}{p_i^2}\bigg)
   \sum_{(\alpha_0,\ldots,\alpha_r) \in A_r}
    \frac{\delta(p_0)^{-2}\delta_{\{0,p_0^{\alpha_0}\}}(p_0)}{p_0^{\alpha_0}}
     \prod_{i = 1}^r 
      \frac{1}{p_i^{\alpha_i}}
       \bigg(1 - \frac{1}{p_i^{\alpha_i + 1}}\bigg) 
  \\
 & \hspace{30pt} =
   \sum_{\alpha_0 \in A_0}
    \frac{\delta(p_0)^{-2}\delta_{\{0,p_0^{\alpha_0}\}}(p_0)}{p_0^{\alpha_0}}
     +
      O\bigg(\frac{1}{t}\sum_{j = 1}^{r - 1} \card A_{r - j}\bigg).
\end{align*}
\end{nixnix}
If $p_0 \equiv 3 \bmod 4$ then repeating the argument one more 
time leads to 
\begin{align}
 \begin{split}
  \label{eq:ssak=2prepf5}
 & 
  \bigg(1 - \frac{1}{p_0}\bigg)
   \prod_{i = 1}^r\bigg(1 - \frac{1}{p_i^2}\bigg)
    \sum_{(\alpha_0,\ldots,\alpha_r) \in A_r}
     \frac{\delta(p_0)^{-2}\delta_{\{0,p_0^{\alpha_0}\}}(p_0)}{p_0^{\alpha_0}}
      \prod_{i = 1}^r 
       \frac{1}{p_i^{\alpha_i}}
        \bigg(1 - \frac{1}{p_i^{\alpha_i + 1}}\bigg) 
  \\
 & \hspace{30pt} =
   1
     +
      O\bigg(\frac{1}{t}\sum_{j = 1}^{r} \card A_{r - j}\bigg).
 \end{split}
\end{align}
If $p_0 = 2$, then by a similar argument that uses 
\eqref{eq:ssak=2prepf2}, we obtain 
\begin{equation*}
   \bigg(1 - \frac{1}{p_0}\bigg)
    \sum_{\alpha_0 \in A_0}
     \frac{\delta(p_0)^{-2}\delta_{\{0,p_0^{\alpha_0}\}}(p_0)}{p_0^{\alpha_0}}
   = 
    1 + O\bigg(\frac{1}{2^{\beta_0}}\bigg).
\end{equation*}
\begin{nixnix}
Assume now that $p_0 = 2$.
Let $\beta_0$ be the integer such that 
$2^{\beta_0} \le t < 2^{\beta_0 + 1}$.
By \eqref{eq:ssak=2prepf2} we have 
\begin{equation*}
   \bigg(1 - \frac{1}{p_0}\bigg)
    \sum_{\alpha_0 \in A_0}
     \frac{\delta(p_0)^{-2}\delta_{\{0,p_0^{\alpha_0}\}}(p_0)}{p_0^{\alpha_0}}
   = 
   \frac{1}{2}
    \bigg\{ 
     1
      +
     \sum_{\alpha_0 = 1}^{\beta_0}
      \frac{1}{2^{\alpha_0}}
       \cdot 
        \frac{2^{\alpha_0 + 1} - 3}{2^{\alpha_0}}
    \bigg\}
   = 
    1 + O\bigg(\frac{1}{2^{\beta_0}}\bigg).
\end{equation*}
\end{nixnix}

Putting this into \eqref{eq:ssak=2prepf4}, we again obtain 
\eqref{eq:ssak=2prepf5}.
In any case, putting \eqref{eq:ssak=2prepf5} into 
\eqref{eq:ssak=2prepf3}, we finally obtain 
\[ 
  \sum_{1 \le g \le t}
  {\textstyle \prod_{i = 0}^r } \delta(p_i)^{-2}\delta_{\{0,g\}}(p_i) 
   =
    t 
     + 
      O\bigg(
        2^r(\card A_r)\prod_{i = 1}^r\bigg(1 + \frac{1}{p_i}\bigg) 
        + \sum_{j = 1}^r \card A_{r - j} 
       \bigg).
\]
As 
$
 \sum_{j = 1}^r \card A_{r - j} 
  \le r(\card A_r) 
   \le 2^r(\card A_r)
$, 
we can disregard this sum in the above $O$-term.
We complete the proof by noting that $A_r$ is in one-to-one 
correspondence with the integers $n \le t$ whose prime divisors 
all lie in $\mathscr{P} = \{p_0,\ldots,p_r\}$, and so  
$\card A_{r} = \Psi(t,\mathscr{P})$.
\end{proof}

\begin{proposition}
 \label{prop:sssak=2}
For $y \ge \e^3$ we have 
\[
  \sum_{0 \le h_1 < h_2 \le y}
   \mathfrak{S}_{\{h_1,h_2\}}
    =
     \frac{y^2}{2} 
      \bigg(1 + O\bigg(\frac{1}{\log\log y}\bigg)\bigg).
\]
\end{proposition}

\begin{proof}
It suffices to establish the estimate for integral $y$, so we 
assume for convenience that $y$ is an integer with $y > \e^3$.
By Lemma \ref{lem:sssak=2aux1}, we have 
\begin{equation}
 \label{eq:sssak=2pf1}
 \sum_{0 \le h_1 < h_2 \le y}
  \mathfrak{S}_{\{h_1,h_2\}}
   =
    \bigg(1 + O\bigg(\frac{1}{\log\log y}\bigg)\bigg)
     \sum_{0 \le h_1 < h_2 \le y}
      {\textstyle \prod_{p \in \mathscr{P}} }
       \delta_{\bz}(p)^{-2}\delta_{\{h_1,h_2\}}(p).
\end{equation}
For any prime $p$ and integers $h_1$ and $h_2$, 
$\delta_{\{h_1,h_2\}}(p)$ depends only on $\nu_p(h_2 - h_1)$, 
i.e.\ $\delta_{\{h_1,h_2\}}(p) = \delta_{\{0,h_2 - h_1\}}(p)$.
Using this, our assumption that $y$ is an integer, and partial 
summation, we verify that  
\[
  \sum_{0 \le h_1 < h_2 \le y}
   {\textstyle \prod_{p \in \mathscr{P}} }
     \delta_{\bz}(p)^{-2}\delta_{\{h_1,h_2\}}(p)
  =
   \int_1^y 
    \Big(\sum_{g \le t} 
    {\textstyle \prod_{p \in \mathscr{P}} }
      \delta_{\bz}(p)^{-2}\delta_{\{0,g\}}(p)\Big) \dd t.
\]
(In fact this holds for any integer $y \ge 1$ and any set of 
primes $\mathscr{P}$.)
\begin{nixnix}
Let $y$ be any number with $y \ge 1$, and let $\mathscr{P}$ be any 
set of primes.
In what follows, to ease notation we we set 
\[
 \Pi_{\mathscr{P},h_1,h_2}
  = 
   {\textstyle \prod_{p \in \mathscr{P}} } \delta_{\bz}(p)^{-2}\delta_{\{h_1,h_2\}}(p)
    \quad 
     \text{and}
      \quad 
       \Pi_{\mathscr{P},0,g}
      = 
        {\textstyle \prod_{p \in \mathscr{P}} } \delta_{\bz}(p)^{-2}\delta_{\{0,g\}}(p).
\]
We have 
\begin{align*}
 \sum_{0 \le h_1 < h_2 \le y}
  \Pi_{\mathscr{P},h_1,h_2}
 & =
    \sum_{g \le y}
     \sums[0 \le h_1 < h_2 \le y][h_2 - h_1 = g] 
      \Pi_{\mathscr{P},h_1,h_2}
 \\
 & = 
        \sum_{g \le y} 
         \Pi_{\mathscr{P},0,g}
          \sums[0 \le h_1 < h_2 \le y][h_2 - h_1 = g] 1
 \\
 & =  
            \sum_{g \le y} 
             \Pi_{\mathscr{P},0,g}
              \big(y - g + O(1)\big)
 \\
 & = 
                y\sum_{g \le y}
                  \Pi_{\mathscr{P},0,g}
                 -
                   \sum_{g \le y} 
                   g\Pi_{\mathscr{P},0,g}
                  +
                     O(y{\textstyle \prod_{p \in \mathscr{P}} } \delta_{\bz}(p)^{-2})
 \\
 & = 
  \int_1^y 
   \Big(\sum_{g \le t} \Pi_{\mathscr{P},0,g}\Big) \dd t
    + 
     O(y{\textstyle \prod_{p \in \mathscr{P}} } \delta_{\bz}(p)^{-2}),
\end{align*}
where the last equality is obtained by partial summation.
Note that if $y$ is an integer, then in the third equality we can 
replace $y - g + O(1)$ by $y - g$, eliminating the two subsequent 
$O$-terms.
\end{nixnix}

Applying Proposition \ref{prop:ssak=2pre} to the integrand, we 
see that 
\begin{equation}
 \label{eq:sssak=2pf2}
   \sum_{0 \le h_1 < h_2 \le y}
   {\textstyle \prod_{p \in \mathscr{P} }} \delta_{\bz}(p)^{-2}\delta_{\{h_1,h_2\}}(p)
 =
    \frac{y^2}{2} 
     +
      O\Big(
            2^{\ecard \mathscr{P}}
            {\textstyle \prod_{p \in \mathscr{P}}  \big(1 + \frac{1}{p}\big) }
              \int_1^y \Psi(t,\mathscr{P}) \dd t 
       \Big).
\end{equation}
\begin{nixnix}
{\small 
\begin{align*}
   \sum_{0 \le h_1 < h_2 \le y}
   {\textstyle \prod_{p \in \mathscr{P} }} \delta_{\bz}(p)^{-2}\delta_{\{h_1,h_2\}}(p)
 & =
     \int_1^y 
      \Big(\sum_{g \le t} {\textstyle \prod_{p \in \mathscr{P} }} \delta_{\bz}(p)^{-2}\delta_{\{0,g\}}(p)\Big) \dd t
 \\
 & =
    \int_1^y t \dd t  
     +
      O\Big(
          2^{\ecard \mathscr{P}}
         {\textstyle \prod_{p \in \mathscr{P}}  \big(1 + \frac{1}{p}\big) }
          \int_1^t \Psi(t,\mathscr{P}) \dd t 
       \Big)
 \\
 & =
    \frac{y^2}{2} 
     +
      O\Big(
          2^{\ecard \mathscr{P}}
         {\textstyle \prod_{p \in \mathscr{P}}  \big(1 + \frac{1}{p}\big) }
          \int_1^y \Psi(t,\mathscr{P}) \dd t 
       \Big)
\end{align*} 
}
\end{nixnix}
As $2^{\ecard \mathscr{P}}$ is far greater than 
$\prod_{p \in \mathscr{P}}\big(1 + \frac{1}{p}\big)$, the bound 
$
 \prod_{p \in \mathscr{P}}\big(1 + \frac{1}{p}\big)
  \le
   \log y
$
will suffice for this product.
By the prime number theorem for arithmetic progressions, we have 
\[
 \card \mathscr{P}
  =
   \frac{\log y}{2\log\log y}
    +
     O\bigg(\frac{\log y}{(\log\log y)^2}\bigg),
\]
so it follows that 
$
  \textstyle 
 2^{\ecard \mathscr{P}}
     \prod_{p \in \mathscr{P}}  \big(1 + \frac{1}{p}\big)  
  \ll
   y^{2/(5\log\log y)}  
$.
\begin{nixnix}
\[
  \textstyle 
 2^{\ecard \mathscr{P}}
     \prod_{p \in \mathscr{P}}  \big(1 + \frac{1}{p}\big)  
  \ll
   \exp\Big\{\frac{\log y}{\log\log y}\Big(\frac{\log 2}{2} + O\Big(\frac{1}{\log\log y}\Big)\Big)\Big\}.
\]
\end{nixnix}
Since $\Psi(t,\mathscr{P}) \le \Psi(t,\log y)$, and since the 
bound $\Psi(x,z) \ll x^{1 - 1/(2\log z)}$ holds uniformly for 
$x \ge z \ge 2$ (see \cite[III.5, Theorem 1]{TEN:95}), we have  
\[
 \int_{1}^y \Psi(t,\mathscr{P}) \dd t
  \ll
   \int_1^{y} t^{1 - 1/(2\log\log y)} \dd t
    \ll
     y^{2 - 1/(2\log\log y)}.
\]
\begin{nixnix}
For a constant $c > 0$ we have
\[
 \int t^{1 - \frac{1}{2c}} \dd t
  =
   \frac{2ct^{2 - \frac{1}{2c}}}{4c - 1} + \text{constant}.
\]
\end{nixnix}
Combining the last two bounds, we see that  
\begin{equation}
 \label{eq:sssak=2pf3}
 2^{\ecard \mathscr{P}}
    {\textstyle \prod_{p \in \mathscr{P}}  \big(1 + \frac{1}{p}\big) }
      \int_1^y \Psi(t,\mathscr{P}) \dd t
       \ll
        y^{2 - 1/(10\log\log y)}.
\end{equation}
\begin{nixnix}
Indeed, 
\[
 2^{\ecard \mathscr{P}}
    {\textstyle \prod_{p \in \mathscr{P}}  \big(1 + \frac{1}{p}\big) }
     \int_1^y \Psi(t,\mathscr{P}) \dd t
  \ll
   y^2
   \textstyle 
    \exp\Big\{-\frac{\log y}{\log\log y}\Big(\frac{1 - \log 2}{2} + O\Big(\frac{1}{\log\log y}\Big)\Big)\Big\},
\]
and $(1 - \log 2)/2 = 0.1534\ldots$.
\end{nixnix}

Combining \eqref{eq:sssak=2pf1}, \eqref{eq:sssak=2pf2} and 
\eqref{eq:sssak=2pf3} gives the result.
\end{proof}

%%%%%%%%%%%%%%%%%%%%%%%%%%%%%%%%%%%%%%%%%%%%%%%%%%%%%%%%%%%%%%%%%%
%%%%%%%%%%%%%%%%%%%%%%%%%% SECTION J5 %%%%%%%%%%%%%%%%%%%%%%%%%%%%
%%%%%%%%%%%%%%%%%%%%%%%%%%%%%%%%%%%%%%%%%%%%%%%%%%%%%%%%%%%%%%%%%%

\section{Proposition \ref{prop:sssa} with a weaker error term, and a proof}
 \label{Jsec:keyprop}

\begin{proposition}
 \label{Jprop:sssa}
Fix an integer $k \ge 1$ and a bounded convex set 
$\sC \subseteq \Delta^k$.
Set $\bo \defeq \emptyset$ or set $\bo \defeq \{0\}$.
For $y \ge 1$ we have 
\begin{align}
 \label{Jeq:sssa}
  \sum_{(h_1,\ldots,h_k) \in y\sC \cap \, \ZZ^k} 
   \mathfrak{S}_{\bo \cup \bh} 
   & =
    y^k \Big( \vol(\sC) + O_{k,\sC}\big(\e^{-\sqrt{\log y}}\big)\Big), 
\end{align}
where $\bh = \{h_1,\ldots,h_k\}$ in the summand and $\vol$ stands 
for volume in $\RR^k$.
\end{proposition}

The proof of Proposition \ref{Jprop:sssa} involves a 
basic lattice point counting argument.
Lemma \ref{Jlem:lip} is a special case of 
\cite[pp.\ 128--129]{LAN:94}.

\begin{lemma}
 \label{Jlem:lip}
Fix an integer $k \ge 1$ and a bounded convex set 
$\sC \subseteq \RR^k$.
For $y \ge 1$ we have 
$
 \#(y\sC \cap \ZZ^k)
  =
   y^k\vol(\sC) + O_{k,\sC}(y^{k - 1}).
$
\end{lemma}

To prove Proposition \ref{Jprop:sssa}, we 
express $\mathfrak{S}_{\bh}$ as a series.
To this end, for a nonempty, finite set $\bh \subseteq \ZZ$ with 
$\card \bh = k$, let 
\[
 \epsilon_{\bh}(p) 
  \defeq \delta_{\bz}(p)^{-k}\delta_{\bh}(p) - 1
  \quad 
   \text{and}
    \quad 
     \epsilon_{\bh}(d)
      \defeq 
       \textstyle
        \prod_{p \mid d} \epsilon_{\bh}(p)
\]
for integers $d$ in the set
\[ 
  \cD 
   \defeq 
    \{\text{$n \in \cD$ : $n$ squarefree and $p \mid n \implies p \not\equiv 1 \bmod 4$}\}.
\]
Note that $1 \in \cD$ and $\epsilon_{\bh}(1) \defeq 1$ by 
convention.
For $p \equiv 3 \bmod 4$ we have,
since
$\delta_{\bz}(p) = (1 + 1/p)^{-1}$ and $0 \le \delta_{\bh}(p) \le 1$,  
that $-1 \le \epsilon_{\bh}(p) \le (1 + 1/p)^k - 1 < 2^k/p$.
Since $p \le k - 1$ implies $p \mid \det(\bh)$, i.e.\ 
$1/p = (\det(\bh),p)/p^2$, we have 
\begin{equation}
 \label{Jeq:epshpbnd}
  |\epsilon_{\bh}(p)|
   \le 
    A_k\frac{(\det(\bh),p)}{p^2}
\end{equation}
for such $p$, where $A_k$ denotes, here and throughout Section 
\ref{Jsec:keyprop}, a sufficiently large (not necessarily optimal) 
number depending on $k$ --- possibly a different number each time.
Inequality \eqref{Jeq:epshpbnd} also holds for $p \ge k$, for in 
that case Proposition \ref{prop:Sp3h} (c) gives
$
 \delta_{\bh}(p) 
  \ge 
   \big(1 + \frac{1}{p}\big)^{-1}
    \big(1 - \frac{k - 1}{p}\big)
     > 
      0
$, 
and hence 
\[
 \epsilon_{\bh}(p)
  \ge 
   \bigg(1 + \frac{1}{p}\bigg)^{k-1}
    \bigg(1 - \frac{k - 1}{p}\bigg)
     -
      1
       \ge 
       -\frac{(k - 1)^2}{p^2}.
\]
Note that \eqref{Jeq:epshpbnd} trivially holds for $p = 2$.
Thus, 
\begin{equation}
 \label{Jeq:epsbndd}
  |\epsilon_{\bh}(d)|
   \le 
    A_k^{\omega(d)}\frac{(\det(\bh),d)}{d^2}
\end{equation}
for all $d \in \cD$.
Finally, notice (recall Definition \ref{def:Sss}) that 
\begin{equation}
 \label{Jeq:defsssS}
 \mathfrak{S}_{\bh}
  =
     \prod_{p \not\equiv 1 \bmod 4}
      \big(1 + \epsilon_{\bh}(p)\big)
  =
      \sum_{d \in \cD}
       \epsilon_{\bh}(d).
\end{equation}
The sum converges absolutely in view of \eqref{Jeq:epsbndd} and the 
following elementary bound, which we will also use in the proof 
of Proposition \ref{Jprop:sssa}.
Recall that for $n \in \NN$, 
$\omega(n) \defeq \#\{\text{$p$ prime : $p \mid n$}\}$ and  
$\rad(n) \defeq \prod_{p \mid n} p$. 

\begin{lemma} 
 \label{Jlem:omegabnd}
Fix any number $A$ satisfying $A \ge 1$.
For $x \ge 1$ we have, uniformly for nonzero integers $D$, the 
bound 
\[
 \sumss[\flat][n > x]
  A^{\omega(n)}
   \frac{(D,n)}{n^2}
    \ll_A
     \frac{(\log 3x)^{A - 1}}{x}
      \sumss[\flat][d \mid D] A^{\omega(D)},
\]
where $\sumsstxt[\flat]$ denotes summation restricted to 
squarefree integers.
\end{lemma}

\begin{proof}
Let $x \ge 1$.
We first consider the case $D = 1$.
Note that   
\begin{equation}
 \label{Jeq:mert}
 \sumss[\flat][n_1 \le x] 
  \frac{(A - 1)^{\omega(n_1)}}{n_1}
   \le 
    \prod_{p \le x}
     \bigg(1 + \frac{A - 1}{p}\bigg)
      \le 
       \prod_{p \le x} 
        \bigg(1 + \frac{1}{p}\bigg)^{A - 1}
         \ll_A
          (\log 3x)^{A - 1},
\end{equation}
because 
$1 + 1/p < \e^{1/p}$ and 
$\sum_{p \le x} 1/p = \log\log 3x + O(1)$ by one of Mertens' 
theorems.
Now, 
\[
 \sumss[\flat][n > x] \frac{A^{\omega(n)}}{n^2}
  =
   \sumss[\flat][n > x] 
    \frac{1}{n^2}
     \sum_{n_1 \mid n} (A - 1)^{\omega(n_1)}
      \le 
       \sumss[\flat][n_1 \ge 1]
        \frac{(A - 1)^{\omega(n_1)}}{n_1^2}
         \sumss[\flat][m > x/n_1]
          \frac{1}{m^2},
\]
the inner sum being $O(n_1/x)$ for $n_1 \le x$ and $O(1)$ for 
$n_1 > x$. 
Thus, 
\[
 \sumss[\flat][n > x] \frac{A^{\omega(n)}}{n^2}
  \ll_A
   \frac{(\log 3x)^{A - 1}}{x}
    +
     \sumss[\flat][n_1 > x]
      \frac{(A - 1)^{\omega(n_1)}}{n_1^2}.
\]
If $A - 1 \le 1$ then this last sum is $O(1/x)$; otherwise, 
repeating the argument as many times as necessary gives
\[
 \sumss[\flat][n > x] \frac{A^{\omega(n)}}{n^2}
  \ll_A
   \frac{(\log 3x)^{A - 1}}{x}.
\]
It is straightforward to deduce from this that for any nonzero 
integer $d$, 
\[
 \sumss[\flat][n > x][d \mid n] \frac{A^{\omega(n)}}{n^2}
  \ll_A
   \frac{A^{\omega(d)}}{d}\cdot 
    \frac{(\log 3x)^{A - 1}}{x}.
\]
Letting $D$ be any nonzero integer, we trivially have 
$(D,n) \le \sum_{d \mid D, \, d \mid n} d$, hence
\[
 \sumss[\flat][n > x]
  A^{\omega(n)} \frac{(D,n)}{n^2}
   \le 
    \sumss[\flat][d \mid D] d 
     \sumss[\flat][n > x][d \mid n] \frac{A^{\omega(n)}}{n^2}
      \ll_A 
       \frac{(\log 3x)^{A - 1}}{x}
        \sumss[\flat][d \mid D] A^{\omega(D)}.
\]
\end{proof}

\begin{lemma}
 \label{Jlem:dethap}
Fix an integer $k \ge 1$, and a squarefree integer $d \ge 1$.
For $y \ge 1$ we have
\[
 \underset{d \mid \det(\{0,h_1,\ldots,h_k\})}
  { 
   \sum_{0 < h_1 < \cdots < h_k \le y}
  } 1
  \le 
   k^{2\omega(d)}
    \bigg(\frac{y^k}{d} + O_k(y^{k - 1})\bigg).
\]
\end{lemma}

\begin{proof}
Let $h_0 = 0,h_1,\ldots,h_k$ be pairwise distinct integers and  
suppose $d$ divides $\prod_{0 \le i < j \le k}(h_i - h_j)$.
Then, since $d$ is squarefree, there exist pairwise coprime 
positive integers $d_{ij}$ such that 
$d = \prod_{0 \le i < j \le k} d_{ij}$ and 
$d_{ij} \mid h_i - h_j$, $0 \le i < j \le k$.
Therefore, 
\[
 \underset{d \mid \det(\{h_0,h_1,\ldots,h_k\})}
  { 
   \sum_{0 < h_1 < \cdots < h_k \le y}
  } 1
   \le 
    \sums[d = d_{01}\cdots d_{(k-1)k}] 
     \hspace{5pt}
     \underset{0 \le i < j \le k - 1 \implies d_{ij} \mid h_i - h_j}
    { 
       \sum_{h_1 \in I_y}
        \sum_{h_2 \in I_y}
         \cdots 
          \sum_{h_{k - 1} \in I_y}
    } 
     \hspace{5pt}
      \sums[h_k \in I_y][0 \le i \le k - 1 \implies d_{ik} \mid h_i - h_k] 1,
\]
where on the right-hand side, the outermost sum is over all 
decompositions of $d$ as a product of $\binom{k + 1}{2}$ positive 
integers, and $I_y \defeq (0,y]$.

Consider the decomposition $d = d_{01}\cdots d_{(k - 1)k}$.
Let us define $d_{j} \defeq \prod_{i = 0}^{j - 1} d_{ij}$ for 
$j = 1,\ldots,k$.
Notice that $d = \prod_{j = 1}^k d_j$.
By the Chinese remainder theorem, the condition on $h_k$ in the 
innermost sum above is equivalent to $h_k$ being in some 
congruence class modulo $d_k$, uniquely determined by 
$h_0,h_1,\ldots,h_{k - 1}$.
The sum is therefore equal to $y/d_k + O(1)$.
Iterating this argument $k$ times we see that the inner sum over 
$h_1,\ldots,h_k$ is equal to  
\[
  \prod_{j = 1}^k
   \bigg(\frac{y}{d_j} + O(1)\bigg)
    =
     \frac{y^k}{d} + O_k(y^{k - 1}).
\]
The result follows by combining and noting that, since $d$ is 
squarefree, the number of ways of writing $d$ as a product of 
$\binom{k + 1}{2}$ positive integers is 
$\binom{k + 1}{2}^{\omega(d)}$, and that 
$\binom{k + 1}{2} \le k^2$.
\end{proof}

Before the final lemma, let us introduce one more piece of 
notation.
For $\alpha \ge 1$, $\T_{\bz}(2^{\alpha + 1})$ and 
$\V_{\bz}(p^{\alpha})$ are nonempty.
For integers $j \ge 1$, let  
\[
 \upsilon_{\bh}(2^{\alpha + 1};j) 
  \defeq 
   \bigg(\frac{\card \T_{\bz}(2^{\alpha + 1})}{2^{\alpha + 1}}\bigg)^{-j}
    \bigg(\frac{\card \T_{\bh}(2^{\alpha + 1})}{2^{\alpha + 1}}\bigg) - 1, 
\]
and for $p \equiv 3 \bmod 4$, define
\[
 \upsilon_{\bh}(p^{\alpha};j) 
  \defeq 
   \bigg(\frac{\card \V_{\bz}(p^{\alpha})}{p^{\alpha}}\bigg)^{-j}
    \bigg(\frac{\card \V_{\bh}(p^{\alpha})}{p^{\alpha}}\bigg) - 1.  
\]
Note that in the exponent we have the parameter $j$, not $k$.
As one might expect, $\upsilon_{\bh}(p^{\alpha};k)$ is a good 
approximation to $\epsilon_{\bh}(p)$ when $\alpha$ is large, and 
we can take advantage of quite strong cancellation in summing 
$\upsilon_{\bh}(p^{\alpha};k)$ rather than $\epsilon_{\bh}(p)$.

\begin{lemma}
 \label{Jlem:cancel}
\textup{(}a\textup{)}
Set $\bo \defeq \emptyset$ or set $\bo \defeq \{0\}$.
Let $(\alpha_p)_{p \not\equiv 1 \bmod 4}$ be a sequence of 
positive integers with $\alpha_2 \ge 2$.
Let $d > 1$ be an integer, none of whose prime divisors 
are congruent to $1$ modulo $4$, and let $R_1,\ldots,R_k$ be 
complete residue systems modulo $\prod_{p \mid d} p^{\alpha_p}$.

\textup{(}a\textup{)}
We have 
\[
 \sum_{h_1 \in R_1} 
  \cdots 
   \sum_{h_k \in R_k}
    \prod_{p \mid d} \upsilon_{\bo \cup \bh}(p^{\alpha_p};\ocard \bo + k)
      =
       0,
\]
where $\bh = \{h_1,\ldots,h_k\}$ in the summand. 
\textup{(}Note that we may have $\card \bh < k$ here.\textup{)}

\textup{(}b\textup{)}
Let $\bh = \{h_1,\ldots,h_k\}$ be a set of integers that are not 
necessarily distinct.
For $p \equiv 3 \bmod 4$ and $\alpha \ge 1$, we have
\begin{equation}
 \label{Jeq:upsbnd1}
 |\upsilon_{\bh}(p^{\alpha};k)|
  \le 
   A_k
    \bigg(
     \frac{k - \card \bh}{p}
      + \frac{(\det(\bh),p)}{p^2}
    \bigg) 
\end{equation}
and 
\begin{equation}
 \label{Jeq:upsbnd2}
 |\epsilon_{\bh}(p) - \upsilon_{\bh}(p^{\alpha};k)|
  \le 
   A_k\bigg(\frac{k - \card \bh}{p} + \frac{1}{p^{\alpha + \alpha \bmod 2}}\bigg),
\end{equation}
where $A_k$ is a \textup{(}sufficiently large\textup{)} quantity 
depending on $k$.
The bounds \eqref{Jeq:upsbnd1} and \eqref{Jeq:upsbnd2} also hold for 
$p = 2$, provided $\alpha \ge 2$.
\end{lemma}

\begin{proof}
(a)
Let $\bh = \{h_1,\ldots,h_k\}$ and $\bh' = \{h_1',\ldots,h_k'\}$  
satisfy $h_i \equiv h_i' \bmod p^{\alpha}$ for $i = 1,\ldots,k$. 
If $p \equiv 3 \bmod 4$, it is clear from \eqref{eq:defVh} that 
$
 \card \V_{\bo \cup \bh}(p^{\alpha}) 
  = 
   \card \V_{\bo \cup \bh'}(p^{\alpha})
$, 
and hence 
$
 \upsilon_{\bo \cup \bh}(p^{\alpha};\ocard \bo + k)
 = 
  \upsilon_{\bo \cup \bh'}(p^{\alpha};\ocard \bo + k)
$.
Similarly, we have (from \eqref{eq:defTh}) that 
$
 \upsilon_{\bo \cup \bh}(2^{\alpha + 1};\ocard \bo + k) 
  = 
   \upsilon_{\bo \cup \bh'}(2^{\alpha + 1};\ocard \bo + k)
$.
Thus, by the Chinese remainder theorem, 
\[
 \sum_{h_1 \in R_1} 
  \cdots 
   \sum_{h_k \in R_k}
    \prod_{p \mid d} \upsilon_{\bo \cup \bh}(p^{\alpha_p};\ocard \bo + k)
      =
       \prod_{p \mid d} 
        \bigg(
         \sum_{h_1 \in \ZZ_{p^{\alpha_p}}}
          \cdots 
           \sum_{h_k \in \ZZ_{p^{\alpha_p}}}
            \upsilon_{\bo \cup \bh}(p^{\alpha_p};\ocard \bo + k) 
         \bigg),
\]
where $\bh = \{h_1,\ldots,h_k\}$ in both summands, and 
$\ZZ_{p^{\alpha_p}} \defeq \{0,\ldots,p^{\alpha_p} - 1\}$.
For $p \equiv 3 \bmod 4$ we have 
\[
 \sum_{h_1 \in \ZZ_{p^{\alpha_p}}}
  \cdots 
   \sum_{h_k \in \ZZ_{p^{\alpha_p}}}
    \card \V_{\bh}(p^{\alpha_p})
   =
   \sum_{a \in \ZZ_{p^{\alpha_p}}}
    \sums[h_1 \in \ZZ_{p^{\alpha_p}}][a + h_1 \in S_p][\nu_p(a + h_1) < \alpha_p]
     \cdots 
      \sums[h_k \in \ZZ_{p^{\alpha_p}}][a + h_k \in S_p][\nu_p(a + h_k) < \alpha_p] 1,
\]
as can be seen by applying the definition \eqref{eq:defVh} of 
$\V_{\bh}(p^{\alpha})$ and changing the order of summation.
For $i = 1,\ldots,k$, each sum over $h_i$ on the right-hand side 
enumerates a translation of $\V_{\bz}(p^{\alpha_p})$, so the 
entire sum (i.e.\ the left-hand side) is equal to 
$p^{\alpha_p}(\card \V_{\bz}(p^{\alpha_p}))^k$.
%
%****************************************************************%
%************************* START DETAIL *************************%
%****************************************************************%
%
\begin{nixnix}
\begin{align*}
  \sum_{h_1 \in \ZZ_{p^{\alpha}}}
   \cdots 
    \sum_{h_k \in \ZZ_{p^{\alpha}}}
     \card \V_{\bh}(p^{\alpha})
  & = 
      \sum_{h_1 \in \ZZ_{p^{\alpha}}}
       \cdots 
        \sum_{h_k \in \ZZ_{p^{\alpha}}}
         \sums[a \in \ZZ_{p^{\alpha}}][\forall i, a + h_i \in S_p][\forall i, \nu_p(a + h_i) < \alpha] 1
 \\ 
  & = 
   \sum_{a \in \ZZ_{p^{\alpha}}}
    \sums[h_1 \in \ZZ_{p^{\alpha}}][a + h_1 \in S_p][\nu_p(a + h_1) < \alpha]
     \cdots 
      \sums[h_k \in \ZZ_{p^{\alpha}}][a + h_k \in S_p][\nu_p(a + h_k) < \alpha] 1
 \\
 & = 
   \sum_{a \in \ZZ_{p^{\alpha}}}
    \sums[a + h_1 \in \ZZ_{p^{\alpha}}][a + h_1 \in S_p][\nu_p(a + h_1) < \alpha]
     \cdots 
      \sums[a + h_k \in \ZZ_{p^{\alpha}}][a + h_k \in S_p][\nu_p(a + h_k) < \alpha] 1   
 \\
 & = 
    \sum_{a \in \ZZ_{p^{\alpha}}}
     (\card \V_{\bz}(p^{\alpha}))
      \cdots 
       (\card \V_{\bz}(p^{\alpha}))
 \\
 & = 
    p^{\alpha}(\card \V_{\bz}(p^{\alpha}))^k
\end{align*}
\end{nixnix}
%
%****************************************************************%
%************************** END DETAIL **************************%
%****************************************************************%
%
Since 
\[
 \upsilon_{\bh}(p^{\alpha_p};k) 
  \defeq 
   \big[
    \big(\card \V_{\bz}(p^{\alpha_p})/p^{\alpha_p}\big)^{-k}
     \big(\card \V_{\bh}(p^{\alpha_p})/p^{\alpha_p}\big) 
   \big] - 1,
\]
it follows that  
\[
 \sum_{h_1 \in \ZZ_{p^{\alpha_p}}}
  \cdots 
   \sum_{h_k \in \ZZ_{p^{\alpha_p}}}
    \upsilon_{\bh}(p^{\alpha_p};k)
   =
   0.
\]
In a similar fashion we obtain
\[
 \sum_{h_1 \in \ZZ_{p^{\alpha_p}}}
  \cdots 
   \sum_{h_k \in \ZZ_{p^{\alpha_p}}}
    \upsilon_{\{0\} \cup \bh}(p^{\alpha_p};1 + k)
   =
   0. 
\]
%
%****************************************************************%
%************************* START DETAIL *************************%
%****************************************************************%
%
\begin{nixnix}
Similarly, 
\[
 \sum_{h_1 \in \ZZ_{p^{\alpha_p}}}
  \cdots 
   \sum_{h_k \in \ZZ_{p^{\alpha_p}}}
    \card \V_{\{0\} \cup \bh}(p^{\alpha_p})
   =
   \sums[a \in \ZZ_{p^{\alpha_p}}][a \in S_p][\nu_p(a) < \alpha_p]
    \sums[h_1 \in \ZZ_{p^{\alpha_p}}][a + h_1 \in S_p][\nu_p(a + h_1) < \alpha_p]
     \cdots 
      \sums[h_k \in \ZZ_{p^{\alpha_p}}][a + h_k \in S_p][\nu_p(a + h_k) < \alpha_p] 1,
\]
which is equal to $\big(\card \V_{\bz}(p^{\alpha_p})\big)^{1 + k}$, 
and since 
\[
  \upsilon_{\{0\} \cup \bh}(p^{\alpha_p};1 + k) 
  \defeq 
   \big[
    \big(\card \V_{\bz}(p^{\alpha_p})/p^{\alpha_p}\big)^{-(1 + k)}
     \big(\card \V_{\{0\} \cup \bh}(p^{\alpha_p})/p^{\alpha_p}\big) 
   \big] - 1,
\]
it follows that 
\[
 \sum_{h_1 \in \ZZ_{p^{\alpha_p}}}
  \cdots 
   \sum_{h_k \in \ZZ_{p^{\alpha_p}}}
    \upsilon_{\{0\} \cup \bh}(p^{\alpha_p};1 + k)
   =
   0 
\]
as well.
\end{nixnix}
%
%****************************************************************%
%************************** END DETAIL **************************%
%****************************************************************%
%
An analogous argument gives the same results for $p = 2$ 
($\alpha_2 \ge 2$).
%
%****************************************************************%
%************************* START DETAIL *************************%
%****************************************************************%
%
\begin{nixnix}
Applying the definition \eqref{eq:defTh} of 
$\T_{\bh}(2^{\alpha + 1})$ and changing the order of summation 
yields 
\begin{align*}
  \sum_{h_1 \in \ZZ_{2^{\alpha + 1}}}
   \cdots 
    \sum_{h_k \in \ZZ_{2^{\alpha + 1}}}
     \card \T_{\bh}(2^{\alpha + 1})
  & = 
      \sum_{h_1 \in \ZZ_{2^{\alpha + 1}}}
       \cdots 
        \sum_{h_k \in \ZZ_{2^{\alpha + 1}}}
         \sums[a \in \ZZ_{2^{\alpha + 1}}][\forall i, a + h_i \in S_2][\forall i, \nu_2(a + h_i) < \alpha] 1
 \\ 
  & = 
   \sum_{a \in \ZZ_{2^{\alpha + 1}}}
    \sums[h_1 \in \ZZ_{2^{\alpha + 1}}][a + h_1 \in S_2][\nu_2(a + h_1) < \alpha]
     \cdots 
      \sums[h_k \in \ZZ_{2^{\alpha + 1}}][a + h_k \in S_2][\nu_2(a + h_k) < \alpha] 1
 \\
 & = 
   \sum_{a \in \ZZ_{2^{\alpha + 1}}}
    \sums[a + h_1 \in \ZZ_{2^{\alpha + 1}}][a + h_1 \in S_2][\nu_2(a + h_1) < \alpha]
     \cdots 
      \sums[a + h_k \in \ZZ_{2^{\alpha + 1}}][a + h_k \in S_2][\nu_2(a + h_k) < \alpha] 1   
 \\
 & = 
    \sum_{a \in \ZZ_{2^{\alpha + 1}}}
     (\card \T_{\bz}(2^{\alpha + 1}))
      \cdots 
       (\card \T_{\bz}(2^{\alpha + 1}))
 \\
 & = 
    2^{\alpha + 1}(\card \T_{\bz}(2^{\alpha + 1}))^k.
\end{align*}
Since 
$
 \upsilon_{\bh}(2^{\alpha + 1}) 
  \defeq 
   \big[
    \big(\card \T_{\bz}(2^{\alpha + 1})/2^{\alpha + 1}\big)^{-k}
     \big(\card \T_{\bh}(2^{\alpha + 1})/2^{\alpha + 1}\big) 
   \big] - 1
$, 
it follows that 
\[
 \sum_{h_1 \in \ZZ_{2^{\alpha + 1}}}
  \cdots 
   \sum_{h_k \in \ZZ_{2^{\alpha + 1}}}
    \upsilon_{\bh}(2^{\alpha + 1})
   =
   0.
\]
\end{nixnix}
%
%****************************************************************%
%************************** END DETAIL **************************%
%****************************************************************%

(b)
Let $\alpha \ge 1$ and $p \equiv 3 \bmod 4$.
Define $\eta_{\bh}(p^{\alpha})$ and $\kappa_{\bh}(p)$ as the 
numbers given by the relations
\[
 \frac{\card \V_{\bh}(p^{\alpha})}{p^{\alpha}}
  \eqdef 
   \delta_{\bh}(p) + \eta_{\bh}(p^{\alpha}) 
    \quad 
     \text{and}
      \quad 
 \delta_{\bh}(p) 
  \eqdef
   \bigg(1 + \frac{1}{p}\bigg)^{-1}
    \bigg(1 - \frac{\kappa_{\bh}(p)}{p}\bigg).
\]
Note that by Proposition \ref{prop:Sp3h}, \eqref{eq:Vhdeltp3bnd} 
and part (c),  
$
 |\eta_{\bh}(p^{\alpha})| 
  < 
   2(\card \bh)/p^{\alpha + (\alpha \bmod 2)}
$ and 
$\kappa_{\bh}(p) \le \min\{\card \bh - 1,p\}$, with 
$\kappa_{\bh}(p) = \card \bh - 1$ if $p \nmid \det(\bh)$.
Also, $\kappa_{\bh}(p) \ge -1$ (because $\delta_{\bh}(p) \le 1$).
Since $\card \bh \le k$ and $\alpha + (\alpha \bmod 2) \ge 2$, we 
have 
\[
 \frac{\card \V_{\bh}(p^{\alpha})}{p^{\alpha}}
  =
   \bigg(1 + \frac{1}{p}\bigg)^{-1}
    \bigg(1 - \frac{\kappa_{\bh}(p)}{p} + O\bigg(\frac{k}{p^2}\bigg)\bigg).
\]
In the special case $\bh = \{0\}$ we can take 
$\kappa_{\bh}(p) = 0$.
We therefore have  
\begin{align*}
 \bigg(\frac{\card \V_{\bz}(p^{\alpha})}{p^{\alpha}}\bigg)^{-k}
  \frac{\card \V_{\bh}(p^{\alpha})}{p^{\alpha}}
  &
   =
    \bigg(1 + \frac{1}{p}\bigg)^{k - 1}
     \bigg[
       \bigg(1 - \frac{\kappa_{\bh}(p)}{p} + O_k\bigg(\frac{1}{p^2}\bigg)\bigg)
     \bigg]
   \\
  &
   =
     \bigg(1 + \frac{k - 1}{p} +O_k\bigg(\frac{1}{p^2}\bigg)\bigg)
      \bigg[1 - \frac{\kappa_{\bh}(p)}{p} + O_k\bigg(\frac{1}{p^2}\bigg)\bigg]
 \\
   & 
    = 
      1 + \frac{k - 1 - \kappa_{\bh}(p)}{p} + O_k\bigg(\frac{1}{p^2}\bigg).
\end{align*}

\noindent 
Thus, writing 
$\kappa_{\bh}(p) \eqdef \card \bh - 1 - \xi_{\bh}(p)$, say, we have 
\[
 |\upsilon_{\bh}(p^{\alpha};k)|
  \le 
   \frac{k - 1 - \kappa_{\bh}(p)}{p} + \frac{A_k}{p^2}
    =
     \frac{k - \card \bh}{p} + \frac{\xi_{\bh}(p)}{p} + \frac{A_k}{p^2}.
\]
(Recall that we use $A_k$ to denote a sufficiently large number 
depending on $k$.)
If $p \mid \det(\bh)$ then 
$\xi_{\bh}(p)/p = \xi_{\bh}(p)(\det(\bh),p)/p^2$, and if 
$p \nmid \det(\bh)$ then, as already noted, 
$\kappa_{\bh}(p) = \card \bh - 1$, i.e.\ $\xi_{\bh}(p) = 0$, so 
$\xi_{\bh}(p)/p = \xi_{\bh}(p)(\det(\bh),p)/p^2$ in any case.
Since, as already noted, 
$-1 \le \kappa_{\bh}(p) \le \card \bh - 1$, 
we have $0 \le \xi_{\bh}(p) \le \card \bh \le k$, so we may write 
$|\xi_{\bh}(p)| \le A_k$.
Hence
\[
 |\upsilon_{\bh}(p^{\alpha};k)|
  \le 
   \frac{k - \card \bh}{p} + A_k\frac{(\det(\bh),p)}{p^2} + \frac{A_k}{p^2}
    \le
     \frac{k - \card \bh}{p} + A_k\frac{(\det(\bh),p)}{p^2},
\]
giving \eqref{Jeq:upsbnd1}.

As can be seen from Proposition \ref{prop:Sp3h}, 
\eqref{eq:delthp3} and part (c), we in fact have
\[
 \frac{\card \V_{\bz}(p^{\alpha})}{p^{\alpha}}
  =
   \bigg(1 + \frac{1}{p}\bigg)^{-1}
    \bigg(1 - \frac{1}{p^{\alpha + \alpha \bmod 2}}\bigg).
\]
Letting $j = \card \bh$ so that 
$
 \epsilon_{\bh}(p) 
  = \delta_{\bz}(p)^{-j}\delta_{\bh}(p)
  = (1 + 1/p)^{j}\delta_{\bh}(p)
$, 
we see that
{\small 
\begin{align*}
 & 
 \epsilon_{\bh}(p) - \upsilon_{\bh}(p^{\alpha};k)
  \\ 
 & \hspace{15pt}
  =   
    \bigg(1 + \frac{1}{p}\bigg)^{j}
     \bigg(\frac{\card \V_{\bh}(p^{\alpha})}{p^{\alpha}} - \eta_{\bh}(p^{\alpha})\bigg)
      -
       \bigg(1 + \frac{1}{p}\bigg)^{k}
        \bigg(1 - \frac{1}{p^{\alpha + \alpha \bmod 2}}\bigg)^{-k}
         \frac{\card \V_{\bh}(p^{\alpha})}{p^{\alpha}}
  \\
 & \hspace{15pt}
  =   
   \bigg(1 + \frac{1}{p}\bigg)^{j}
    \frac{\card \V_{\bh}(p^{\alpha})}{p^{\alpha}}
     \bigg\{1 - \bigg(1 + \frac{1}{p}\bigg)^{k - j}\bigg(1 - \frac{1}{p^{\alpha + \alpha \bmod 2}}\bigg)^{-k}\bigg\}
      - \bigg(1 + \frac{1}{p}\bigg)^{j}\eta_{\bh}(p^{\alpha}). 
\end{align*}
}

\noindent 
Now, $(1 + 1/p)^{j} \le A_k$, 
$\card \V_{\bh}(p^{\alpha})/p^{\alpha} \le 1$, 
$\eta_{\bh}(p^{\alpha}) \le A_k/p^{\alpha + \alpha \bmod 2}$ (see 
above), and the term $\{\cdots\}$ in brackets is at most 
$(k - j)/p + A_k/p^2$, but of course if $k = j$ 
then it is at most $A_k/p^{\alpha + \alpha \bmod 2}$.
Hence \eqref{Jeq:upsbnd2}.

The case for \eqref{Jeq:upsbnd1} and \eqref{Jeq:upsbnd2} with 
$p = 2$ ($\alpha \ge 2)$ is similar.
%
%****************************************************************%
%************************* START DETAIL *************************%
%****************************************************************%
%
\begin{nixnix}
To do...
\end{nixnix}
%
%****************************************************************%
%************************** END DETAIL **************************%
%****************************************************************%
%
\end{proof}

\begin{proof}[Proof of Proposition \ref{Jprop:sssa}]
Fix an integer $k \ge 1$ and a bounded convex set 
$\sC \subseteq \Delta^k$, where 
$
 \Delta^k \defeq \{(x_1,\ldots,x_k) \in \RR^k : 0 < x_1 < \cdots < x_k\}
$ 
(see \eqref{eq:defsimplex}). 
Set $\bo \defeq \emptyset$ or set $\bo \defeq \{0\}$.
Let $y \ge 1$.
To ease notation throughout, let $\cH \defeq y\sC \cap \ZZ^k$, 
$\vbh = (h_1,\ldots,h_k)$, and $\bh = \{h_1,\ldots,h_k\}$.
Note that $0 < h_1 < \cdots < h_k \ll_{\sC} y$ for $\vbh \in \cH$.
Now set $z \defeq \e^{2\sqrt{\log y}}$.
The estimate \eqref{Jeq:sssa} is trivial if $y \ll_{k,\sC} 1$, so 
we may assume that $\cH \ne \emptyset$, and also that 
$\log\log z \ge 1$.

In view of \eqref{Jeq:defsssS} we see, upon partitioning the sum 
over $d$ and changing order of summation, that  
\begin{equation}
 \label{Jeq:lemsssapf1}
 \sum_{\vbh \in \cH} \mathfrak{S}_{\bo \cup \bh}
  =
   \sum_{\vbh \in \cH} \epsilon_{\bo \cup \bh}(1)
    +
      \sums[d \in \cD][1 < d \le z]
       \sum_{\vbh \in \cH} \epsilon_{\bo \cup \bh}(d)
      +
       \sums[d \in \cD][d > z]
        \sum_{\vbh \in \cH} \epsilon_{\bo \cup \bh}(d).
\end{equation}
Since $\epsilon_{\bo \cup \bh}(1) = 1$ we have, by 
Lemma \ref{Jlem:lip}, that 
\begin{equation}
 \label{Jeq:volH}
  \sum_{\vbh \in \cH} \epsilon_{\bo \cup \bh}(1)
   =
    y^k \vol(\sC) + O_{k,\sC}(y^{k - 1}).
\end{equation}
We show that the second and third sums make a negligible 
contribution to the right-hand side of \eqref{Jeq:lemsssapf1}.

Recall that $A_k$ conveniently stands for a sufficiently large 
(not necessarily optimal) number depending on $k$ (or more 
precisely, in this proof, on $\ocard \bo + k$), possibly a 
different number in any two occurrences.
By \eqref{Jeq:epsbndd} and the trivial bound 
$
 (\det(\bo \cup \bh),d) 
   \le 
    \sum_{c \mid d, \, c \mid \det(\bo \cup \bh)} c
$, 
we have 
\[
 \sums[d \in \cD][d > z]
  \sum_{\vbh \in \cH} |\epsilon_{\bo \cup \bh}(d)|
   \le 
    \sums[d \in \cD][d > z]
     \sum_{\vbh \in \cH} A_k^{\omega(d)}\frac{(\det(\bo \cup \bh),d)}{d^2}
      \le 
       \sums[d \in \cD][d > z]
        \frac{A_k^{\omega(d)}}{d^2}
         \sum_{c \mid d}  
          \sums[\vbh \in \cH][c \mid \det(\bo \cup \bh)] c.
\]
If $\vbh \in \cH$ then $0 < h_1 < \cdots < h_k \ll_{\sC} y$, and 
if $c \mid \det(\bo \cup \bh)$ then $c \mid \det(\{0\} \cup \bh)$, 
even in the case where $\bo$ is empty, so Lemma \ref{Jlem:dethap} 
yields, for squarefree $d$, 
\[
 \sum_{c \mid d}  
  \sums[\vbh \in \cH][c \mid \det(\bo \cup \bh)] c
   \le 
    \sum_{c \mid d} c 
     \sums[0 < h_1 < \cdots < h_k \ll_{\sC} y][c \mid \det(\{0,h_1,\ldots,h_k\})] 1
      \ll_{k,\sC}
    A_k^{\omega(d)}
     \bigg(
      y^k 
      +
       y^{k - 1}
        \sums[c \mid d, \, p \mid c \implies p \ll_{\sC} y] c
     \bigg).
\]
(If $0 < h_1 < \cdots < h_k \ll_{\sC} y$ and 
$c \mid \det(\{0,h_1,\ldots,h_k\})$ then all prime divisors of $c$ 
are $\ll_{\sC} y$;  
if $c \mid d$ then $\omega(c) \le \omega(d)$; 
if $d$ is squarefree then $\sum_{c \mid d} 1 = 2^{\omega(d)}$.)
Recalling that $\sumsstxt[\flat]$ denotes summation restricted to 
squarefree integers, we have 
\[
 \sumss[\flat][d \ge 1] 
  \frac{A_k^{\omega(d)}}{d^2}
   \sums[c \mid d, \, p \mid c \implies p \ll_{\sC} y] c
    \le 
     \sumss[\flat][c \ge 1, \, p \mid c \implies p \ll_{\sC} y] \frac{A_k^{\omega(c)}}{c}
      \sumss[\flat][b \ge 1] \frac{A_k^{\omega(b)}}{b^2}
       \ll_k
        \sumss[\flat][c \ge 1, \, p \mid c \implies p \ll_{\sC} y] \frac{A_k^{\omega(c)}}{c},
\]
as can be seen by writing $d = bc$ and changing order of 
summation, and    
\[
 \sumss[\flat][c \ge 1, \, p \mid c \implies p \ll_{\sC} y] \frac{A_k^{\omega(c)}}{c}
  \le 
   \prod_{p \ll_{\sC} y}
    \bigg(1 + \frac{A_k}{p}\bigg)
%      \le 
%       \prod_{p \ll_{\sC} y}
%        \bigg(1 + \frac{1}{p}\bigg)^{A_k}
        \ll_{k,\sC} (\log y)^{A_k}  
\]
(see \eqref{Jeq:mert}).
Combining and applying Lemma \ref{Jlem:omegabnd} (with $D = 1$), we 
see that 
{\small 
\begin{equation}
 \label{Jeq:lemsssapf2}
 \sums[d \in \cD][d > z]
  \sum_{\vbh \in \cH} |\epsilon_{\bo \cup \bh}(d)|
   \ll_{k,\sC}
    y^k\sums[d \in \cD][d > z] \frac{A_k^{\omega(d)}}{d^2}
     + y^{k - 1}(\log y)^{A_k} 
     \ll_{k,\sC}
       y^k\frac{(\log z)^{A_k}}{z} + y^{k - 1}(\log y)^{A_k}.
\end{equation}
}

\noindent 
(Recall that $\cD$ by definition contains only squarefree 
integers.)

For the sum in \eqref{Jeq:lemsssapf1} over $1 < d \le z$, let 
$(\alpha_p)_{p \not\equiv 1 \bmod 4}$ be the sequence of integers 
satisfying 
\[
 1 + \frac{\log z}{\log p}
  < 
   \alpha_p
    \le 
     2 + \frac{\log z}{\log p}
\]
for all $p \not\equiv 1 \bmod 4$.
We claim that for $d \in \cD$ with $1 < d \le z$, we have 
\begin{equation}
 \label{Jeq:dbnds}
  \textstyle 
   d^2z^{1/2}
   <
    \prod_{p \mid d} p^{\alpha_p}
     <
      \e^{c_1(\log z)^2/\log\log z} 
\end{equation}
for a suitable absolute constant $c_1 > 0$.
To see this, let $\theta_p = (\log z)/\log p$, so that 
$p^{\theta_p} = z$ and $1 + \theta_p < \alpha_p \le 2 + \theta_p$.
For $p \le z$, we have $\theta_p \ge 1$.
For $1 \le \theta_p < 2$ (i.e.\ $z^{1/2} < p \le z$), 
$\alpha_p = 3$ and $1/\theta_p > 1/2$, hence 
$
 p^{\alpha_p} 
  = 
   p^2p
    =
     p^2z^{1/\theta_p}
      >
       p^2z^{1/2}
$.
For $\theta_p \ge 2$ (i.e.\ $p \le z^{1/2}$), $\theta_p/2 \ge 1$ 
and 
$
 \alpha_p - \theta_p/2 
  = 
   \alpha_p - \theta_p + \theta_p/2
    > 1 + \theta_p/2
     \ge 
      2
$,
hence 
$
 p^{\alpha_p}
  =
   p^{\alpha_p - \theta_p/2}p^{\theta_p/2}
     =
      p^{\alpha_p - \theta_p/2}z^{1/2}
       > p^2z^{1/2}
$.
Thus, if $d$ is squarefree and $d \le z$, then 
$
 \prod_{p \mid d} p^{\alpha_p}
  > 
   \prod_{p \mid d} p^2z^{1/2}
    = 
     d^2z^{\omega(d)/2}
$.
If $d > 1$ then $\omega(d) \ge 1$.
For squarefree $d \le z$ we have  
$
 \prod_{p \mid d} p^{\alpha_p}
  \le 
   \prod_{p \mid d} (p^2z)
    = 
     d^2z^{\omega(d)}
      <
       z^{2 + \omega(d)}
$.
By the elementary bound $\omega(d) \ll (\log d)/\log\log d$ we 
have 
$
 2 + \omega(d) \le c_1(\log z)/\log\log z
$
for $d \le z$, where $c_1 > 0$ is a suitable absolute constant.
Hence 
$
 z^{2 + \omega(d)} 
  \le 
   \e^{c_1(\log z)^2/\log\log z}.
$

Let us set   
$
 \upsilon_{\bo \cup \bh}(p^{\alpha_p}) 
  \defeq \upsilon_{\bo \cup \bh}(p^{\alpha_p}; \ocard \bo + k)
$
to ease notation in what follows.
Writing
$
 \epsilon_{\bo \cup \bh}(d)
  =
   \prod_{p \mid d}\big(\upsilon_{\bo \cup \bh}(p^{\alpha_p}) + \epsilon_{\bo \cup \bh}(p) - \upsilon_{\bo \cup \bh}(p^{\alpha_p})\big)
$
and developing the product, the sum in \eqref{Jeq:lemsssapf1} over 
$1 < d \le z$ becomes 
\[
 \sums[d \in \cD][1 < d \le z] 
  \sum_{\vbh \in \cH}
   \bigg\{ 
    \prod_{p \mid d} \upsilon_{\bo \cup \bh}(p^{\alpha_p})
    + 
      \sums[bc = d][c > 1] 
       \prod_{p \mid b} \upsilon_{\bo \cup \bh}(p^{\alpha_p})
        \prod_{p \mid c} \big(\epsilon_{\bo \cup \bh}(p) - \upsilon_{\bo \cup \bh}(p^{\alpha_p})\big)
  \bigg\}.
\]
%
%****************************************************************%
%************************* START DETAIL *************************%
%****************************************************************%
%
\begin{nixnix}
We have 
\begin{align*}
 \epsilon_{\bo \cup \bh}(d)
 & =
   \prod_{p \mid d}\big(\upsilon_{\bo \cup \bh}(p^{\alpha_p}) + \epsilon_{\bo \cup \bh}(p) - \upsilon_{\bo \cup \bh}(p^{\alpha_p})\big)
 \\
 & =
     \prod_{p \mid d} \upsilon_{\bo \cup \bh}(p^{\alpha_p})
 \\
 & \hspace{30pt} + 
       \sums[bc = d][c > 1] 
        \prod_{p \mid b} \upsilon_{\bo \cup \bh}(p^{\alpha_p})
         \prod_{p \mid c} \big(\epsilon_{\bo \cup \bh}(p) - \upsilon_{\bo \cup \bh}(p^{\alpha_p})\big), 
\end{align*}
hence 
\begin{align*}
 & 
 \sums[d \in \cD][1 < d \le z] 
  \sum_{\vbh \in \cH}
   \epsilon_{\bo \cup \bh}(d)
 \\
 & \hspace{15pt} 
   =
     \sums[d \in \cD][1 < d \le z] 
      \sum_{\vbh \in \cH}
       \bigg\{ 
        \prod_{p \mid d} \upsilon_{\bo \cup \bh}(p^{\alpha_p})
 \\
 & \hspace{30pt} + 
          \sums[bc = d][c > 1] 
           \prod_{p \mid b} \upsilon_{\bo \cup \bh}(p^{\alpha_p})
            \prod_{p \mid c} \big(\epsilon_{\bo \cup \bh}(p) - \upsilon_{\bo \cup \bh}(p^{\alpha_p})\big)
      \bigg\}.
\end{align*}
\end{nixnix}
%
%****************************************************************%
%************************** END DETAIL **************************%
%****************************************************************%
%
By Lemma \ref{Jlem:cancel} (b), \eqref{Jeq:upsbnd1}, for squarefree 
$b$ and $\vbh \in \cH$, we see, upon noting that 
$\ocard \bo + k - \card(\bo \cup \bh) = 0$, that  
$
 \prod_{p \mid b} |\upsilon_{\bo \cup \bh}(p^{\alpha_p})|
%   \le 
%    \prod_{p \mid b} \frac{A_k\det(\bo \cup \bh),p)}{p^2}
    \le 
     A_k^{\omega(b)}\det(\bo \cup \bh,b)/b^2
$.
Similarly, applying Lemma \ref{Jlem:cancel} (b), 
\eqref{Jeq:upsbnd2}, we obtain, for 
squarefree $c$ and $\vbh \in \cH$, the bound 
$
 \prod_{p \mid c} |\epsilon_{\bo \cup \bh}(p) - \upsilon_{\bo \cup \bh}(p^{\alpha_p})| 
%   \le 
%    \prod_{p \mid c} A_k/p^{\alpha + \alpha \bmod 2}
    \le 
     A_k^{\omega(c)}\prod_{p \mid c} 1/p^{\alpha_p + \alpha_p \bmod 2}
$.
Furthermore, by the lower bound in \eqref{Jeq:dbnds}, we have 
$\prod_{p \mid c} 1/p^{\alpha_p} < 1/(c^2z^{1/2})$.
Combining, we see that 
{\small
\begin{align*} 
 & 
 \sums[d \in \cD][1 < d \le z] 
  \sum_{\vbh \in \cH}
   \sums[bc = d][c > 1] 
    \prod_{p \mid b} |\upsilon_{\bo \cup \bh}(p^{\alpha_p})|
     \prod_{p \mid c} |\epsilon_{\bo \cup \bh}(p) - \upsilon_{\bo \cup \bh}(p^{\alpha_p})|
 \\
 & \hspace{5pt}
   \le 
    z^{-1/2}
     \sums[d \in \cD][1 < d \le z] 
      \sum_{\vbh \in \cH}
       \sums[bc = d][c > 1] 
        \frac{A_k^{\omega(b)}(\det(\bo \cup \bh),b)}{b^2}
         \cdot 
          \frac{A_k^{\omega(c)}}{c^2}
   \le 
    z^{-1/2}
     \sums[d \in \cD][1 < d \le z] 
      \frac{A_k^{\omega(d)}}{d^2}
       \sum_{a \mid d}
        \sums[\vbh \in \cH][a \mid \det(\bo \cup \bh)] a.
\end{align*}
}

\noindent 
(For the last inequality note that in the innermost sum, 
$A_k^{\omega(b)}A_k^{\omega(c)} = A_k^{\omega(d)}$, 
$(\det(\bo \cup \bh),b) \le (\det(\bo \cup \bh),d)$, 
$\sum_{bc = d} 1 = 2^{\omega(d)}$, then use the trivial bound 
$
 (\det(\bo \cup \bh),d) 
   \le 
    \sum_{a \mid d, \, a \mid \det(\bo \cup \bh)} a
$.)
We invoke Lemma \ref{Jlem:dethap} again, this time noting that 
if $a \mid d$ and $d \le z \ll y$, then 
$y^k/a + O_k(y^{k - 1}) \ll_k y^k/a$; also,  
$\sum_{a \mid d} A_k^{\omega(a)} \le A_k^{\omega(d)}$ (recall that 
$d$ is squarefree and the convention $A_{k}$ might denote 
different constants.)
We find that 
{\small 
\begin{align*}
 \sums[d \in \cD][1 < d \le z] 
  \frac{A_k^{\omega(d)}}{d^2}
   \sum_{a \mid d}
    \sums[\vbh \in \cH][a \mid \det(\bo \cup \bh)] \hspace{-6pt} a
 & 
  \le 
   \sums[d \in \cD][1 < d \le z] 
    \frac{A_k^{\omega(d)}}{d^2}
     \sum_{a \mid d}
      \sums[0 < h_1 < \cdots < h_k \ll_{\sC} y][a \mid \det(\{0,h_1,\ldots,h_k\})] a
 & \hspace{-6pt} 
   \ll_{k,\sC}
    y^k 
    \sumss[\flat][d \ge 1]  
     \frac{A_k^{\omega(d)}}{d^2}
      \ll_{k,\sC}
       y^k,
\end{align*}
} 
so combining yields 
\begin{equation}
 \label{Jeq:penult}
  \sums[d \in \cD][1 < d \le z] 
   \sum_{\vbh \in \cH}
    \epsilon_{\bo \cup \bh}(d)
     =
      \bigg( \,
       \sums[d \in \cD][1 < d \le z] 
        \sum_{\vbh \in \cH}
         \prod_{p \mid d} \upsilon_{\bo \cup \bh}(p^{\alpha_p})
      \bigg)
       +
         O_{k,\sC}\big(y^kz^{-1/2}\big).
\end{equation}

Consider an arbitrary $d \in \cD$ with $1 < d \le z$.
We set $d_{\alpha} \defeq \prod_{p \mid d} p^{\alpha_p}$, and 
partition $\RR^k$ into cubes 
\[
 C_{d_{\alpha},\vbt} 
  \defeq 
   \{(x_1,\ldots,x_k) \in \RR^k : t_id_{\alpha} \le x_i < (t_i + 1)d_{\alpha}, i = 1,\ldots,k\},
\]
with $\vbt \defeq (t_1,\ldots,t_k)$ running over $\ZZ^k$.
Each $\vbh \in \cH$ is a point in a unique cube of this form: we 
call $\vbh$ a {\em $d_{\alpha}$-interior} point if this cube is 
entirely contained in $y\sC$, and $\vbh$ a 
{\em $d_{\alpha}$-boundary} point if this cube has a nonempty 
intersection with the boundary of $y\sC$.
We partition $\cH$ into $d_{\alpha}$-interior points and 
$d_{\alpha}$-boundary points.
As $\vbh$ runs over all $d_{\alpha}$-interior points of $\cH$, 
$h_i$ ($i = 1,\ldots,k$) runs over a pairwise disjoint union of 
complete residue systems modulo $d_{\alpha}$.
By Lemma \ref{Jlem:cancel} (a), it follows that 
\begin{equation}
 \label{Jeq:2ndlast}
  \sums[d \in \cD][1 < d \le z] 
   \sum_{\vbh \in \cH}
    \prod_{p \mid d} \upsilon_{\bo \cup \bh}(p^{\alpha_p})
  =
   \sums[d \in \cD][1 < d \le z]
    \sums[\vbh \in \cH][\text{$d_{\alpha}$-boundary}]
     \prod_{p \mid d} \upsilon_{\bo \cup \bh}(p^{\alpha_p}).
\end{equation}

For each $d \in \cD$, $1 < d \le z$ we have, by \eqref{Jeq:dbnds}, 
that $d_{\alpha} < \e^{c_1(\log z)^2/\log \log z}$.
Since $z = \e^{2\sqrt{\log y}}$, this means that 
$y/d_{\alpha} \ge y^{1 - O(1/\log\log y)}$.
From the proof of Lemma \ref{Jlem:lip} (see 
\cite[pp.\ 128--129]{LAN:94}), one can see that there are 
$\ll_{k,\sC} (y/d_{\alpha})^{k - 1}$ cubes $C_{d_{\alpha},\vbt}$ 
that have a nonempty intersection with the boundary of $y\sC$,  
and there are at most $d_{\alpha}^k$ points $\vbh$ of $\cH$ in 
each such cube.
Hence, in all, there are 
$
 \ll_{k,\sC} 
  y^{k - 1}d_{\alpha} < y^{k - 1}\e^{c_1(\log z)^2/\log \log z}
$ 
$d_{\alpha}$-boundary points $\vbh \in \cH$.
In view of this and Lemma \ref{Jlem:cancel} (b), 
\eqref{Jeq:upsbnd1}, we see that 
\begin{align}
 \label{Jeq:last}
  \begin{split}
 & 
  \sums[d \in \cD][1 < d \le z]
   \sums[\vbh \in \cH][\text{$d_{\alpha}$-boundary}]
    \prod_{p \mid d} |\upsilon_{\bo \cup \bh}(p^{\alpha_p})|
 \\
 & \hspace{30pt}
     \le 
      \sums[d \in \cD][1 < d \le z] 
       \sums[\vbh \in \cH][\text{$d_{\alpha}$-boundary}]
        \frac{A_k^{\omega(d)}}{d}
         \ll_{k,\sC}
          y^{k - 1}\e^{c_1(\log z)^2/\log\log z}(\log z)^{A_k}.
  \end{split}
\end{align}
(Here we have again used the elementary bound \eqref{Jeq:mert}.)

Combining \eqref{Jeq:penult}, \eqref{Jeq:2ndlast} and 
\eqref{Jeq:last}, we obtain 
\begin{equation}
 \label{Jeq:ult}
  \sums[d \in \cD][1 < d \le z] 
   \sum_{\vbh \in \cH}
    \epsilon_{\bo \cup \bh}(d) 
     \ll_{k,\sC}
      y^{k - 1}\e^{2c_1(\log z)^2/\log\log z} 
       + 
        y^kz^{-1/2}. 
\end{equation}
Putting \eqref{Jeq:volH}, \eqref{Jeq:lemsssapf2} and \eqref{Jeq:ult} 
into \eqref{Jeq:lemsssapf1}, and recalling that 
$z = \e^{2\sqrt{\log y}}$, we obtain \eqref{Jeq:sssa}.
\end{proof}

\end{jetsam}

\end{document}